\documentclass[10pt]{article} 
\usepackage[preprint]{tmlr}

\usepackage{mathtools,amsfonts,amssymb}
\usepackage{enumerate}
\usepackage{bm}
\usepackage{hyperref}
\usepackage{url}
\usepackage{booktabs}       
\usepackage{nicefrac}       
\usepackage{microtype}      
\usepackage{graphicx}
\usepackage{xcolor}
\usepackage[parfill]{parskip}
\usepackage{subcaption}
\usepackage{enumitem} 
\usepackage{diffcoeff}          
\usepackage[ruled]{algorithm2e}
\DontPrintSemicolon
\usepackage{natbib}

\usepackage{amsthm}
\usepackage{thmtools}           
\declaretheorem[name=Theorem]{theorem}
\newtheorem*{sketchproof}{Sketch of the Proof.}
\newtheorem{lemma}{Lemma} 
\newtheorem{corollary}{Corollary}

\graphicspath{{./plots/}}

\usepackage{amsmath,amsfonts,bm}

\def\1{\bm{1}}

\def\va{{\bm{a}}}

\def\vc{{\bm{c}}}

\def\vg{{\bm{g}}}
\def\vh{{\bm{h}}}

\def\vr{{\bm{r}}}

\def\vu{{\bm{u}}}
\def\vv{{\bm{v}}}
\def\vw{{\bm{w}}}
\def\vx{{\bm{x}}}
\def\vy{{\bm{y}}}
\def\vz{{\bm{z}}}

\def\vzero{{\bm{0}}}
\def\vone{{\bm{1}}}

\def\vbeta{{\bm{\beta}}}

\def\vdelta{{\bm{\delta}}}
\def\vepsilon{{\bm{\epsilon}}}

\def\veta{{\bm{\eta}}}
\def\vtheta{{\bm{\theta}}}

\def\vlambda{{\bm{\lambda}}}
\def\vmu{{\bm{\mu}}}

\def\vxi{{\bm{\xi}}}

\def\mA{{\bm{A}}}
\def\mB{{\bm{B}}}
\def\mC{{\bm{C}}}
\def\mD{{\bm{D}}}

\def\mH{{\bm{H}}}
\def\mI{{\bm{I}}}
\def\mJ{{\bm{J}}}
\def\mK{{\bm{K}}}
\def\mL{{\bm{L}}}
\def\mM{{\bm{M}}}

\def\mQ{{\bm{Q}}}
\def\mR{{\bm{R}}}
\def\mS{{\bm{S}}}

\def\mV{{\bm{V}}}

\def\mOmega{{\bm{\Omega}}}

\def\mSigma{{\bm{\Sigma}}}
\def\mTheta{{\bm{\Theta}}}

\DeclareMathAlphabet{\mathsfit}{\encodingdefault}{\sfdefault}{m}{sl}
\SetMathAlphabet{\mathsfit}{bold}{\encodingdefault}{\sfdefault}{bx}{n}

\def\bE{{\mathbb{E}}}

\def\bP{{\mathbb{P}}}

\def\bR{{\mathbb{R}}}

\def\sA{{\mathcal{A}}}
\def\sB{{\mathcal{B}}}
\def\sC{{\mathcal{C}}}
\def\sD{{\mathcal{D}}}
\def\sE{{\mathcal{E}}}
\def\sF{{\mathcal{F}}}
\def\sG{{\mathcal{G}}}

\def\sJ{{\mathcal{J}}}

\def\sL{{\mathcal{L}}}
\def\sM{{\mathcal{M}}}
\def\sN{{\mathcal{N}}}

\def\sP{{\mathcal{P}}}
\def\sQ{{\mathcal{Q}}}
\def\sR{{\mathcal{R}}}

\def\sZ{{\mathcal{Z}}}

\newcommand{\E}{\mathbb{E}}

\newcommand{\R}{\mathbb{R}}

\newcommand{\KL}{D_{\mathrm{KL}}}
\newcommand{\Var}{\mathrm{Var}}

\newcommand{\Cov}{\mathrm{Cov}}

\DeclareMathOperator*{\argmax}{arg\,max}
\DeclareMathOperator*{\argmin}{arg\,min}

\newcommand{\cons}{h}               
\newcommand{\w}{w}               
\newcommand{\ndata}{n}
\newcommand{\nis}{L}
\newcommand{\dparam}{p}
\newcommand{\LP}{\mathrm{LP}}
\newcommand{\IS}{\mathrm{IS}}

\newcommand{\expit}{\operatorname{expit}}

\newcommand{\vech}{\operatorname{vech}}
\newcommand{\dtv}{d_\mathrm{TV}}
\newcommand{\EL}{\operatorname{EL}}
\newcommand{\eigen}{\operatorname{eigen}}
\newcommand{\rank}{\operatorname{rank}}
\newcommand{\tr}{\operatorname{tr}}
\newcommand{\ind}{\mathbf{1}}
\newcommand{\nsites}{D}

\newcommand{\pr}{\text{Pr}}   

\newcommand{\high}{\mathrm{high}}
\newcommand{\low}{\mathrm{low}}

\title{Expectation-propagation for Bayesian empirical likelihood inference}

\author{\name Kenyon Ng \email kenyon.ng@gmail.com \\
  \addr College of Computing and Data Science, Nanyang Technological University
  \AND
  \name Weichang Yu \email weichang.yu@unimelb.edu.au \\
  \addr School of Mathematics and Statistics, The University of Melbourne \AND
  \name Howard D. Bondell \email howard.bondell@unimelb.edu.au \\
  \addr School of Mathematics and Statistics, The University of Melbourne }

\def\month{MM}  
\def\year{YYYY} 
\def\openreview{\url{https://openreview.net/forum?id=XXXX}} 

\begin{document}

\maketitle

\begin{abstract}
  Bayesian inference typically relies on specifying a parametric model that
  approximates the data-generating process. However, misspecified models can
  yield poor convergence rates and unreliable posterior calibration. Bayesian
  empirical likelihood offers a semi-parametric alternative by replacing the
  parametric likelihood with a profile empirical likelihood defined through
  moment constraints, thereby avoiding explicit distributional assumptions.
  Despite these advantages, Bayesian empirical likelihood faces substantial
  computational challenges, including posterior sampling difficulties due to the
  non-convex posterior support. This paper introduces an expectation-propagation
  approach for Bayesian empirical-likelihood posterior approximation, balancing
  computational cost and accuracy without altering the target posterior via
  adjustments such as pseudo-observations. Empirically, we show that our
  approach can achieve a superior cost–accuracy trade-off relative to existing
  methods, including Hamiltonian Monte Carlo and variational Bayes.
  Theoretically, we show that the approximation and the Bayesian
  empirical-likelihood posterior are asymptotically equivalent.
\end{abstract}


\section{Introduction}
\label{sec:introduction}

Consider a set of observations $\sD_n = \{\vz_i\}_{i=1}^n$ assumed to be drawn
independently and identically distributed (i.i.d.) from a distribution
$F_{0}(\vz) = f(\vz; \vtheta_{0})$ indexed by
$\vtheta_{0} \in \mTheta \subset \R^{\dparam}$. Bayesian inference provides a
principled framework to learn about this parameter by updating our beliefs based
on the data. Given prior information encoded in a prior distribution
$p(\vtheta)$, our belief about $\vtheta$ is updated via Bayes’ rule:
$p(\vtheta \mid \sD_{n}) \propto p(\vtheta) \prod^{\ndata}_{i=1} f(\vz_{i}; \vtheta)$.
The posterior distribution $p(\vtheta \mid \sD_{n})$ is the central object of
interest in most Bayesian analyses. When the data distribution is correctly
specified, the posterior is consistent \citep{doob49application} and
asymptotically normal \citep{lecam53asymptotic,kleijn12bernsteinvonmises} under
a range of regularity conditions. However, the choice of an appropriate
distributional family is often a difficult task. On one hand, simple data
distributions impose strong assumptions which could result in invalid inference
\citep{bissiri16general}. On the other hand, elaborate distributions (e.g.,
mixture models) with flexible assumptions are computationally challenging to
fit.



There are several non‑parametric or semi‑parametric approaches to tackle model
misspecification~\citep{ferguson73bayesian,blackwell73ferguson,antoniak74mixtures,lee25conditional,grunwald12safe,bissiri16general}.
Several of these approaches require careful tuning of concentration parameters
and/or learning rate to achieve theoretical properties for uncertainty
quantification. However, the tuning procedure can be computationally costly.

Bayesian empirical likelihood \citep{lazar03bayesian} is another semi-parametric
approach to Bayesian inference. Rather than employing a parametric likelihood in
the posterior, it replaces the likelihood with a \emph{profile empirical
  likelihood} \citep{owen88empirical}, which is defined as a solution to a
constrained optimization problem subject to
$\E_{\vz \sim F_{0}}[h(\vz, \vtheta)] = \vzero$ for some constraint function
$h$. Crucially, under this framework, there is neither a parametric family of
distributions to be specified nor hyperparameters (e.g., concentration
parameters or learning rates) to be tuned. The framework of
\citet{lazar03bayesian}, however, is not the only way to incorporate empirical
likelihood into Bayesian workflows. For instance, the exponentially tilted
empirical likelihood \citep{schennach05bayesian,yiu20inference} has been adapted
to moment-misspecification settings \citep{chib18bayesian}. Penalised empirical
likelihoods have been proposed to address degeneracy when the dimension of the
constraint function $h$ is greater than the sample size~$n$
\citep{chang25bayesian}, and profile empirical likelihood has been applied
within approximate Bayesian computation \citep{mengersen13bayesian}. Other
related Bayesian moment-based inference methods have also been introduced
\citep{bornn19moment,florens21gaussian}.


Despite its appeal, Bayesian empirical likelihood is often challenging to
implement due to the profile empirical likelihood. In practice, the solution to
the underlying constrained optimisation problem must be obtained numerically and
may not even exist, resulting in a posterior with highly non-convex support.
Consequently, efficient computation of the Bayesian empirical likelihood
posterior $p_{\EL}$ remains an open problem. Hamiltonian Monte Carlo
\citep[HMC,][]{duane87hybrid} is a common choice to approximate $p_{\EL}$ when
the moment condition function $h$ is smooth, and it can generally provide
accurate approximations \citep{chaudhuri17hamiltonian}. However, HMC may be
infeasible due to computational constraints or working with non-smooth $h$. When
the quantity of interest is the posterior moments, a direct estimate can be
obtained without computing the problematic posterior \citep{vexler14posterior}.
For applications requiring fast approximation of the posterior,
\citet{yu24variational} proposed approximating~$p_{\EL}$ with the Gaussian
minimizer of the Kullback–Leibler (KL) divergence
$\argmin_{q \in \sQ} \bE_{\vtheta \sim q}[\log q(\vtheta) - \log p_{\EL}(\vtheta \mid \sD_{n})]$
within a class of Gaussian distribution $\sQ$, where $p_{\EL}$ denotes the
posterior of Bayesian empirical likelihood. Here, a mismatch between the support
of $p_{\EL}$ and $q$ results in an infinite KL divergence. To address
the mismatch in support, \citet{yu24variational} replace $p_{\EL}$ with a
posterior based on the adjusted empirical likelihood function
\citep{chen08adjusted} in the KL objective. The replacement posterior has a
support matching the Gaussian distribution, and thus the KL divergence is
finite. However, for a finite sample size $n$, the resultant approximate
posterior is sensitive to the adjustment level in the approximate empirical
likelihood, and the replacement posterior adds an extra layer of approximation.

\subsection{Contribution}
\label{sec:contribution}

We study the suitability of expectation‑propagation
\citep{opper00gaussian,minka01expectation}, an algorithm that enjoys
considerable empirical success for approximating Bayesian posteriors
\citep{hernandez-lobato16scalable,hasenclever17distributed,hall20fast}, for
approximating the Bayesian empirical likelihood posterior. More crucially, the
algorithm remains appropriate even when the supports of the approximate and the
target posteriors do not match. We propose an implementation for Bayesian
empirical likelihood and provide practical guidance on algorithmic stability. In
terms of theory, we prove that our proposed expectation-propagation posterior is
asymptotically equivalent to $p_{\EL}$. While related results appear in
\cite{dehaene18expectation}, we emphasise that their results assume a fixed
posterior support, whereas a key pathology of many Bayesian empirical likelihood
posteriors are their data-truncated posterior supports. In fact, we provide
non-restrictive sufficient conditions to establish the asymptotic equivalence
with the exact posterior. Moreover, intermediate results pertaining to the
behaviour of the Bayesian empirical likelihood posterior support are novel
additions to the empirical likelihood literature. Through an extensive set of
experiments, we show that the algorithm achieves a better cost–accuracy
trade‑off than HMC \citep{chaudhuri17hamiltonian} and is free from any
adjustment parameters that can impair approximation quality in small‑$n$
settings \citep{yu24variational}.

This paper begins by reviewing the necessary background on Bayesian empirical
likelihood and expectation-propagation in Section~\ref{sec:background}. We then
present our algorithm in the context of Bayesian empirical likelihood inference
in Section~\ref{sec:epel}, establish its asymptotic properties in
Section~\ref{sec:theory}, and demonstrate its performance on a variety of
modelling tasks in Section~\ref{sec:experiments}. We conclude in
Section~\ref{sec:discussion}.

\section{Background}
\label{sec:background}
We begin this section with a review of Bayesian empirical likelihood, with
particular emphasis on the computational challenges of implementing it. This is
followed by a review of expectation-propagation
\citep{opper00gaussian,minka01expectation}, on which our proposed method is
based.

\subsection{Bayesian empirical likelihood}
\label{sec:bayes-el}

Bayesian empirical likelihood \citep{lazar03bayesian}, building on empirical
likelihood \citep{owen88empirical}, offers a semi‑parametric alternative that
avoids restrictive distributional assumptions. Rather than specifying a full
data‑generating distribution, practitioners specify a moment constraint function
$\cons: \sZ \times \mTheta \to \R^K$ and assume that observations and
$\vtheta$ satisfy $\bE_{\vz \sim F_0}[\cons(\vz, \vtheta)] = \vzero$. The
likelihood in the posterior is replaced by the profile empirical likelihood,
$\EL(\vtheta)$, defined as the constrained maximum
\begin{equation}
  \label{eq:pel}
  \EL(\vtheta)
  = \max_{\vw} \prod_{i=1}^{\ndata} \w_{i}, \quad \text{subject to}
  \quad
  \sum_{i=1}^{\ndata} \w_{i} \cons(\vz_{i}, \vtheta) = \vzero,\quad
  \sum_{i=1}^{\ndata} \w_{i} = 1,
\end{equation}
over discrete distributions
$\vw = \{\w_{i}\}^{\ndata}_{i=1} \in [0, 1]^{\ndata}$ on the observations
$\{\vz_{i}\}^{\ndata}_{i=1}$. The primary aim of this work is to develop a
computationally efficient method to obtain the Bayesian empirical likelihood
posterior $p_{\EL}(\vtheta \mid \sD_{n}) \propto p(\vtheta) \EL(\vtheta)$, where
$\EL(\vtheta)$ depends on the observations $\sD_{n}$.

This constrained maximization problem in \eqref{eq:pel} can be solved by
introducing a Lagrange multiplier $\vlambda$ with solution
$w_{i} = \ndata^{-1} (1 + \vlambda_{\EL}^{\top} \vh_{i})^{-1}$, where
$\vlambda_{\EL}$ is the root of
\begin{equation}
  \label{eq:lambda-equation}
  \sum_{i=1}^{\ndata} \frac{\vh_{i}}{1 + \vlambda_{\EL}^{\top} \vh_{i}} = \vzero.
\end{equation}
Here, we write $\vh_{i} = \cons(\vz_{i}, \vtheta)$ for brevity. A unique solution of $\vlambda_{\EL}$ exists if and only if $\vzero$ is in the
convex hull
$\sC_{\vtheta} = \{ \xi \in \R^K : \xi = \sum^{\ndata}_{i=1} \w_{i} \vh_{i}, \sum^{\ndata}_{i=1} w_{i} = 1, w_{i} > 0 \}$
\citep{owen88empirical}. Therefore, for $\vtheta$ where $\vlambda_{\EL}$ does
not exist, it is customary to set $\EL(\vtheta) = 0$, and the resultant
posterior is
\begin{equation}
  \label{eq:posterior}
  p_{\EL}(\vtheta \mid \sD_{n}) \propto
  \begin{cases}
    p(\vtheta) \prod_{i=1}^{\ndata} \w_i(\vtheta),
      & \vtheta \in \mTheta_B, \\
    0, & \vtheta \notin \mTheta_B,
  \end{cases}
\end{equation}
where $\mTheta_{B} = \{\vtheta : \vzero \in \sC_{\vtheta}\}$. This
data-dependent posterior support is often highly non-convex, particularly for
small $\ndata$, which makes posterior computation difficult because evaluations
of $\EL(\vtheta)$ require repeated constrained optimisations and will fail
outside $\mTheta_B$. We outline our approach to address this difficulty in
Section~\ref{sec:epel}.

\subsection{Expectation-Propagation}
\label{sec:ep-review}

Expectation‑propagation \citep{opper00gaussian,minka01expectation} is an
algorithm for approximating Bayesian posteriors. It requires the target
factorizes (up to a constant) and that the approximating distribution has the
same form. Our posterior satisfies this requirement. For notational convenience, let
$w_{0}(\vtheta) = p(\vtheta)$ so that
$p_{\EL}(\vtheta \mid \sD_{n}) \propto \prod_{i=0}^{\ndata} w_{i}(\vtheta)$, and
refer to $q(\vtheta) = \prod_{i=0}^{n} q_{i}(\vtheta)$ as the approximating
distribution. In EP terminology, $q_{i}$ is a site approximation, and both
$w_{i}$ and $q_{i}$ are collectively known as \emph{sites}.

The KL divergence is unsuitable as a variational loss when the variational
support strictly contains the posterior support. A straightforward alternative
is the reverse KL divergence:
\begin{equation}
  \label{eq:reverse-kl}
  \KL(p_{\EL}(\vtheta \mid \sD_{n}) \| q(\vtheta) ) =
  \bE_{\vtheta \sim p_{\EL}(\vtheta \mid \sD_{n})}[ \log p_{\EL}(\vtheta \mid \sD_{n}) - \log q(\vtheta)],
\end{equation}
which is finite for any $q$ with larger support than $p_{\EL}$. However, it is
typically intractable to optimize directly: common variational tools such as the
log‑derivative or reparameterization tricks are not directly applicable
\citep{mohamed20monte}.

Expectation‑propagation roughly minimizes \eqref{eq:reverse-kl} by iteratively
updating each $q_{i}$ as follows:
\begin{enumerate}
  \item Compute the \emph{cavity} distribution
        $q_{-i}(\vtheta) = \prod_{j \neq i} q_{j}(\vtheta)$;
  \item Form the \emph{tilted} distribution
        $q_{\backslash i}(\vtheta) \propto q_{-i}(\vtheta) w_{i}(\vtheta)$ and
        compute the KL-projection
        $\widetilde{q}_{i} = \argmin_{q \in \sQ} \KL(q_{\backslash i}(\vtheta) \| q(\vtheta))$; \label{item:ep-2}
  \item Update the site $q_{i}(\vtheta) \propto \widetilde q_{i}(\vtheta) / q_{-i}(\vtheta)$.
\end{enumerate}
Expectation‑propagation sidesteps direct optimization of \eqref{eq:reverse-kl}.
Instead, it computes the KL-projection of a `localized' posterior, namely the
tilted distribution $q_{\backslash i}(\vtheta)$. Despite its apparent aim, the
solution of expectation-propagation generally does not coincide with the
minimizer of \eqref{eq:reverse-kl}. Rather, under certain conditions, it behaves
similarly to the Laplace approximation and we
formalise this insight in Section~\ref{sec:theory}.

While the above description of expectation-propagation does not impose
restrictions on the approximating distribution $q$, in practice this is typically
chosen to come from an exponential-family. Then, the multiplication and division
of the distribution is simply the addition and subtraction of its corresponding
natural parameters respectively, and KL-projection can be done by matching
moments of sufficient statistics between $\widetilde{q}$ and $q_{\backslash i}$.
In this work, we use a Gaussian approximating distribution $q$ parameterized by
its natural parameters $\veta$: the linear shift $\vr$ and precision $\mQ$,
which correspond to the mean $\vmu = \mQ^{-1} \vr$ and covariance
$\mSigma = \mQ^{-1}$. Then, each site approximation $q_{i}$ is also Gaussian
with natural parameters $\vr_{i}$ and $\mQ_{i}$, and the natural parameters of
the global approximation are the sums: $\vr = \sum_{i=0}^{n} \vr_{i}$ and
$\mQ = \sum_{i=0}^{n} \mQ_{i}$. The KL projection in Step~\ref{item:ep-2} can be
solved by matching the first and second moments of $\widetilde{q}$ and
$q_{\backslash i}$, and this step will be discussed in detail in
Section~\ref{sec:epel}.

\section{Expectation-Propagation for Bayesian Empirical Likelihood}
\label{sec:epel}
Our proposal to address the challenge of implementing Bayesian empirical
likelihood is to approximate $p_{\EL}$ using expectation-propagation
\citep{opper00gaussian,minka01expectation}, which we describe in detail in this
section. Throughout the rest of this paper, we refer to this algorithm as
Expectation-Propagation for Bayesian Empirical Likelihood (EPEL). The majority
of the discussion here is on the local KL‑projection (Step~\ref{item:ep-2}),
which is the key step of expectation-propagation \citep{vehtari20expectation}.
In our setting, the KL-projection step is accomplished by matching the means and
covariances of $\widetilde{q}$ and $q_{\backslash i}$, i.e., setting the linear
shift $\widetilde{\vr}_{i}$ and precision $\widetilde{\mQ}_{i}$ of
$\widetilde{q}_{i}$ to
$(\Cov_{\vtheta \sim q_{\backslash i}}[\vtheta])^{-1} \E_{\vtheta \sim q_{\backslash i}}[\vtheta]$
and $(\Cov_{\vtheta \sim q_{\backslash i}}[\vtheta])^{-1}$ respectively. As the
tilted distribution $q_{\backslash i}$ is known only up to a normalizing
constant, we propose approximating $q_{\backslash i}$ with importance sampling.

\subsection{Approximating moments of tilted distribution}
The site updates involve the moments of an intractable tilted distribution. To
circumvent the intractability, we use importance sampling:
\begin{equation}
  \label{eq:is-tilted}
  \widetilde{\vmu}_{\IS} =
  \frac{\sum^{\nis}_{l=1} \xi^{l} \vtheta^{l}}{\sum^{\nis}_{l=1} \xi^{l}},
  \quad
  \widetilde{\mSigma}_{\IS} =
  \frac{\sum^{\nis}_{l=1} \xi^{l} (\vtheta^{l} - \widetilde{\vmu}_{\IS})(\vtheta^{l} - \widetilde{\vmu}_{\IS})^{\top}}{\sum^{\nis}_{l=1} \xi^{l}},
\end{equation}
where $\vtheta^{1}, \ldots, \vtheta^{\nis}$ are drawn from a proposal
distribution $s$, and
$\xi^{l} = \frac{q_{-i}(\vtheta^{l}) \w_{i}(\vtheta^{l})}{s(\vtheta^{l})}$
are the importance weights. The natural parameters are computed with
$\widetilde{\mQ}_{\IS} = \widetilde{\mSigma}_{\IS}^{-1}$ and
$\widetilde{\vr}_{\IS} = \widetilde{\mQ}_{\IS} \widetilde{\vmu}_{\IS}$. The
precision $\widetilde{\mQ}_{\IS}$ obtained this way is generally biased, but we
find the bias has negligible effect on the EPEL convergence. For the proposal
distribution, we suggest the Laplace approximation of $q_{\backslash i}$:
\begin{equation}
  \label{eq:laplace-tilted}
  \widetilde{\vmu}_{\LP} = \argmax_{\vtheta} \log q_{\backslash i}(\vtheta),
  \quad \widetilde{\mQ}_{\LP} = -\nabla_{\vtheta} \nabla_{\vtheta} \log q_{\backslash i}(\vtheta) \vert_{\vtheta = \widetilde{\vmu}_{\LP}}.
\end{equation}
This approximation requires $\log q_{\backslash i}$ to be twice differentiable,
which holds under mild conditions (Theorem~\ref{thm:continuous-diff}).
Operationally, we compute $\widetilde{\vmu}_{\LP}$ using a second-order Newton’s
method initialized at the cavity mean. This optimizer requires the gradient
$ \nabla_{\vtheta} \log q_{\backslash i}(\vtheta) = \nabla_{\vtheta} \log q_{-i}(\vtheta) + \nabla_{\vtheta} \log \w_{i}(\vtheta)$,
where
\begin{equation*}
  \nabla_{\vtheta} \log \w_{i}(\vtheta)
  = \frac{\vh_{i}^{\top} \diffp{\vlambda}{\vtheta}
    + \vlambda^{\top} \diffp{\vh_{i}}{\vtheta}}{1 + \vlambda^{\top} \vh_{i}}.
\end{equation*}
The Jacobian $\diffp{\vlambda}{\vtheta}$ can be derived by applying the implicit
function theorem on \eqref{eq:lambda-equation}. This yields the linear system
\begin{equation}
  \label{eq:dlambda}
  \left[\sum_{i=1}^{\ndata} \w_{i}^{2} \vh_{i} \vh_{i}^{\top}\right]
  \diffp{\vlambda}{\vtheta}
  = \sum_{i=1}^{\ndata} \w_{i} \left[ \ndata^{-1} \bm{I}
    - \w_{i} \vh_{i} \vlambda^{\top}\right] \diffp{\vh_{i}}{\vtheta},
\end{equation}
where $\bm{I}$ is the identity matrix. This gradient is finite if
$\sum_{i=1}^{\ndata} \w_{i}^{2} \vh_{i} \vh_{i}^{\top}$ is invertible, i.e.,
$\{\vh_{1}, \ldots, \vh_{n}\}$ spans $\R^{K}$. In our implementation, we use
automatic differentiation \citep{baydin18automatic} to compute both the gradient
and Hessian of $\log q_{\backslash i}(\vtheta)$.

In practice, the computational performance of the expectation propagation
algorithm may be sensitive to the initializations of the site approximations.
Following Theorem \ref{thm:stable-region}, we initialize our sites such that
they are roughly within a \textit{stable region} of the Laplace approximate
posterior. To do so, we use \eqref{eq:laplace-tilted} as a deterministic warm-up
for EPEL to obtain good initialisations: during the warm-up iterations, each
site update uses the Laplace mean and precision of its tilted distribution
directly in the moment-matching step. If this Laplace approximation fails to
converge or does not yield a positive-definite precision matrix, we fall back to
importance sampling. After the warm-up iterations, the actual expectation
propagation iterations are executed, where importance sampling in
\eqref{eq:is-tilted} is used for the KL projection and the Laplace approximation
serves as the proposal distribution.


\subsection{Implementation}
\label{sec:implementation}
With all the tools in place, we now describe EPEL. There are two variations
relative to the vanilla implementation in Section~\ref{sec:ep-review}. First, we
use a damping factor $\alpha \in (0, 1)$ to reduce the update size of the site
parameters, i.e.,
$\Delta \veta_{i} = \alpha \cdot (\widetilde{\veta}_{i} - \veta^{t})$ where
$\widetilde{\veta}_{i}$ is the natural parameter of the KL-projection of $q_{\backslash i}$,
and $\veta^{t}$ the global approximation at the $t$-th iteration, as suggested
in \citet{vehtari20expectation}. We find, both empirically and also from
the analysis in Theorem~\ref{thm:stable-region}, that the updates in the early
iterations tend to be noisy. Therefore, it is beneficial to restrict the update
size to avoid divergence. Second, to enable algorithmic parallelization, we
update the global parameters $\veta$ only after all sites have been updated at
least once \citep{dehaene18expectation}. The pseudocode of EPEL is presented in
Algorithm~\ref{alg:exp-prop}, and further implementation details are given in
Supplementary Material, Section~\ref{app:sec:epel-practical-notes}.

\begin{algorithm}[t]
  \caption{Expectation-propagation for Bayesian empirical likelihood}
  \label{alg:exp-prop}
  \KwSty{Initialize} $t = 0$, $\veta^{0} = \sum_i \veta_i^0$ for all sites
  $i = 0,\ldots,\ndata$, and damping factor $\alpha < 1$\;

  \Repeat{
    $\Delta\veta_i$ is sufficiently small for all $i = 0,\ldots,\ndata$
  }{
    \For{each site $i = 0,\ldots,\ndata$}{
      $q_{-i}(\vtheta; \veta_{-i}) \gets$ cavity distribution
      $\sum_{j \neq i} \veta_j^t$\;

      $\widetilde{\veta}_i \gets$ KL-projection of
      $q_{\backslash i}(\vtheta) \propto
      q_{-i}(\vtheta; \veta_{-i}) w_i(\vtheta)$ using either
      \eqref{eq:is-tilted} or \eqref{eq:laplace-tilted}\;

      $\Delta\veta_i \gets \alpha \cdot
      (\widetilde{\veta}_i - \veta^t)$\;
    }
    $\veta^{t+1} \gets \veta^t + \sum_i \Delta \veta_i$\;
    $\veta_i^{t+1} \gets \veta_i^t + \Delta \veta_i$ and
    $t \leftarrow t + 1$\;
  }
  \KwSty{Output} $\veta^{t+1}$\;
\end{algorithm}

\subsection{Post-processing EPEL solution for skewed-posterior estimation}
\label{sec:skew-post-processing}
The accuracy of the EPEL approximation can further improve from a post-processing step that is
particularly useful for skewed Bayesian empirical likelihood posteriors. Following the skew-symmetric
construction of \citet{pozza26skewsymmetric}, we perturb the fitted Gaussian
EPEL approximation without refitting the EP sites. Let $\bar q_{\vtheta_c}$
denote the Gaussian EPEL approximation centred at $\vtheta_c$, and let
$2\vtheta_c-\vtheta$ be the reflection of $\vtheta$ about this centre. We define
the skewing factor
\begin{equation*}
  a_{\vtheta_c}(\vtheta)
  = \frac{p_{\EL}(\vtheta \mid \sD_{n})}{p_{\EL}(\vtheta \mid \sD_{n}) + p_{\EL}(2\vtheta_c - \vtheta \mid \sD_{n})}.
\end{equation*}
The skew-corrected approximation is then
\begin{equation*}
  q_{\vtheta_c}(\vtheta)
  =
  2 \bar q_{\vtheta_c}(\vtheta) a_{\vtheta_c}(\vtheta).
\end{equation*}
For points outside the posterior support $\mTheta_B$, we define
$\log p_{\EL}(\vtheta \mid \sD_n)=-\infty$. Thus, the expression for
$a_{\vtheta_c}(\vtheta)$ remains valid whenever at least one of $\vtheta$ and
$2\vtheta_c-\vtheta$ lies in $\mTheta_B$. If both points lie outside
$\mTheta_B$, we follow \citet{pozza26skewsymmetric} and set
$a_{\vtheta_c}(\vtheta)=1/2$.


\section{Asymptotic behaviour of EPEL}
\label{sec:theory}
In this section, we establish the asymptotic equivalence of EPEL and $p_{\EL}$
through three main contributions. First, we prove that one cycle of the EPEL
update is asymptotically equivalent to a Newton–Raphson step
(Theorem~\ref{thm:asymptotically-newton-raphson}). Second, we show that the
first and second moments of the EPEL posterior are asymptotically equivalent to
those of the Laplace approximation of $p_{\EL}$
(Theorem~\ref{thm:stable-region}). Third, we derive a Bernstein–von–Mises
theorem demonstrating that the EPEL posterior and $p_{\EL}$ converge to the same
normal distribution (Theorem~\ref{thm:ep-bvm}). All results are stated in the
limit $n \to \infty$ unless otherwise specified.


Our proof requires a non-trivial extension of the techniques in
\citet{dehaene18expectation} and \citet{yu24variational}. In particular, the
support of the Bayesian empirical likelihood posterior is highly non-smooth due
to its data-dependent support, whereas the assumptions in
\citet{dehaene18expectation} implicitly require the smoothness of the target
posterior and its support to be fixed for all $n$. This difficulty motivates the
central contribution of our theoretical work, i.e., we establish sufficient
conditions for the Bayesian empirical-likelihood posterior to be smooth throughout the
feasible parameter space (Theorem~\ref{thm:continuous-diff}). This smoothness then validates the adaptation of proof techniques in \citet{dehaene18expectation} to the Bayesian empirical likelihood paradigm. To the best of our
knowledge, this smoothness property has not been discussed in any preceding
works on Bayesian empirical likelihood.



\subsection{Notations and assumptions}
The true parameter is $\vtheta^{\star}$, and $\bP^{\star}$ denotes the
probability measure under $F_{0}$. The negative logarithm of each target site is
$\phi_{i}(\cdot) = - \log w_{i}(\cdot)$, $i = 1, \ldots, n$, with
$w_{0}(\vtheta)$ denoting the prior $p(\vtheta)$. One‑ and zero‑vectors are
$\vone_{K}$ and $\vzero_{K}$. Subscripts $i$ index sites; superscripts $t$ index
iterations. For theorems \ref{thm:continuous-diff} to \ref{thm:stable-region}, we require some or all of the below stated assumptions:
\begin{enumerate}[label=(\Roman*)]
  \item The parameter space $\mTheta$ is bounded and
        $\vtheta^\star$ is a unique interior point. \label{item:theta-bound}
  \item For countable $h(\sZ,\vtheta)= \{h(\vz, \vtheta): \vz \in \sZ \}$: for
        each $\vtheta \in \mTheta$, there exist $L \geq K + 1$ non‑zero vectors
        $\sP_{\vtheta} = \{ g_{1,\vtheta}, \ldots, g_{L,\vtheta} \}$ whose
        convex hull contains $\vzero_K$ in its interior; and
        $\min_{a \in [L]} \bP^\star \{ h(\vz,\vtheta) = g_{a,\vtheta} \} > 0$.
        For uncountable $h(\sZ,\vtheta)$: for each $\vtheta \in \mTheta$, there
        exist $L \geq K + 1$ closed, connected sets
        $\sG_{1,\vtheta}, \ldots, \sG_{L,\vtheta}$ such that any
        $\{ g_a \in \sG_{a,\vtheta} \}_{a=1}^{L}$ have a convex hull with
        $\vzero_K$ in its interior and $\vzero_K \notin \sG_{a,\vtheta}$; the
        induced density of $h(\vz,\vtheta)$ is strictly positive on
        $\bigcup_{a=1}^{L} \sG_{a,\vtheta}$; and
        $\min_{a \in [L]} \bP^\star \{ h(\vz,\vtheta) \in \sG_{a,\vtheta} \} > 0$.
        \label{item:h-location}
  \item The first, second, and third derivatives of $h (\vz, \vtheta)$ with
        respect to $\vtheta$ are continuously differentiable on $\mTheta$ for
        all $\vz \in \sZ$.
        \label{item:h-smoothness}
  \item The prior $p(\vtheta)$ is positive on a neighbourhood of
        $\vtheta^\star$. Moreover, there exists $M_p > 0$ such that, for all
        $(a, b, c, d) \in \{ 1, \ldots, p \}^4$,
        \begin{equation*}
          \left| \diffp{\log p(\vtheta)}{\theta_a, \theta_b, \theta_c, \theta_d} \right|
          \le M_p \quad \text{for all } \vtheta \in \mTheta .
        \end{equation*} \label{item:prior-smoothness}
\end{enumerate}
Assumption~\ref{item:theta-bound} is a standard condition for
Bernstein-von-Mises type of results in empirical likelihood and mirrors many of
the previous works, e.g.,
\citep{chernozhukov03mcmc,chib18bayesian,zhao20bayesian,yiu20inference}.
Assumption~\ref{item:h-location} is a restriction on $h$. Intuitively, it
requires that, for each $\vtheta \in \mTheta$, there is a non-zero probability
of $\sC_{\vtheta}$ containing $\vzero_{K}$ in its interior; more details in
Supplementary Material, Section~\ref{app:sec:assump2-comments}. This assumption
also implies $\operatorname{span}(\sP_{\vtheta}) = \R^K$ and
$\operatorname{span}( \{ g_a \in \sG_{a,\vtheta} \}_{a=1}^{L} ) = \R^K$, since a
$\sP_{\vtheta}$ that spans a $K-1$ or lower dimension subspace of $\R^{K}$ can
never contain $\vzero_{K}$ in its interior. Assumptions~\ref{item:h-smoothness}
and~\ref{item:prior-smoothness} imply that $h (\vz, \vtheta)$ is differentiable
with respect to $\vtheta$ up to the fourth-order in the domain $\mTheta$ for all
$\vz \in \sZ$. Fourth-order differentiability is needed for analysing the
behaviour of EPEL, as we need the quadratic expansion of the Hessian of
$\phi^{(2)}(\vtheta)$.

\subsection{Asymptotic smoothness of BayesEL posterior}
Our first major result pertains to the asymptotic smoothness
of the empirical likelihood function in the entire parameter space
$\mTheta$. In fact, the result follows from a theoretical analysis of the behaviour of the posterior support. While it is widely known in the empirical likelihood literature that this support expands with $n$, here we establish sufficient conditions to guarantee that the posterior is smooth throughout a compact $\mTheta$ for a sufficiently large $n$. To the best of our knowledge, no similar analysis has been presented in preceding empirical likelihood literature.
\begin{theorem}
  \label{thm:continuous-diff}
  Assume \ref{item:theta-bound} to \ref{item:h-smoothness} hold. Then, for a sufficiently large $n$,
  \begin{equation*}
    \left| \diffp{\phi_i (\vtheta; \sD_{n})}{\theta_a, \theta_b, \theta_c, \theta_d} \right| < \infty,
    \quad \emph{for all} \
    (a, b, c, d) \in \{1, \dots, p \}^4, \
    i = 1, \ldots, n, \
    \vtheta \in \mTheta.
  \end{equation*}
\end{theorem}
The proof and details are in Supplementary Material,
Section~\ref{app:sec:smoothness-el}. Note that a related result first appeared
in \cite{owen90empirical} which says that the true value is excluded from the
empirical likelihood support finitely often. The above theorem builds on
Lemma~\ref{app:thm:single-point-supported} in the Supplementary Material which
says that any point in the parameter space is excluded from the support finitely
often.

\subsection{Asymptotic equivalence to Newton-Raphson updates}
We prove tbat a cycle of EPEL is asymptotically equivalent to performing a Newton-Raphson
update on the negative log-posterior
$\psi(\cdot) = \sum_{i=0}^{n} \phi_{i}(\cdot)$. More concretely, with the global
approximation $\vr^{t}, \mQ^{t}$ and $\vmu^{t} = (\mQ^{t})^{-1} \vr^{t}$, the
$(t+1)$‑th iterate is asymptotically equivalent to
\begin{equation*}
  \vmu^{t+1} =  \vmu^{t} - \{ \psi^{(2)}(\vmu^{t}) \}^{-1} \psi^{(1)}(\vmu^{t}),
\end{equation*}
with $\mQ^{t+1} \approx \psi^{(2)}(\vmu^{t})$ and
$\vr^{t+1} \approx \mQ^{t+1} \vmu^{t} - \psi^{(1)}(\vmu^{t})$, where it is
evident that the right-hand side is the Newton-Raphson update with the current
state equal to $\vmu^{t}$.

\begin{theorem}
  \label{thm:asymptotically-newton-raphson}
  Assume~\ref{item:theta-bound} to~\ref{item:prior-smoothness} hold. Consider the
  EPEL Gaussian approximation at iteration~$t$, $\{\vr^{t}_i \}_{i=0}^n$ and
  $\{\mQ_i^{t} \}_{i=0}^n$ for the linear-shift and precision of the site
  approximations. The global approximation mean
  $\vmu^{t} = \{ \sum_i \mQ_i^{t} \}^{-1} \sum_i \vr^{t}_i$ is a fixed vector,
  and the global precision and linear-shift are
  $\mQ^{t} = \sum_{i=0}^n \mQ_i^{t}$ and $\vr^{t} = \sum_{i=0}^n \vr_i^{t}$
  respectively. Moreover, assume that (i)
  $\min_i \min \eigen ( \sum_{j \neq i} \mQ_{j}^{t}) = pn + O_{p}(1)$ with some
  positive constant $p$; (ii) the means of the tilted distributions are not too
  far apart, i.e.,
  $\sum_i \lVert \vr^{t}_i - \mQ_i^{t} \vmu^{t} + \phi_i^{(1)} (\vmu^{t}) \rVert = O_p(n)$.
  Then, the new global linear shift and new global precision after one EP cycle
  has the asymptotic behaviour:
  \begin{equation*}
    \left \lVert \vr^{t+1} +  \psi^{(1)} (\vmu^{t}) - \mQ^{t+1}  \vmu^{t} \right \rVert = O_p(1)
  \end{equation*}
  and
  \begin{equation*}
    \left \lVert \mQ^{t+1} - \psi^{(2)} (\vmu^{t}) \right \rVert = O_p(1)
  \end{equation*}
  where $\psi (\cdot) = \sum_{i=0}^n \phi_i (\cdot)$.
\end{theorem}
\begin{sketchproof} We first show the update of each individual site
  $\vr_{i}^{t}$ and $\mQ_{i}^{t}$ is equivalent to the first and second
  derivatives of $\phi_{i}$ evaluated at $\vtheta = \vmu^t$ (Supplementary
  Material, Theorem~\ref{app:thm:cavity-asymptote}). The key step of
  Theorem~\ref{app:thm:cavity-asymptote} is the application of Brascamp-Lieb
  inequality to upper-bound the covariance of the tilted distribution, then
  Taylor expand the terms in the upper-bound. Our stated result in this theorem
  follows by an application of the triangle inequality and
  Theorem~\ref{app:thm:cavity-asymptote}. Details are provided in Supplementary
  Material, Section~\ref{app:sec:equiv-nr-update}.
\end{sketchproof}

\textit{Remarks.} This theorem shows that, for a sufficiently large $n$, EPEL
behaves like Newton-Raphson if EPEL is initialized with a sufficiently large
precision, and the cavity and global mean are not too far apart. This connection
is important to establish the convergence of EPEL.

\subsection{Convergence of EPEL}


Let $\widehat{\vtheta} = \argmin_{\vtheta} \psi(\vtheta)$ be the
\emph{maximum-a-posteriori} solution. We now prove that $\mQ^{t}$ and $\vr^{t}$
are asymptotically equivalent to that of the Laplace approximation of
$p_{\EL}$, with precision $\psi^{(2)}(\widehat{\vtheta})$ and linear shift
$\psi^{(2)}(\widehat{\vtheta}) \widehat{\vtheta} - \psi^{(1)}(\widehat{\vtheta})$.
In particular, the EPEL algorithm converges asymptotically to a solution which coincides with the precision and linear shift of the Laplace approximate posterior in two iterations.
\begin{theorem}
  \label{thm:stable-region}
  Assume~\ref{item:theta-bound} to~\ref{item:prior-smoothness} hold. Consider
  the EPEL Gaussian initializations $\{\vr_i^0 \}_{i=0}^{n}$ and
  $\{\mQ_i^0 \}_{i=0}^{n}$ that satisfy
  \begin{equation*}
    n \max_i \lVert \mQ_i^0 \widehat{\vtheta} - \phi_i^{(1)} (\widehat{\vtheta})  - \vr_i^0    \rVert = \Delta_{\vr}^0
  \end{equation*}
  and
  \begin{equation*}
    n \max_i \lVert  \phi_i^{(2)} (\widehat{\vtheta})  - \mQ_i^0  \rVert = \Delta_{\vbeta}^0,
  \end{equation*}
  where $ \Delta_{\vr}^0 = \sqrt{n}$ and $\Delta_{\vbeta}^0 = \tfrac{1}{2} \lVert \psi^{(2)} (\widehat{\vtheta}) \rVert$. Then, for every $t=1,2,\ldots$, there exists a $\Delta_{\vr}^t$ and $\Delta_{\vbeta}^t$ such that
  \begin{equation*}
    \lVert \mQ_i^t \widehat{\vtheta} - \phi_i^{(1)} (\widehat{\vtheta}) - \vr_i^t \rVert \le n^{-1} \Delta_{\vr}^t
  \end{equation*}
  and
  \begin{equation*}
    \lVert \phi_i^{(2)} (\widehat{\vtheta})  - \mQ_i^t  \rVert \le n^{-1} \Delta_{\vbeta}^t,
  \end{equation*}
  where $\{\vr_i^t \}_{i=0}^{n}$ and $\{\mQ_i^t \}_{i=0}^{n}$ denote the EP
  approximation parameters after the $t$-th cycle. Moreover, if
  $ n^{-1} \psi^{(2)} (\widehat{\vtheta}) $ converges in probability to a
  constant positive definite matrix $\mV_{\vtheta^\star}$ with finite
  eigenvalues,
  $\sum_{i=0}^n \lVert \phi_i^{(1)} (\widehat{\vtheta}) \rVert = O_p(n)$, and
  $\sum_{i=0}^n \lVert \phi_i^{(2)} (\widehat{\vtheta}) \rVert = O_p(n)$, then
  \begin{equation*}
    \Delta_{\vr} ^t= O_p(1) \quad \emph{and} \quad \Delta_{\vbeta}^t = O_p(1) \quad \emph{for all} \quad t=2,3 \ldots.
  \end{equation*}
\end{theorem}
\begin{sketchproof}
  We start from the bounds in Supplementary Material,
  Theorem~\ref{app:thm:cavity-asymptote} and replace $\vmu^{t}$ with
  $\widehat{\vtheta}$ via Taylor expansions and algebra. This yields recursive
  bounds in terms of $\Delta_{\vr}^{t}$ and $\Delta_{\vbeta}^{t}$. We substitute
  the initial $\Delta_{\vr}^{0}$ and $\Delta_{\vbeta}^{0}$ to obtain the stated
  rates.
\end{sketchproof}

\textit{Remarks.}
By the triangle inequality,
$\lVert \mQ^t \widehat{\vtheta} - \psi^{(1)} (\widehat{\vtheta}) - \vr^t \rVert \le \Delta_{\vr}^t$
and
$\lVert \psi^{(2)} (\widehat{\vtheta}) - \mQ^t \rVert \le \Delta_{\vbeta}^t$.
Notably, the first iteration has $\Delta_{\vbeta}^{1} = O_{p}(\sqrt{n})$ and
$\Delta_{\vr}^{1} = O_{p}(1)$. Empirically, we also observe a spike in $\mQ^{1}$
in the first iteration, followed by gradual convergence. The first and second
moments of EPEL are asymptotically equivalent to that of the Laplace
approximation of $p_{\EL}$ after at least two iterations under the required
assumptions. Details are provided in Supplementary Material,
Section~\ref{app:sec:epel-convergence}.

\subsection{EPEL Bernstein-von-Mises Theorem}
Finally, we establish a Bernstein–von–Mises result for EPEL under the stability
conditions from Theorem~\ref{thm:stable-region} and iterations of at least two
cycles. The Bayesian empirical‑likelihood posterior $p_{\EL}$ is asymptotically
normal and equivalent to the EPEL solution. For this result, we require the
following additional assumptions, which are standard for Bayesian empirical likelihood theory:
\begin{enumerate}[start=5, label=(\Roman*)]
  \item We assume
        $\mS^\star = \bE\left [ h (\vz, \vtheta^\star) h (\vz, \vtheta^\star)^\top \right ]$
        is positive definite with bounded entries. \label{item:S-star-pd-bounded}
  \item For any $a > 0$, there exists $\nu_{EL} > 0$ such that as
        $n \rightarrow \infty$, we have \label{item:el-bounded}
        \begin{equation*}
          \lim_{n \rightarrow \infty} \bP^\star \left ( \sup_{\lVert \vtheta - \vtheta^\star \rVert \ge a} n^{-1} \left \{ \log \EL_n (\vtheta) - \log \EL_n (\vtheta^\star) \right \} \le - \nu_{\EL} \right ) = 1
        \end{equation*}
  \item The quantities
        $\lVert \partial h(\vz, \vtheta^\star) / \partial \vtheta \rVert$,
        $\lVert \partial^2 h(\vz, \vtheta) / \partial \vtheta \partial \vtheta^\top \rVert$,
        $\lVert h(\vz, \vtheta) \rVert^3$ are bounded by some integrable
        function $\widetilde{G} (\vz)$ in a neighbourhood of $\vtheta^\star$. \label{item:h-bounded}
  \item We assume $\mD = \bE\{ \partial h(\vz, \vtheta) / \partial \vtheta \}$
        is full rank, where the expectation is taken with respect to the true
        distribution of $\vz$. \label{item:D-full-rank}
\end{enumerate}

\begin{theorem}
  \label{thm:ep-bvm}
  Assume conditions~\ref{item:theta-bound} to~\ref{item:D-full-rank} hold.
  Consider the EPEL Gaussian posterior parameterized by the linear-shift
  $\vr = \sum_{i=1}^{n} \vr_i$, precision $\mQ = \sum_{i=1}^{n} \mQ_i$ that are
  obtained after at least two EP cycles with initializations
  $\{\vr_i^0 \}_{i=0}^{n}$ and $\{\mQ_i^0 \}_{i=0}^{n}$ that satisfy
  \begin{equation*}
    n \max_i \lVert \mQ_i^0 \widehat{\vtheta} - \phi_i^{(1)} (\widehat{\vtheta}) - \vr_i^0 \rVert = \Delta_{\vr}^0
  \end{equation*}
  and
  \begin{equation*}
    n \max_i \lVert \phi_i^{(2)} (\widehat{\vtheta}) - \mQ_i^0 \rVert = \Delta_{\vbeta}^0,
  \end{equation*}
  where $ \Delta_{\vr}^0 = \sqrt{n}$ and $\Delta_{\vbeta}^0 = \tfrac{1}{2} \lVert \psi^{(2)} (\widehat{\vtheta}) \rVert$. Then, we have
  \begin{equation*}
    \dtv ( \mathcal{N}(\mQ^{-1}\vr,\mQ^{-1}), p_{\EL}(\vtheta \mid \sD_{n}) ) = o_p(1).
  \end{equation*}
\end{theorem}
\begin{sketchproof}
  First, we show that the asymptotic behaviour of $\psi$ and $\phi_{i}$ satisfy
  the rates required in Theorem~\ref{thm:stable-region}. We then show that the
  total variation distance between the EPEL solution and the Laplace
  approximation is $O_{p}(n^{-1/2})$. We then show the total variation distance
  between the posterior $p_{\EL}$ and the Laplace approximation is $o_{p}(1)$
  (Supplementary Material, Lemma~\ref{app:thm:posterior-mode-bvm}). Applying
  triangle inequality yields the claim. Details are provided in Supplementary
  Material, Section~\ref{app:sec:epel-bvm}.
\end{sketchproof}

\section{Experiments}
\label{sec:experiments}
In this section, we examine the cost–accuracy trade-off across different methods
for computing the Bayesian empirical likelihood posterior. As this type of
analysis is uncommon in the computational statistics literature, we briefly
explain the motivation for our numerical study and the way in which we report
the results. In particular, our study is motivated by skepticism about whether
certain approximate posterior inference methods are computationally worthwhile
in practice, given that algorithms such as expectation propagation and
variational Bayes can require non-negligible computation time to achieve
algorithmic convergence, while MCMC methods may, in some settings, explore the
posterior distribution surprisingly efficiently.

The methods under
consideration are: Laplace approximation, HMC \citep{chaudhuri17hamiltonian}, variational Bayes
\citep{yu24variational} with a full-covariance Gaussian variational
distribution, and the Metropolis–Hastings algorithm with a random-walk proposal.


We tracked the posterior approximation for each method over many iterations. For
HMC and Metropolis--Hastings, this was the empirical distribution of the
cumulative samples. For variational Bayes and EPEL, it was the current
variational approximation to $p_{\EL}(\vtheta \mid \sD_{n})$. The posterior
approximations were compared against a ‘gold standard’ to assess their quality.
This gold standard was an empirical distribution formed by either
$2 \times 10^{6}$ draws from an HMC sampler for examples with a differentiable
$\cons$, or otherwise $10^{7}$ draws from the Metropolis–Hastings algorithm. The
exact implementation and hyperparameters of these algorithms are given in
Supplementary Material, Section~\ref{app:sec:hyperparams}.

To assess the difference between the gold standard and each posterior
approximation, we drew 1000 samples from each and computed the optimal
non-bipartite pairings \citep[NBPs;][]{derigs88solving}. The NBP computation was
done on the pooled 2000 samples (1000 from each of the gold standard and the
assessed method). We then counted the cross-match NBPs, which are pairs with
exactly one sample from the gold standard and one sample from the assessed
method. A larger number of cross-match NBPs indicates that the two distributions
are more similar. We used 474 pairs as a cut-off to determine whether an
approximate posterior had achieved sufficient accuracy. This threshold
corresponds to the 0.05 quantile of the NBP statistic’s exact null distribution
(based on 1000 pseudo-samples), i.e., $\Pr(\mathrm{NBP} \geq 474) \geq 0.05$,
under the null hypothesis that the approximation is identical to the
gold-standard posterior. An NBP reading greater than 474 would fail to reject
the null hypothesis and suggests that the two distributions are statistically
indistinguishable. The NBPs were computed with the \texttt{nbpMatching}
\texttt{R} package \citep{beck24nbpmatching}.

We report the NBP statistic in two ways. First,
Figure~\ref{fig:nbp-median-main-experiments} shows how the distribution of the NBP statistic for
each method evolves with wall-clock time. This summarizes the cost--accuracy
trade-off. Second, Tables~\ref{tab:nbp-threshold-standard}
and~\ref{tab:nbp-threshold-skew} report the time required to reach sufficient
approximation quality. We define this time as the first recorded wall-clock
time at which the 50 replicate NBP statistics are significantly above 474 by a
one-sided Wilcoxon signed-rank test at the 5\% level.

We assigned a $\sN(0, 10^{2})$ prior on each element of $\vtheta$. The EPEL and
variational Bayes are initialized with the Laplace approximation of $p_{\EL}$.
This initial approximation is also included as a baseline in our comparisons.
The sites of EPEL $\veta^{0}_{i}$ are initialized as
$\veta^{0}_{i} = \veta^{0} / (\ndata + 1)$ where $\veta^{0}$ is the natural
parameter corresponding to the Laplace approximation of $p_{\EL}$. For the MCMC
methods (HMC and Metropolis–Hastings), we ran a sufficiently long burn-in period before
collecting samples. All results are averaged over 50 independent runs.

We implemented all methods with \texttt{JAX} \citep{bradbury18jax}, a
\texttt{Python} package which facilitates automatic differentiation and obviates
the need for manual derivations of derivatives for gradient-based methods. We
ran our experiments in containers to ensure reproducibility. Each container was
assigned two virtual cores and sufficient memory. The workload was run on
AMD EPYC 9474F 3.6 GHz processors. Code is available at
\url{https://github.com/weiyaw/epel}.

\subsection{Instrumental variables regression}
We begin with an instrumental variables regression using the \texttt{wage}
dataset of \citet{mroz87sensitivity}, which records the wages of 428 women in the
labour force. To control for the effects of unmeasured confounders, \cite{mroz87sensitivity} proposed using father's education and mother's education as instrumental variables. We model wage as a function of education, experience, and squared
experience. Let $y_i$ denote the logarithm of wage,
$\vx_i=(1,e_i,x_i,x_i^2)^\top$, and $\vz_i=(1,f_i,m_i,x_i,x_i^2)^\top$, where
$e_i$, $x_i$, $x_i^2$, $f_i$, and $m_i$ denote education, experience, squared
experience, father's education, and mother's education, respectively. All
continuous covariates are standardized. With
$\vtheta=(\theta_0,\theta_1,\theta_2,\theta_3)^\top$, the residual is
$r_i(\vtheta)=y_i-\vx_i^\top\vtheta$, and the constraint function is
$\cons(\vz_i,\vtheta)=\vz_i r_i(\vtheta)$. This gives five constraints for four
regression coefficients.

In this example, the posterior is relatively close to Gaussian. The
Gaussian-based methods (EPEL, variational Bayes, and the Laplace approximation)
therefore perform well from the outset; see Figure~\ref{fig:nbp-median-wage}.
EPEL continues to improve with additional computation and eventually reaches a
quality comparable to that of HMC.

\begin{table}[t]
  \centering
  \begin{tabular}{lrrrr}
    \toprule
    Setup & EPEL & HMC & MH & VB \\
    \midrule
    Instrumental variables regression & 0.0 & 78.3 & 257.4 & 0.0 \\
    Quantile regression & 18.1 & 27.1 & 63.3 & 4.5 \\
    Generalized estimating equations & 34.9 & 122.1 & 662.7 & -- \\
    \bottomrule
  \end{tabular}
  \caption{Time (in seconds) to reach sufficient approximation quality, defined
    as the first recorded time at which the 50 replicate NBP statistics are
    significantly above 474 by a one-sided Wilcoxon signed-rank test at the 5\%
    level. A dash indicates that the method did not attain this threshold within
    the recorded time grid.}
  \label{tab:nbp-threshold-standard}
\end{table}

\begin{figure*}[t]
  \centering
  \begin{subfigure}{0.49\textwidth}
    \includegraphics[width=\linewidth]{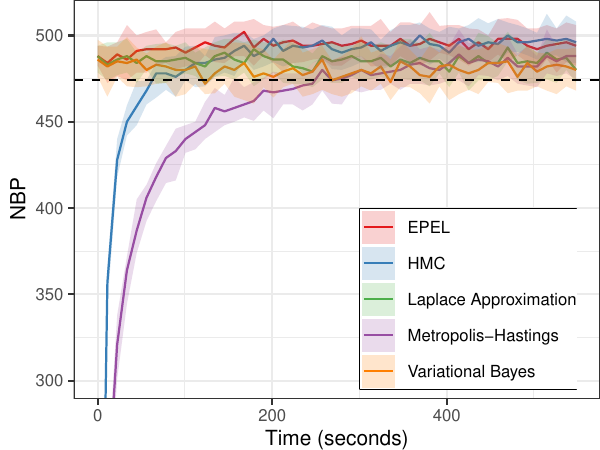}
    \caption{Instrumental variables regression}
    \label{fig:nbp-median-wage}
  \end{subfigure}%
  \begin{subfigure}{0.49\textwidth}
    \includegraphics[width=\linewidth]{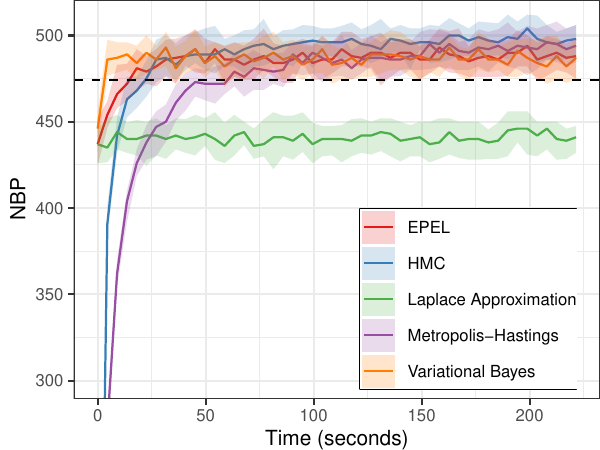}
    \caption{Quantile regression}
    \label{fig:nbp-median-quantregression}
  \end{subfigure}
  \begin{subfigure}{0.49\textwidth}
    \includegraphics[width=\linewidth]{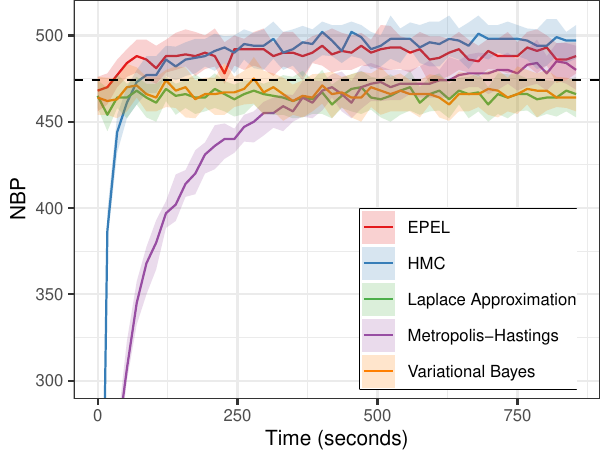}
    \caption{Generalized estimating equations}
    \label{fig:nbp-median-gee}
  \end{subfigure}%
  \begin{subfigure}{0.49\textwidth}
    \includegraphics[width=\linewidth]{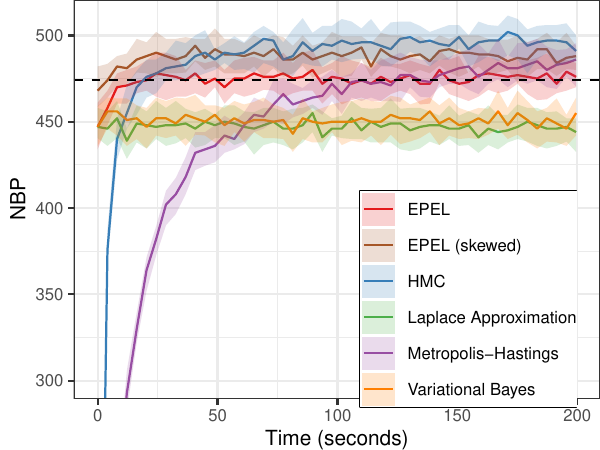}
    \caption{Logistic regression with skewed posterior}
    \label{fig:nbp-median-skew-logistic}
  \end{subfigure}
  \caption{NBP statistics tracked over computation time for all methods and
    experimental setups. Coloured curves show the median NBP with respect to the
    gold standard over 50 repetitions. Shaded bands denote the 0.25--0.75
    quantiles. The dotted horizontal line marks the accuracy threshold of 474.}
  \label{fig:nbp-median-main-experiments}
\end{figure*}

\subsection{Quantile regression}
We now consider quantile regression to estimate a function
$\vx \mapsto \vx^{\top} \vtheta$ such that
$\pr(y \leq \vx^{\top} \vtheta \mid \vx) = \tau$ for some quantile
$\tau \in [0, 1]$. Here, we set $\tau = 0.7$. We generated $n=100$ observations
from $y_i = \vx_i^\top \vtheta_0 + \epsilon_i$, where $\vx_i = (1, x_i)^\top$,
$x_i \sim \mathcal{N}(0,1)$, $\epsilon_i \sim \mathcal{N}(0,1)$, and
$\vtheta_0 = (0.5,1)^\top$. Quantile regression with Bayesian empirical
likelihood has previously been studied in \citet{yang12bayesian}. In their work,
they use the constraint function
$\cons(\vx, y, \vtheta) = \rho_{\tau}(y - \vx^{\top} \vtheta)\vx$, with quantile
score function $\rho_{\tau}(u) = (1-\tau)\ind\{u < 0\} - \tau\ind\{u > 0\}$ and
$\rho_{\tau}(0)=0$. This form of $\cons$ has zero value for
$\nabla_{\vtheta} \cons$ almost everywhere, and is problematic for
gradient-based methods, e.g., HMC and variational Bayes. For methods that
require a well-behaved $\nabla_{\vtheta} \cons$, we replace $\rho_{\tau}$ with a
smooth approximation,
$\widetilde \rho_{\tau}(u) = \expit(-u / \epsilon_{\rho}) - \tau$, with a
sufficiently small $\epsilon_{\rho}$. We use $\epsilon_{\rho} = 0.1$ here.

A well-behaved $\nabla_{\vtheta} \cons$ is not strictly necessary for EPEL if
$q_{\backslash i}$ is approximated with importance sampling. In practice, we
prefer the Laplace-based computation of $q_{\backslash i}$, as we observed
substantial gains in computation time with negligible impact on accuracy from
approximating~$\rho_{\tau}$.

Figure~\ref{fig:nbp-median-quantregression} shows that all methods except
the Laplace approximation achieve good approximation quality. Variational Bayes
reaches the threshold fastest, in 4.5 seconds. It is followed by EPEL in
18.1 seconds, HMC in 27.1 seconds, and Metropolis--Hastings in 63.3 seconds.
Within the recorded computation times, the smooth approximation to $\rho_{\tau}$
has negligible effect on approximation quality. The resulting approximations
are similar across the applicable methods.

\subsection{Generalized estimating equations}
Generalized Estimating Equations (GEE) is a classical technique for estimating
regression coefficients in longitudinal data analysis. In this example, we adopt
the over-identification case, i.e., the number of estimating equations is
greater than $\dparam$, from \citet{chang18new}. Our data were generated from a
repeated-measures model
$\vy_{i} = \vx^{\top}_{i} \vtheta_{0} + \vepsilon_{i}, i = 1, \ldots, 50$, where
$\vtheta_{0} = (3, 1.5, 0, 0, 2)^{\top}$, $\vy_{i} = (y_{i1}, y_{i2})^{\top}$,
$\vx_{i} = (x_{ijk})^{\dparam, 2}_{j,k=1} \in \R^{\dparam \times 2}$, and
$\vepsilon_{i} = (\epsilon_{i1}, \epsilon_{i2})^{\top}$. The
$(x_{ij1}, \ldots, x_{ijp})^{\top}$ are generated from
$\mathcal{N}(\vzero, \mSigma)$ with
$\mSigma = (0.5^{\lvert k - l \rvert})^{\dparam, \dparam}_{k, l=1} \in \R^{\dparam \times \dparam}$.
The $\vepsilon_{i}$ are generated from a bivariate Gaussian distribution with
zero mean and unit-marginal compound-symmetry covariance matrix with
off-diagonal terms set to $0.7$. We use the quadratic inference function
introduced in \citet{qu00improving} to construct the constraint function
\begin{equation*}
  \cons(\vy, \vx, \vtheta) =
  \begin{bmatrix*}
    \vx \bm{M}_{1} (\vy - \vx^{\top} \vtheta) \\
    \vx \bm{M}_{2} (\vy - \vx^{\top} \vtheta),
  \end{bmatrix*}
\end{equation*}
where $\bm{M}_{1}$ is a two-dimensional identity matrix, and $\bm{M}_{2}$ a
unit-marginal compound-symmetry covariance matrix with off-diagonal terms set to
$0.7$.

In this example, HMC achieves the highest approximation quality and is closely
followed by EPEL; see Figure~\ref{fig:nbp-median-gee}. EPEL reaches the required
threshold in 34.9 seconds. This is roughly one-third of the 122.1 seconds
required by HMC. Neither variational Bayes nor the Laplace approximation attains
sufficient approximation quality within the recorded time grid.
Metropolis--Hastings reaches the threshold only after substantially longer
computation, taking 662.7 seconds.


\subsection{Logistic regression}
We conclude the experiment with logistic regression on the \texttt{O-rings} data
\citep{dalal89risk}. The data contain 23 pre-Challenger space-shuttle missions,
with launch temperature and the number of damaged components recorded for each
mission. We use a binary response indicating whether any component was damaged,
with launch temperature as the only covariate. For logistic regression, we specify
an orthogonality constraint
$\cons(y, \vx, \vtheta) = \vx(y - \expit(\vx^{\top} \vtheta))$. All continuous
covariates are standardized.

\begin{figure*}[t]
  \centering
  \includegraphics[width=\textwidth]{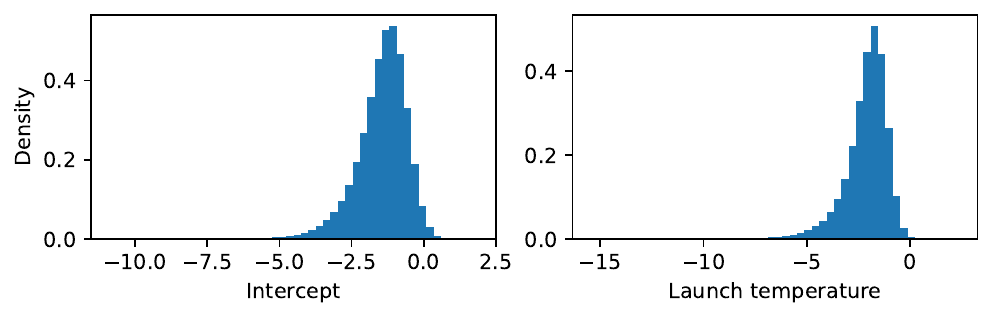}
  \caption{Gold-standard marginal posteriors for the \texttt{O-ring}
    logistic regression.}
  \label{fig:orings-gold-posterior}
\end{figure*}

The main computational challenge in this example is the skewness of the
posterior; see Figure~\ref{fig:orings-gold-posterior}. As expected, both
variational Bayes and the Laplace approximation perform poorly in this setting;
see Figure~\ref{fig:nbp-median-skew-logistic}. Standard EPEL provides a better
approximation, and its quality improves further after applying the
post-processing correction in Section~\ref{sec:skew-post-processing}. The
post-processed EPEL approximation has the best cost--accuracy trade-off. It
requires only 8.1 seconds compared with 28.5 seconds for HMC to reach sufficient
approximation quality (Table~\ref{tab:nbp-threshold-skew}).

\begin{table}[t]
  \centering
  \begin{tabular}{lrrrrr}
    \toprule
    Setup & EPEL & \shortstack[c]{EPEL\\(post-process)} & HMC & MH & VB \\
    \midrule
    Logistic regression & 85.5 & 8.1 & 28.5 & 130.3 & -- \\
    \bottomrule
  \end{tabular}
  \caption{Time (in seconds) to reach sufficient approximation quality,
    defined as the first recorded time at which the 50 replicate NBP statistics
    are significantly above 474 by a one-sided Wilcoxon signed-rank test at the
    5\% level. The EPEL (post-process) column reports EPEL after the skewness
    correction in Section~\ref{sec:skew-post-processing}. Missing entries
    indicate that the method did not attain this threshold within the recorded
    time grid.}
  \label{tab:nbp-threshold-skew}
\end{table}


\section{Discussion}
\label{sec:discussion}
Computing the posterior of Bayesian empirical likelihood is a challenging and
costly endeavour. To address this, we propose EPEL, a computationally efficient
alternative to existing methods. For sufficiently large datasets and under
standard regularity conditions, we show that EPEL and the exact posterior of
Bayesian empirical likelihood are asymptotically equivalent. In the course of showing this asymptotic equivalence, we proposed a set of realistic sufficient conditions for which the Bayesian empirical likelihood is smooth throughout the feasible parameter space. Through extensive
experiments, we also demonstrate that EPEL generally achieves a better
cost–accuracy trade-off than HMC, while avoiding the problematic multi-layered
approximations in variational Bayes empirical likelihood when $\ndata$ is small
relative to $\dparam$. Furthermore, EPEL can be implemented without
differentiating the constraint function when the tilted moments are estimated by
importance sampling, making it useful for models where a smooth approximation of
the constraint is unavailable. Moreover, we demonstrated that the accuracy of EPEL posterior approximation can be further enhanced through post-processing.

Like all methods, EPEL has limitations. First, our theoretical justification
covers only the case in which the approximate posterior is Gaussian. Such
Gaussian approximations may be inappropriate in rare circumstances where the
posterior is a product of multiple empirical likelihoods; see, e.g.,
\citet{chaudhuri11empirical}. Second, the algorithm may occasionally fail due to
non-convergence of the inner-loop optimization used for Laplace approximation.
Third, our proposed approximate posterior is justified only when
$\dparam < \ndata$. We leave extensions to high-dimensional Bayesian empirical
likelihood inference \citep{chang25bayesian} for future work.


\section*{Acknowledgement}
The authors would like to thank Susan Wei, Liam Hodgkinson, Minh-Ngoc Tran and Dino Sejdinovic for helpful comments in the development of this work. Most of this work was conducted while K.N. was at University of Melbourne. K.N. is
supported by the Australian Government Research Training Program and the
Statistical Society of Australia PhD Top-up Scholarship. Computational resources
were provided by the ARDC Nectar Research Cloud and the Melbourne Research
Cloud. H.B. is supported by the Australian Research Council.



\bibliography{main}
\bibliographystyle{dcu-custom}

\newpage

{\LARGE\bfseries Supplementary Material\par}
\vspace{1.5em}

\appendix
\setcounter{theorem}{4}
\setcounter{lemma}{2}

\section{Further numerical considerations for EPEL}
\label{app:sec:epel-practical-notes}

\paragraph{Practical consideration for importance sampler.} We monitor the
effective sample size \citep[Eq. 11.8]{gelman13bayesian} to determine the
quality of the approximations. For numerical stability when computing
$\widetilde{\mSigma}_{\IS}$, we perform QR decomposition on the weighted scatter
matrix $\mS = \mQ \mR$, where each row of $\mS$ is
$\sqrt{\xi^{l}} (\vtheta^{l} - \widetilde{\vmu}_{\IS})^{\top}$. Then,
$\widetilde{\mSigma}_{\IS} = (\sum^{\nis}_{l=1} \xi^{l})^{-1} \mR^{\top} \mR$,
and $\widetilde{\mQ}_{\IS}$ can be obtained directly from $\mR$.

\paragraph{Reducing the number of sites.} To increase computational efficiency,
the individual terms in $\EL(\vtheta)$ can be pooled to reduce the number of
sites. That is, the global approximation is
$q(\vtheta) = \prod_{j=1}^{\nsites} q_{j}(\vtheta)$, where
$\nsites < \ndata + 1$ is the number of sites and each of these $q_{j}(\vtheta)$
will correspond to multiple~$w_{i}$.

\paragraph{`Warm-up' cycles using Laplace approximate updates.}
The expectation-propagation algorithm is widely known to be numerical unstable,
especially with poor initializations. To improve numerical stability, we
introduce warm-up cycles in the EPEL procedure with the moments of the Laplace
approximate posterior as the updates for $\widetilde{\veta}_{i}$. Once
$\|\Delta \veta_{i}\|$ is small, we switch over to importance sampling to ensure
we are indeed approximating the moments of the tilted distributions.

\paragraph{Positive-definiteness of the precision of Laplace's approximation.}
Laplace's approximation requires a numerical optimizer to compute the mode of
the target distribution. Due to machine precision errors, the optimizer may not
converge sufficiently close to the mode, and the Hessian at the optimizer
solution can be substantially different from that at the true mode. The negative
Hessian may not even be positive-definiteness, yielding an ill-defined
approximation. When this occurs when approximating $q_{\backslash i}$, we use
importance sampling as a fallback, using either the cavity $q_{-i}$ or the
current global approximation $q^{t}$ as the proposal; the latter is guaranteed
to be a proper by construction.

\paragraph{Positive‑definiteness of the global precision.}
Since $\mQ^{t}$ is updated with addition and subtraction, the new update
$\mQ^{t+1}$ may not be positive‑definite. We enforce positive‑definiteness by
starting from a positive‑definite $\mQ^{0}$ and dynamically decreasing the
damping factor if the proposed update $\sum_{i} \Delta \veta_{i}$ would render
$\mQ^{t+1}$ improper. Nonetheless, this procedure does not guarantee
positive‑definiteness of each site precision $\mQ_{i}$, which is acceptable as
the site approximations are intermediate quantities
\citep{vehtari20expectation}.

\section{Further Comments of Assumption~\ref{app:item:h-location}}
\label{app:sec:assump2-comments}
Assumption~\ref{app:item:h-location} is in fact a \emph{compatibility} assumption
imposed on the combination of the true data-generating process $F_{\vz}$, the
specified parameter space $\mTheta$ and the moment conditions $h$. Intuitively,
it says that the specified parameter space \emph{cannot} be larger than the
\emph{induced parameter space}, i.e., the set of possible values of $\vtheta$ induced by the true data-generating
distribution. In particular, consider the data-generating distribution $F_\vz$
and $h$ such that there exists a $\vtheta^\star \in \bR^p$ that satisfies
$\bE \{ h(\vz,\vtheta^\star) \} = \vzero$. The induced parameter space of
$(F_\vz, h)$ is
\begin{equation*}
  \mTheta^{F_{\vz}, h} = \{ \vtheta \in \bR^p \, : \, h(\vz, \vtheta) \text{ satisfies (*)} \},
\end{equation*}
where (*) is the following condition: \emph{If $h(\vz, \vtheta)$ is a discrete
  random vector, there exists
  $\sP_{\vtheta} = \{ g_{1,\vtheta}, \ldots, g_{L,\vtheta} \}$ such that
  $\min_{a \in [L]} \bP^\star \{ h(\vz,\vtheta) = g_{a,\vtheta} \} > 0$ and
  $\vzero_K$ is in the interior of the convex hull of $\sP_{\vtheta}$. If
  $h(\vz, \vtheta)$ is a continuous random vector, there exists closed and
  connected sets $\sG_{1,\vtheta}, \ldots, \sG_{L,\vtheta}$ such that every
  combination of points in the respective sets
  $\{ g_a \in \sG_{a,\vtheta} \}_{a =1}^{L} $ form a convex hull containing
  $\vzero_K$ in its interior and $\vzero_K \notin \sG_{a,\vtheta}$. Moreover,
  the induced density of $h(\vz,\vtheta)$ is strictly positive for all points in
  $\bigcup_{a=1}^{L} \sG_{a,\vtheta}$ and
  $\min_{a \in [L]} \bP^\star \{ h(\vz,\vtheta) \in \sG_{a,\vtheta} \} > 0$.}
Then, assumption~\ref{app:item:h-location} is violated whenever $\mTheta$ is not a
subset of $\mTheta^{F_{\vz}, h}$. While this compatibility assumption is
non-restrictive, we provide some conceivable examples where it may be violated.
A common issue among these examples is a misspecification of the data space.

\subsection*{Inference of Pareto means}
Suppose the observed data $z_i$ are i.i.d.~$\mathrm{Pareto}( x_m, \alpha^\star)$
for some $x_m > 0$ and $\alpha^\star > 0$ and the modeller specified the moment
condition function $h(z, \theta) = z - \theta$. Here, the accommodated parameter
space is $\mTheta^{F_{\vz}, h} = ( x_m , \infty)$. On the other hand, if the
modeller misjudged the lower bound of the data support so that the specified the
parameter space is $\mTheta = [ x_m/2, \vtheta_{\high} ]$ for some
$\vtheta_{\high} \gg x_m /2$. Then, $\mTheta$ is clearly not a subset of
$\mTheta^{F_{\vz}, h}$. In particular, the distribution of $h(\vz, x_m /2)$ does
not satisfy (*).

\subsection*{Regression with bounded response and covariate support}
Consider the observed data $\{(x_i, y_i)\}_{i=1}^n$ that is generated independently through the following regression model: $y = \theta_0^\star + \theta_1^\star x + \epsilon^\star$, where the support of the covariate is $[x_{\min}, x_{\max}]$ and the support of the random error is $[\epsilon_{\low}, \epsilon_{\high}]$ for some $\epsilon_{\low} < 0 < \epsilon_{\high}$ and $x_{\min} >0$. Suppose the modeller specifies the moment condition
\begin{equation*}
  h(x_i, y_i, \vtheta) =
  \begin{pmatrix}
    y_i - \theta_0 - \theta_1 x_i \\
    x_i (y_i - \theta_0 - \theta_1 x_i)
  \end{pmatrix}
\end{equation*}
and
$\mTheta = [ \theta_{0; \low}, \theta_{0; \high}] \times [ \theta_{1; \low}, \theta_{1; \high}]$,
where $\theta_{0; \low} < \theta_0^\star < \theta_{0; \high}$ and
$\theta_{1; \low} < \theta_1^\star < \theta_{1; \high}$. Since $x$ is strictly
positive with probability one, a necessary and sufficient condition for which
assumption~\ref{app:item:h-location} is violated is whenever we can find a point $\vtheta$ in the specified parameter space $\mTheta$ such that $y - \theta_0 - \theta_1 x > 0$ for all $x \in [x_{\min}, x_{\max}]$ and $\epsilon^\star \in [\epsilon_{\low}, \epsilon_{\high}]$ or $y - \theta_0 - \theta_1 x < 0$ for all $x \in [x_{\min}, x_{\max}]$ and $\epsilon^\star \in [\epsilon_{\low}, \epsilon_{\high}]$.
These are equivalent to the conditions
$\epsilon^\star > (\theta_0 - \theta_0^\star) + x(\theta_1 - \theta_1^\star)$ for all $x \in [x_{\min}, x_{\max}]$ or $\epsilon^\star < (\theta_0 - \theta_0^\star) + x(\theta_1 - \theta_1^\star)$ for all $x \in [x_{\min}, x_{\max}]$.
Hence, a sufficient condition for which $\mTheta$ is not a subset of the
accommodated parameter space $\mTheta^{F_{x,y},h}$ is
\begin{equation*}
  \epsilon_{\low} > (\theta_{0, \low} - \theta_0^\star) + (\theta_{1, \low} - \theta_1^\star) x_{\min}
\end{equation*}
or
\begin{equation*}
  \epsilon_{\high} < (\theta_{0, \high} - \theta_0^\star) + (\theta_{1, \high} - \theta_1^\star) x_{\min}.
\end{equation*}
\section{Asymptotics of EPEL}
\subsection{Notations and assumptions}
Suppose we have i.i.d.~data $\sD_{n} = \{ \vz_i \}_{i=1}^n$ and we are interested in
inferring an unknown parameter $\vtheta$ such that
\begin{equation*}
  \bE \left \{ h (\vz_i, \vtheta) \right \} = \vzero,
\end{equation*}
for some moment  condition function $h: \sZ \times \mTheta \rightarrow \bR^K$ for some data space $\sZ$ and parameter space $\mTheta$. The log-empirical likelihood function is given by
\begin{equation*}
  \log \EL (\vtheta; \sD_{n}) = \sum_{i=1}^n \log \{ w_i (\vtheta; \sD_{n}) \},
\end{equation*}
where $w_i (\vtheta; \sD_{n}) = \{ n + n \vlambda (\vtheta) ^\top h (\vz_i, \vtheta) \}^{-1}$ if there exists a solution $\vlambda (\vtheta) \in \bR^K$ that satisfies
\begin{equation*}
  \sum_{i=1}^n \frac{h (\vz_i, \vtheta)}{1 + \vlambda^\top h (\vz_i, \vtheta) } = \vzero
\end{equation*}
and $w_i (\vtheta; \sD_{n}) = 0$ otherwise. The support of the empirical likelihood
posterior is denoted as
\begin{equation*}
  \mTheta_{B;n} = \left \{ \vtheta : \text{there exists weights $\vw$ such that} \; \sum_{i=1}^n w_i h(\vz_i, \vtheta), \; w_ i > 0, \; \text{and} \; \sum_{i=1}^n w_i = 1  \right \}.
\end{equation*}

Furthermore, let $\phi_i (\vtheta; \sD_{n}) = - \log \{ w_i (\vtheta; \sD_{n}) \}$,
$\phi_{0}(\vtheta) = - \log p(\vtheta)$, and
$\psi (\cdot) = \sum_{i=0}^n \phi_i (\cdot)$.

We make the following assumptions:
\begin{enumerate}[label=(\Roman*)]
  \item The parameter space $\mTheta$ is bounded and
        $\vtheta^\star$ is a unique interior point. \label{app:item:theta-bound}
  \item For countable $h(\sZ,\vtheta)= \{h(\vz, \vtheta): \vz \in \sZ \}$: for
        each $\vtheta \in \mTheta$, there exist $L \geq K + 1$ non‑zero vectors
        $\sP_{\vtheta} = \{ g_{1,\vtheta}, \ldots, g_{L,\vtheta} \}$ whose
        convex hull contains $\vzero_K$ in its interior; and
        $\min_{a \in [L]} \bP^\star \{ h(\vz,\vtheta) = g_{a,\vtheta} \} > 0$.
        For uncountable $h(\sZ,\vtheta)$: for each $\vtheta \in \mTheta$, there
        exist $L \geq K + 1$ closed, connected sets
        $\sG_{1,\vtheta}, \ldots, \sG_{L,\vtheta}$ such that any
        $\{ g_a \in \sG_{a,\vtheta} \}_{a=1}^{L}$ have a convex hull with
        $\vzero_K$ in its interior and $\vzero_K \notin \sG_{a,\vtheta}$; the
        induced density of $h(\vz,\vtheta)$ is strictly positive on
        $\bigcup_{a=1}^{L} \sG_{a,\vtheta}$; and
        $\min_{a \in [L]} \bP^\star \{ h(\vz,\vtheta) \in \sG_{a,\vtheta} \} > 0$.
        \label{app:item:h-location}
  \item The first, second, and third derivatives of $h (\vz, \vtheta)$ with
        respect to $\vtheta$ are continuously differentiable on $\mTheta$ for
        all $\vz \in \sZ$.
        \label{app:item:h-smoothness}
  \item The prior $p(\vtheta)$ is positive on a neighbourhood of
        $\vtheta^\star$. Moreover, there exists $M_p > 0$ such that, for all
        $(a, b, c, d) \in \{ 1, \ldots, p \}^4$,
        \begin{equation*}
          \left| \diffp{\log p(\vtheta)}{\theta_a, \theta_b, \theta_c, \theta_d} \right|
          \le M_p \quad \text{for all } \vtheta \in \mTheta .
        \end{equation*} \label{app:item:prior-smoothness}
\end{enumerate}

\subsection{Smoothness of the empirical likelihood}
\label{app:sec:smoothness-el}
In parametric likelihood inference, the likelihood function is smooth throughout
the parameter space $\mTheta$ for any finite sample size $n \ge 1$. The same
cannot be said for empirical likelihood inference. In fact, for a finite sample
size, the empirical likelihood function may not be smooth in regions of
$\mTheta$ that are not near $\vtheta^\star$. In the first segment of this
section, we prove that the empirical likelihood function is non-smooth
throughout $\mTheta$ finitely often.

\begin{lemma}
  \label{app:thm:single-point-supported}
  Assume~\ref{app:item:h-location} holds. Then, for any $\vtheta \in \mTheta$,
  \begin{equation*}
    \bP^\star \left \{ \vtheta \notin \operatorname{int}(\mTheta_{B;n}) \; \emph{finitely often} \right \} = 1.
  \end{equation*}
\end{lemma}
\begin{proof}
  Denote
  $\sE_{n,\vtheta} = \{ \vtheta \notin \operatorname{int}(\mTheta_{B;n}) \}$.
  By Borel-Cantelli Lemma, we need only to show that
  \begin{equation*}
    \sum_{n=1}^\infty \bP^\star (\sE_{n,\vtheta}) < \infty.
  \end{equation*}
  For each $n \le L-1$, we have $\bP^\star (\sE_{n,\vtheta}) \le 1$. Hence,
  \begin{equation*}
    \sum_{n=1}^\infty \bP^\star (\sE_{n,\vtheta}) \le L - 1 + \sum_{n=L}^\infty \bP^\star (\sE_{n,\vtheta}).
  \end{equation*}
  Next, we bound $\sum_{n=L}^\infty \bP^\star \left (  \sE_{n,\vtheta} \right )$. In the countable $h(\sZ,\vtheta)$ case, we let
  \begin{equation*}
    \sF_{n,\vtheta} = \{ h(\vz_{n-L+1},\vtheta) = g_{1,\vtheta}, \, \ldots \, ,h(\vz_{n},\vtheta) = g_{L,\vtheta} \}.
  \end{equation*}
  Observe that $\bigcup_{m=L}^n \sF_{m,\vtheta} \subseteq \sE_{n,\vtheta}^c$
  for all $n \ge L$ and hence
  $\sE_{n,\vtheta} \subseteq \bigcap_{m=L}^n \sF_{m,\vtheta}^c$. Then,
  \begin{align*}
    \bP^\star ( \sE_{n,\vtheta} )
     & \le \bP^\star \left( \textstyle \bigcap_{m=L}^n \sF_{m,\vtheta}^c \right)                                   \\
     & \le \bP^\star \left( \textstyle \bigcap_{m=0}^{ \lfloor (n-L)/L \rfloor } \sF_{ (m+1) L ,\vtheta}^c \right) \\
     & = \prod_{m=0}^{ \lfloor (n-L)/L \rfloor } \bP^\star ( \sF_{ (m+1) L ,\vtheta}^c ).
  \end{align*}
  The final equality is due to the independence among
  $\sF_{L ,\vtheta}, \sF_{2L ,\vtheta}, \sF_{3L ,\vtheta}, \ldots$. Now, for
  any $n \ge L$, we have
  $\bP(\sF_{n,\vtheta}^c) = 1 - \bP^\star ( h(\vz,\vtheta) = g_{1,\vtheta} ) \times \ldots \times\bP^\star ( h(\vz,\vtheta) = g_{L,\vtheta} )\le 1 - \inf_{\vtheta \in \mTheta} \bP^\star ( h(\vz,\vtheta) = g_{1,\vtheta} ) \times \ldots \times\bP^\star ( h(\vz,\vtheta) = g_{L,\vtheta} ) < 1$,
  where the last inequality follows from~\ref{app:item:h-location}. Hence,
  \begin{align*}
    \sum_{n=L}^\infty \bP^\star (\sE_{n,\vtheta})
     & \le \sum_{n=L}^\infty \prod_{m=0}^{ \lfloor (n-L)/(L) \rfloor  } \bP^\star ( \sF_{ (m+1)L,\vtheta}^c ) \\
     & \le L \times \frac{  1 - \inf_{\vtheta \in \mTheta} \bP^\star ( h(\vz,\vtheta)
      = g_{1,\vtheta} ) \times  \ldots \times\bP^\star ( h(\vz,\vtheta)
      = g_{L,\vtheta} ) }{ \inf_{\vtheta \in \mTheta} \bP^\star ( h(\vz,\vtheta)
      = g_{1,\vtheta} ) \times  \ldots \times\bP^\star ( h(\vz,\vtheta)
      = g_{L,\vtheta} )  } < \infty.
  \end{align*}
  In the uncountable $h(\sZ,\vtheta)$ case, the required result holds by replacing $\sF_{n,\vtheta}$ with
  \begin{equation*}
    \sF_{n,\vtheta} = \{ h(\vz_{n-L+1},\vtheta) \in \sG_{1,\vtheta}, \, \ldots \, ,h(\vz_{n},\vtheta) \in \sG_{L,\vtheta} \}.
  \end{equation*}
  and $\bP^\star ( h(\vz,\vtheta)  = g_{a,\vtheta} )$ with $\bP^\star ( h(\vz,\vtheta)  \in  \sG_{a,\vtheta} )$ for every $a \in [L]$.
\end{proof}
Note that Lemma~\ref{app:thm:single-point-supported} may be compared with equation
(2.7) of \citet{owen90empirical} which establishes that
$\vtheta^\star \notin \operatorname{int}(\mTheta)$ finitely often.
\begin{lemma}
  \label{app:thm:lemma-truth-supported}
  Assume ~\ref{app:item:theta-bound}, ~\ref{app:item:h-location}, and~\ref{app:item:h-smoothness} holds. Then
  \begin{equation*}
    \bP^\star \left \{ \mTheta \nsubseteq \operatorname{int} ( \mTheta_{B;n} ) \; \emph{finitely often}  \right \} = 1.
  \end{equation*}
\end{lemma}
\begin{proof}
  Let $\sA_n = \{ \mTheta \nsubseteq \operatorname{int} ( \mTheta_{B;n} ) \}$.
  Similar to Lemma~\ref{app:thm:single-point-supported}, it suffices to show that
  \begin{equation*}
    \sum_{n=1}^\infty \bP^\star (\sA_n) < \infty.
  \end{equation*}
  For each $n \le L - 1$, we have $\bP^\star (\sA_n) \le 1$. Hence,
  \begin{equation*}
    \sum_{n=1}^\infty \bP^\star (\sA_n) \le L - 1 + \sum_{n=L}^\infty \bP^\star (\sA_n).
  \end{equation*}
  Note that it is difficult to directly analyse
  $\sA_n = \bigcup_{\vtheta \in \mTheta} \sE_{n,\vtheta}$ as it is an
  uncountable union of sets. To circumvent this issue, we express $\mTheta$ as
  a finite union of small compact balls (by assumption \ref{app:item:theta-bound}) and study a sufficient condition for which
  each of these balls are wholly included in the posterior support
  $\mTheta_{B;n}$. For any $\epsilon > 0$, consider a finite sequence of
  equal-sized possibly overlapping balls
  $\{ \overline{\sB}_{\epsilon} ( \vtheta^{(r)} ) \}_{r=1}^R$ such that each
  ball has radius $\epsilon$, $\vtheta^{(r)} \in \operatorname{int}(\mTheta)$,
  and
  $\mTheta \subseteq \bigcup_{r=1}^R \overline{\sB}_{\epsilon} ( \vtheta^{(r)} )$.
  \footnote{Strictly speaking, $R$ depends on $\epsilon$ but we omit this
    dependence from its notations for brevity. Also, for any $\epsilon > 0$,
    the integer $R$ is always finite.} For each ball, define a corresponding
  sliced ball as
  $\sB_{\epsilon } ( \vtheta^{(r)} ) = \overline{\sB}_{\epsilon} ( \vtheta^{(r)} ) \cap \mTheta$.
  Then, it suffices to show that
  \begin{equation*}
    \sum_{r=1}^R  \sum_{n = L}^\infty \bP \left \{ \sB_{ \epsilon  } ( \vtheta^{(r)} ) \nsubseteq \operatorname{int} ( \mTheta_{B;n} ) \right \} < \infty
  \end{equation*}
  for some sufficiently small $\epsilon$. The key to establishing the aforementioned inequality is to carefully ``tune" $\epsilon$ and then show that
  $\vtheta^{(r)} \in \operatorname{int}(\mTheta_{B;n})$ implies
  $\sB_{ \epsilon } ( \vtheta^{(r)} ) \subseteq \operatorname{int} ( \mTheta_{B;n} )$.
  In the countable $h(\sZ,\vtheta)$ case, we set
  \begin{equation*}
    0 < \epsilon < \frac{ \inf_{\vtheta \in \mTheta} \{ \text{shortest distance between origin and convex hull of } \sP_{\vtheta}  \}  }{2 b} ,
  \end{equation*}
  where
  $b = \sup_{\vtheta \in \mTheta, \vz \in \sZ} \lVert h^{(1)}( \vz, \vtheta ) \rVert$
  and we know from assumption~\ref{app:item:h-smoothness} that $b < \infty$. Note that $h^{(1)}$ denotes the derivative of $h$ with respect to $\vtheta$. By the
  smoothness of $h(\vz, \vtheta)$ in $\vtheta$, any perturbation of size
  $\lVert \vu \rVert \le \epsilon$ to $\vtheta$ defines a mapping
  $\sM_{\vtheta \to \vtheta + \vu}$ from a point in $\bR^K$ to another point
  $ \bR^K$. For example,
  $\sM_{\vtheta \to \vtheta + \vu} h(\vz, \vtheta) = h(\vz, \vtheta+\vu)$.
  Here, by using assumption~\ref{app:item:h-smoothness} and a Taylor's expansion,
  we bound the distance between the mapping input and output by
  \begin{equation*}
    \lVert h(\vz, \vtheta + \vu) - h(\vz, \vtheta) \rVert \le \epsilon b < \frac{ \inf_{\vtheta \in \mTheta} \{ \text{shortest distance between origin and convex hull of } \sP_{\vtheta}  \}  }{2}.
  \end{equation*}
  Now, since (i) $\vzero$ is in the interior of the convex hull of $\sP_{\vtheta}$;
  (ii) the mapped output $\sM_{\vtheta \rightarrow \vtheta + \vu} \sP_{\vtheta}$
  are no further than $\tfrac{1}{2} \times$ shortest distance between origin
  and convex hull of $\sP_{\vtheta}$ from their respective input, consequently the
  interior of the convex hull of
  $\sM_{\vtheta \rightarrow \vtheta + \vu} \sP_{\vtheta}$ contains $\vzero$.
  Hence we have proved that $\vtheta^{(r)} \in \operatorname{int}(\mTheta_{B;n})$
  implies
  $\sB_{ \epsilon } ( \vtheta^{(r)} ) \subseteq \operatorname{int} ( \mTheta_{B;n} )$.
  Consequently,
  \begin{equation*}
    \sum_{n = L}^\infty \bP \left \{ \sB_{ \epsilon  } ( \vtheta^{(r)} ) \nsubseteq \operatorname{int} ( \mTheta_{B;n} ) \right \} \le \sum_{n = L}^\infty \bP \left \{ \vtheta^{(r)} \notin \operatorname{int}(\mTheta_{B;n}) \right \}.
  \end{equation*}
  By following steps in Lemma~\ref{app:thm:single-point-supported}, we have
  $\sum_{r=1}^R \sum_{n = L}^\infty \bP \left \{ \vtheta^{(r)} \notin \operatorname{int}(\mTheta_{B;n}) \right \} < \infty$.
  Thus, we have proven our required result for the countable $h(\sZ,\vtheta)$
  case. In the uncountable $h(\sZ,\vtheta)$ case, let
  $\sG_{\vtheta} = \sG_{1,\vtheta} \times \ldots \times \sG_{L,\vtheta}$. Then,
  by assumption~\ref{app:item:h-location}, $\vzero$ is contained in the interior of the convex hull
  of any $\vg \in \sG_{\vtheta}$ for any $\vtheta \in \mTheta$. In this case, we
  set
  \begin{equation*}
    0 < \epsilon < \frac{ \inf_{\vtheta \in \mTheta, \vg \in \sG_{\vtheta}} \{ \text{shortest distance between origin and convex hull of } \vg  \}  }{2 b}.
  \end{equation*}
  Following similar arguments to the countable $h(\sZ,\vtheta)$ case, any
  perturbation $\lVert \vu \rVert \le \epsilon$ to $\vtheta$ defines a mapping
  $\sM_{\vtheta \rightarrow \vtheta + \vu}$ from a point in $\bR^K$ to another
  point $ \bR^K$. In fact, we have
  \begin{equation*}
    \lVert h(\vz, \vtheta + \vu) - h(\vz, \vtheta) \rVert
    \le \epsilon b
    < \frac{ \inf_{\vtheta \in \mTheta, \vg \in \sG_{\vtheta}} \{ \text{shortest distance between origin and convex hull of } \vg \}}{2}.
  \end{equation*}
  Now, consider any $\vg \in \sG_{\vtheta}$. If (i) $\vzero$ is in the interior
  of the convex hull of $\vg$; (ii) the mapped output
  $\sM_{\vtheta \rightarrow \vtheta + \vu} \vg$ are no further than
  $\tfrac{1}{2} \times$ shortest distance between origin and convex hull of
  $\vg$ from their respective input, then the interior of the convex hull of
  $\sM_{\vtheta \rightarrow \vtheta + \vu} \vg$ contains $\vzero$. Hence, we
  have proved that $\vtheta^{(r)} \in \operatorname{int}(\mTheta_{B;n})$ implies
  $\sB_{ \epsilon } ( \vtheta^{(r)} ) \subseteq \operatorname{int} ( \mTheta_{B;n} )$.
  By applying Lemma~\ref{app:thm:single-point-supported} in a similar way to the
  countable case, we have proven our result for the uncountable
  $h(\sZ, \vtheta)$ case.
\end{proof}
From the previous lemma, we can deduce that $\mTheta \subseteq \operatorname{int}(\mTheta_{B;n})$ for
a sufficiently large $n$. Under this paradigm, we establish the smoothness of
each site over the entire parameter space $\mTheta$. Details are provided in the
next result.


\begin{theorem}
  \label{app:thm:continuous-diff}
  Assume \ref{app:item:theta-bound} to \ref{app:item:h-smoothness} hold. Then, for a sufficiently large $n$,
  \begin{equation*}
    \left| \diffp{\phi_i (\vtheta; \sD_n)}{\theta_a, \theta_b, \theta_c, \theta_d} \right| < \infty,
    \quad \emph{for all} \
    (a, b, c, d) \in \{1, \dots, p \}^4, \
    i = 1, \ldots, n, \
    \vtheta \in \mTheta.
  \end{equation*}
\end{theorem}
\begin{proof}
  From Lemma~\ref{app:thm:lemma-truth-supported}, we know that each $\vtheta$ in
  $\mTheta$ corresponds to a set of finite and non-zero $w_{1}, \ldots, w_{n}$
  when $n$ is sufficiently large. Then, $\phi_i (\vtheta; \sD_n) = - \log w_i$
  is at least fourth-differentiable with respect to $w_{i}$, and
  $w_{i} = n^{-1}(1 + \vlambda^{\top} \vh_{i})^{-1}$ is at least
  fourth-differentiable with respect to $\vlambda^{\top} \vh_{i}$. From
  \ref{app:item:h-smoothness}, the constraint $\vh_{i}$ is at least
  fourth-differentiable, so it is sufficient to show that
  \begin{equation*}
    \left| \diffp{\lambda_k}{\theta_a, \theta_b, \theta_c, \theta_d} \right| < \infty,
    \quad \text{for all} \
    (a, b, c, d) \in \{1, \dots, p \}^4, \
    k = 1, \ldots, K, \
    \vtheta \in \mTheta,
  \end{equation*}
  is differentiable for each entries of $\vlambda = (\lambda_1, \ldots, \lambda_K)$. Let
  $\widetilde{\mH} = [w_{1} \vh_{1}, \ldots, w_{n} \vh_{n}]$. We apply the implicit
  function theorem on
  \begin{equation}
    \sum_{i=1}^{n} \frac{\vh_{i}}{1 + \vlambda^{\top} \vh_{i}} = \vzero.
  \end{equation}
  to get the first-order derivatives,
  \begin{equation}
    \label{app:eq:deriv-proof-1}
    \widetilde{\mH} \widetilde{\mH}^{\top}
    \diffp{\vlambda}{\theta_a}
    = \sum_{i=1}^{n} w_{i} \left( n^{-1} \mI
    - w_{i} \vh_{i} \vlambda^{\top}\right) \diffp{\vh_{i}}{\theta_a}.
  \end{equation}
  The right-hand side is finite, and we only need to show that
  $\widetilde{\mH} \widetilde{\mH}^{\top}$ is invertible.
  Assumption~\ref{app:item:h-location} implicitly implies that
  $\vh_{1}, \ldots, \vh_{n}$ spans $\R^{K}$ when $n$ is sufficiently large
  \footnote{Any $\vh_{1}, \ldots, \vh_{n}$ that does not span $\R^{K}$ can never
    contain $\vzero_{K}$ in the interior of its convex hull. More formally, we
    show that $\{\vh_1, \ldots, \vh_n\}$ spans $\bR^K$ if $\vzero_{K}$ is in the
    interior of the convex hull of $\{\vh_1, \ldots, \vh_n\}$. Since
    $\vzero_{K}$ is in the interior of the convex hull of
    $\{\vh_1, \ldots, \vh_n\}$, there exists a ball of radius $\xi$ centered at
    $\vzero_{K}$ such that the ball is a subset of the interior of the convex
    hull of $\{\vh_{1}, \ldots, \vh_{n}\}$. Now, any point in $\vv \in \bR^K$,
    we can write
    $\vv = \frac{\lVert \vv \rVert}{\xi} (\xi \frac{\vv}{\lVert \vv \rVert})$.
    Since $\xi \frac{\vv}{\lVert \vv \rVert}$ is a surface of the ball, then
    $\xi \frac{\vv}{\lVert \vv \rVert} = w_1 \vh_1 + \dots + w_n \vh_n$ where
    each $w_i$ is strictly positive. Then,
    $\vv = \frac{\lVert \vv \rVert}{\xi} (w_1 \vh_1 + \dots + w_n \vh_n)$.
    Hence, $\{\vh_1, \ldots \vh_n\}$ spans $\bR^K$.} Therefore,
  $\rank(\widetilde{\mH}) = \rank(\widetilde{\mH}\widetilde{\mH}^\top) = K$,
  $\widetilde{\mH}\widetilde{\mH}^{\top}$ is invertible, and
  $\partial \vlambda / \partial \theta_a$ is finite.

  For second-order derivatives, differentiate both side
  of~\eqref{app:eq:deriv-proof-1} with respect to $\theta_b$ and we get
  \begin{equation}
    \label{app:eq:deriv-proof-2}
    \widetilde{\mH}\widetilde{\mH}^{\top}
    \diffp{\vlambda}{\theta_a,\theta_b}
    =
    - \diffp{\widetilde{\mH}\widetilde{\mH}^{\top}}{\theta_b} \diffp{\vlambda}{\theta_a}
    + \diffp{}{\theta_b} \left\{ \sum_{i=1}^{n} w_{i} \left( n^{-1} \mI
    - w_{i} \vh_{i} \vlambda^{\top}\right) \diffp{\vh_{i}}{\theta_a} \right\}.
  \end{equation}
  As pointed out in \eqref{app:eq:deriv-proof-1} the matrix
  $\widetilde{\mH}\widetilde{\mH}^{\top}$ is invertible. The right-hand side
  only involves $\diffp{\vh_{i}}{\theta_a, \theta_b}$ and first-order
  derivatives of $w_{i}$ and $\vlambda$ with respect to $\vtheta$. The
  derivative $\diffp{\vh_{i}}{\theta_a, \theta_b}$ is finite by the smoothness
  assumption on $h$. The first-order derivative of $\vlambda$ has been proven to
  be finite and, as a direct consequence, the first-order derivative of $w_i$ is
  also finite. Therefore, $\diffp{\vlambda}{\theta_a, \theta_b}$ is also finite.

  We can differentiate both side of \eqref{app:eq:deriv-proof-2} to obtain the
  expression for third-order derivatives, and differentiate the subsequent
  expression again to obtain the fourth-order derivatives. The highest order
  derivatives in these expressions are
  $\diffp{\vh_{i}}{\theta_a, \theta_b, \theta_c, \theta_d}$ and the third-order
  derivatives of $w_{i}$ and $\vlambda$ with respect to $\vtheta$. These are all
  finite, and thus $\diffp{\vlambda}{\theta_a, \theta_b, \theta_c, \theta_d}$ is
  also finite.
\end{proof}
In the rest of the proof, we let $\phi^{(t)}(\vtheta)$ denote the $t$-order
tensor derivative of the scalar-valued function $\phi$. For example,
$\phi^{(2)}(\vtheta)$ denotes the Hessian matrix
$\tfrac{\partial \phi(\vtheta)}{\partial \vtheta \partial \vtheta^\top}$. Let
$\lVert \cdot \rVert$ denote the Euclidean norm of a vector, matrix, third or
fourth-order tensor. An obvious consequence of
Lemma~\ref{app:thm:lemma-truth-supported} and Theorem~\ref{app:thm:continuous-diff} is
the following upper bound on the respective Euclidean distance of the higher
derivatives:
\begin{corollary}
  \label{app:thm:useful-bound}
  Assume~\ref{app:item:theta-bound} to~\ref{app:item:prior-smoothness} hold. For a
  sufficiently large $n$, we have positive constants $K_2$, $K_3$, and $K_4$
  such that for all $i = 0, \ldots, n$:
  \begin{equation*}
    \sup_{\vtheta \in \mTheta} \left \lVert \phi_i^{(2)}(\vtheta)  \right \rVert \le K_2, \quad
    \sup_{\vtheta \in \mTheta} \left \lVert \phi_i^{(3)}(\vtheta) \right \rVert \le K_3, \quad
    \sup_{\vtheta \in \mTheta} \left \lVert \phi_i^{(4)}(\vtheta) \right \rVert \le K_4,
  \end{equation*}
  where $\lVert \cdot \rVert$ denotes tensor Euclidean norms.
\end{corollary}
Throughout the rest of the work, we will denote $K_M = \max \{ K_2, K_3, K_4 \}$.

\subsection{Asymptotic equivalence to Newton-Raphson updates}
\label{app:sec:equiv-nr-update}
In this segment of the proof, we works towards a result that establishes the
asymptotic equivalence ($n \rightarrow \infty$) of one update cycle of the EP
algorithm and one update cycle of the Newton-Raphson algorithm for
\emph{maximum-a-posteriori} (MAP) computation.

More concretely, at the $t$-th iteration, denote the linear shift and precision
of the $i$-th site approximations as $\vr_{i}^{t}$ and $\mQ_{i}^{t}$
respectively. We will show that an EP update of the global approximation at the
$i$-th iteration, $\vr^{t} = \sum_{i=0}^{n} \vr_{i}^{t}$ and
$\mQ^{t} = \sum_{i=0}^{n} \mQ_{i}^{t}$, is asymptotically equivalent to
performing an update on $\vmu^{t} = (\mQ^{t})^{-1} \vr^{t}$ with the following
Newton-Raphson update:
\begin{equation}
  \label{app:eq:one-step-newton-update}
  \vmu^{t+1} =  \vmu^{t} - \{ \psi^{(2)}(\vmu^{t}) \}^{-1} \psi^{(1)}(\vmu^{t}),
\end{equation}
then setting $\mQ^{t+1} = \psi^{(2)}(\vmu^{t})$ and
$\vr^{t+1} = \mQ^{t+1} \vmu^{t} - \psi^{(1)}(\vmu^{t})$.

To do so, we need to analyse the tilted distribution of EP. We start with the
cavity distribution at the $t$-th iteration, which we parameterise in the
following form
\begin{equation*}
  q_{-i} (\vtheta) = (2 \pi)^{-p/2} \det (\vbeta )^{1/2} \exp \left [ - \tfrac{1}{2} \left \{ \vtheta - (\vmu^{t} - \vbeta^{-1} \vdelta)  \right \}^\top \vbeta \left \{ \vtheta - (\vmu^{t} - \vbeta^{-1} \vdelta)  \right \}   \right ].
\end{equation*}
where $\vbeta = \sum_{j \neq i} \mQ_{j}^{t}$ is a $p \times p$ cavity precision
matrix, and the $p$-dimensional cavity mean, $\vmu^{t} - \vbeta^{-1} \vdelta$,
is expressed as a small deviation $\vbeta^{-1} \vdelta$ from an initial mean
estimate $\vmu^{t}$. Similarly, the linear-shift parameter of the cavity
distribution can be expressed as $\vbeta \vmu^{t} - \vdelta$. Both $\vbeta$ and
$\vdelta$ differ across sites but we omit the subscript~$i$ to avoid cumbersome
notation \footnote{In fact, $\vdelta := \vdelta_i = \vbeta \vmu^t - (\vr^t - \vr_i^t)$}. For each site $i= 0,\ldots, n$ , the corresponding tilted distribution
is
\begin{equation*}
  q_{\backslash i} ( \vtheta ) \propto w_i (\vtheta; \sD_{n}) q_{-i} (\vtheta),
\end{equation*}
with mean $\vmu_{\backslash i} = \E_{\vtheta \sim q_{\backslash i}} [\vtheta]$
and precision
$\mQ_{\backslash i} = \E_{\vtheta \sim q_{\backslash i}} [(\vtheta - \vmu_{\backslash i})(\vtheta - \vmu_{\backslash i})^{\top}]^{-1}$.
We use $w_0$ to denote the prior $p(\vtheta)$.

For two compatible square matrices $\mA$ and $\mC$, we write $\mA \ge \mC$ if
$\mA - \mC$ is positive semidefinite and $\mA \le \mC$ if $\mA - \mC$ is
negative semidefinite. For a positive/negative semidefinite matrix $\mM$, we
denote the $j$-th eigenvalue as $\eigen_j(\mM)$. Also, we denote the maximum and
minimum eigenvalues as $\max \eigen (\mM)$ and $\min \eigen (\mM)$ respectively.

The proofs of Theorems~\ref{app:thm:cavity-asymptote} and
\ref{app:thm:asymptotically-newton-raphson} presented in this section are similar to
Theorems 1 and 2 \cite{dehaene18expectation}. For completeness, we re-present
the proofs here in the context of Bayesian empirical likelihood approximation,
with special care taken to consider the multivariable property of $\vtheta$. The
next theorem establishes the limiting behaviour of the tilted distribution for
the site approximations as the precision norm diverges. Our result is stated
under the condition $\min \eigen (\vbeta) \rightarrow \infty$ and
$\max \eigen (\vbeta) / \min \eigen (\vbeta) \rightarrow b \in (0, \infty)$
which implies that $\lVert \vbeta \rVert \rightarrow \infty$, where $\vbeta$
denotes the cavity precision matrix.

\begin{theorem}
  \label{app:thm:cavity-asymptote}
  Assume~\ref{app:item:theta-bound} to~\ref{app:item:prior-smoothness} hold. Consider the cavity distribution with precision
  matrix $\vbeta$ and mean vector $\vmu^{t} - \vbeta^{-1} \vdelta$ for any
  $\vmu^{t} \in \mTheta$ and $\vdelta \in \bR^p$. Then, for any sufficiently
  large $n$, the mean of the $i$-th tilted distribution has the limiting
  behaviour in terms of $\min \eigen (\vbeta) \rightarrow \infty$ and
  $\max \eigen (\vbeta) / \min \eigen (\vbeta) \rightarrow b \in (0, \infty)$:
  \begin{equation*}
    \lVert \vmu_{\backslash i} - \vmu^{t} \rVert = O \left (  \Delta_{\vr;i}  / \min \eigen (\vbeta)  \right ),
  \end{equation*}
  where $\Delta_{\vr;i} =  \lVert \phi_i^{(1)} (\vmu^{t}) + \vdelta \rVert$. Moreover, the site approximation's parameters have the limiting behaviour:
  \begin{equation*}
    \left \lVert \mQ_{i}^{t+1} - \phi_i^{(2)} (\vmu^{t})  \right \rVert = O \left \{ ( 1 + \Delta_{\vr;i} )/ \min \eigen (\vbeta) \right \}
  \end{equation*}
  and
  \begin{equation*}
    \left \lVert \vr_{i}^{t+1} + \phi_i^{(1)} (\vmu^{t}) - \mQ_{i}^{t+1}\vmu^{t} \right \rVert = O \left ( \Delta_{r;i}^2  \{ \min \eigen (\vbeta) \}^{-2} +  p^2 K_M  \{ \min \eigen (\vbeta) \}^{-1} \right )
  \end{equation*}
  for some constant $K_M > 0$.
\end{theorem}
\begin{proof}
  We first show that the tilted distribution is strictly log-concave, i.e., showing
  $-\diffp[2]{}{\vtheta}\log q_{\backslash i}(\vtheta) = \phi_i^{(2)} (\vtheta) + \vbeta$
  is positive definite. Following Theorem~\ref{app:thm:continuous-diff}, there
  exists a positive definite matrix $\mB$ such that
  $\phi_i^{(2)} (\vtheta) - \phi_i^{(2)} (\vmu^{t}) + \mB$ is positive
  definite for all $\vtheta \in \mTheta$, or equivalently
  $\phi_i^{(2)} (\vtheta) > \phi_i^{(2)} (\vmu^{t}) - \mB$. One example is
  $\mB = (K_2 + 0.1) \mI$. Substitute this inequality into the second-order
  derivative and we obtain a lower bound
  \begin{equation*}
    -\diffp[2]{}{\vtheta}\log q_{\backslash i}(\vtheta) = \phi_i^{(2)} (\vtheta) + \vbeta \ge \phi_i^{(2)} (\vmu^{t}) + \vbeta - \mB.
  \end{equation*}
  For a $\vbeta$ with sufficiently large minimum eigenvalue, we can ensure
  that the lower bound is positive definite, and thus
  $\phi_i^{(2)} (\vtheta) + \vbeta$ is positive definite and the tilted
  distribution is log-concave.

  For a $\vbeta$ with sufficiently large minimum eigenvalue, we can then
  invoke the Brascamp-Lieb inequality:
  \begin{align}
    \bE_{\vtheta \sim q_{\backslash i}} \{ (\vtheta - \vmu_{\backslash i}) (\vtheta - \vmu_{\backslash i})^\top  \} & \le \bE_{\vtheta \sim q_{\backslash i}} \left [ \left \{ - \frac{\partial^2 \log q_{\backslash i} (\vtheta) }{\partial \vtheta \partial \vtheta^\top} \right \}^{-1} \right ] \nonumber \\
                                                                                                                    & =  \bE_{\vtheta \sim q_{\backslash i}} \left [ \left \{ \phi_i^{(2)} (\vtheta) + \vbeta \right \}^{-1} \right ] \nonumber                                                               \\
                                                                                                                    & \le  \left \{ \phi_i^{(2)} (\vmu^{t})  + \vbeta - \mB \right \}^{-1}
    \label{app:eq:tilted-covariance-bound}
  \end{align}
  Unless specified otherwise, all expectations for the remainder of this proof
  are taken with respect to $q_{\backslash i}(\vtheta)$. Before proceeding, we
  derive another bound on
  $\bE \{ \lVert \vtheta - \vmu_{\backslash i} \rVert^2 \}$ that will be
  useful for further derivation. This is done by taking the trace on both
  sides of the previous inequality:
  \begin{align}
    \label{app:CoarseBound}
    \bE \{ \lVert \vtheta - \vmu_{\backslash i} \rVert^2 \}
     & \le \tr \left (  \left \{ \phi_i^{(2)} (\vmu^{t})  + \vbeta - \mB \right \}^{-1} \right ) \nonumber            \\
     & = \sum_{j} \eigen_j \left (  \left \{ \phi_i^{(2)} (\vmu^{t})  + \vbeta - \mB \right \}^{-1} \right )\nonumber \\
     & \le p \max \eigen \left (  \left \{ \phi_i^{(2)} (\vmu^{t})  + \vbeta - \mB \right \}^{-1} \right ) \nonumber  \\
     & \le p \left \lVert \left \{ \phi_i^{(2)} (\vmu^{t})  + \vbeta - \mB \right \}^{-1} \right \rVert.
  \end{align}
  Back to \eqref{app:eq:tilted-covariance-bound}, we now deal with $\left \{ \phi_i^{(2)} (\vtheta ) + \vbeta \right \}^{-1}$. By Taylor's expansion, we get (for any $\vu \in \bR^p$):
  \begin{align*}
     & \vu^\top \left \{ \phi_i^{(2)} (\vtheta ) + \vbeta \right \}^{-1} \vu                                                                                                                                                                                                                                                                                                                      \\
     & \approx \vu^\top \left \{ \bE \{ \phi_i^{(2)} (\vtheta ) \} + \vbeta \right \}^{-1} \vu +  \frac{ \partial \vu^\top \left \{  \bE \phi_i^{(2)} (\vtheta ) + \vbeta \right \}^{-1} \vu  }{ \partial \vech( \phi_i^{(2)} (\vtheta ))  } \vech \left \{ \phi_i^{(2)} (\vtheta ) - \bE \{ \phi_i^{(2)} (\vtheta ) \} \right \}                                                                 \\
     & + \tfrac{1}{2} \vech \left \{ \phi_i^{(2)} (\vtheta ) - \bE \{ \phi_i^{(2)} (\vtheta ) \} \right \}^\top  \frac{ \partial^2 \vu^\top \left \{ \bE \phi_i^{(2)} (\vtheta ) + \vbeta \right \}^{-1} \vu  }{ \partial \vech( \phi_i^{(2)} (\vtheta )) \partial \vech( \phi_i^{(2)} (\vtheta ))^\top  } \vech \left \{ \phi_i^{(2)} (\vtheta ) - \bE \{ \phi_i^{(2)} (\vtheta ) \} \right \} , \\
     & = \vu^\top \mR  \vu - \left [ \vech \left ( \mR \vu \vu^\top \mR \right )^\top  \mK \right ]^\top \vech \left \{ \phi_i^{(2)} (\vtheta ) - \bE \{ \phi_i^{(2)} (\vtheta ) \} \right \}                                                                                                                                                                                                     \\
     & + \vech \left \{ \phi_i^{(2)} (\vtheta ) - \bE \{ \phi_i^{(2)} (\vtheta ) \} \right \}^\top \mK^\top  \mR  \otimes ( \mR \vu \vu^\top \mR ) \mK \vech \left \{ \phi_i^{(2)} (\vtheta ) - \bE \{ \phi_i^{(2)} (\vtheta ) \} \right \},
  \end{align*}
  where the first approximate equality is an application of Taylor's expansion
  with respect to $\vech( \phi_i^{(2)} (\vtheta ) )$ about its
  expectation, the second equality follows by evaluation the first and second
  order derivatives and denoting $\mK$ as a $p^2$ by $p(p+1)/2$ duplication
  matrix \citep{magnus88matrix} and
  $\mR = \left \{ \bE \{ \phi_i^{(2)} (\vtheta ) \} + \vbeta \right \}^{-1}$.
  By taking expectation on both sides of the inequality, we have
  \begin{align*}
     & \vu^\top \bE \left \{ \phi_i^{(2)} (\vtheta ) + \vbeta \right \}^{-1} \vu                                                                                                                                                                                     \\
     & \approx  \vu^\top \mR  \vu +\tr \left [   \Var \left \{ \vech \{ \phi_i^{(2)} (\vtheta ) \} \right\}    \mK^\top  \mR  \otimes ( \mR \vu \vu^\top \mR ) \mK  \right ],                                                                                        \\
     & \le \vu^\top \mR  \vu + \bE \left \lVert \vech \left \{ \phi_i^{(2)} (\vtheta ) - \bE \{ \phi_i^{(2)} (\vtheta ) \} \right \} \right \rVert^2  \tr \left [      \mK^\top  \mR  \otimes ( \mR \vu \vu^\top \mR )  \mK  \right ],                               \\
     & \le \vu^\top \mR  \vu + K_M \bE \lVert \vtheta - \vmu_{\backslash i}  \rVert^2 \tr \left [      \mK^\top  \mR   \otimes ( \mR \vu \vu^\top  \mR  )  \mK  \right ]                                                                                             \\
     & \le \vu^\top \mR  \vu + K_M \bE \lVert \vtheta - \vmu_{\backslash i}  \rVert^2 \tr \left ( \mK     \mK^\top    \right )    \tr \left ( \mR \otimes \mR \vu \vu^\top  \mR \right )                                                                             \\
     & \le \vu^\top \mR  \vu +  K_M \bE \lVert \vtheta - \vmu_{\backslash i}  \rVert^2 \tr \left ( \mK     \mK^\top \right ) \tr \left ( \vu     \vu^\top \right ) \tr \left (  \widetilde{\mR} \right )^3                                                           \\
     & \le \vu^\top \mR  \vu +  p^3 K_M \bE \lVert \vtheta - \vmu_{\backslash i}  \rVert^2 \tr \left ( \mK     \mK^\top \right ) \tr \left ( \vu     \vu^\top \right )  \left \lVert \left \{ \phi_i^{(2)} (\vmu^{t})  + \vbeta - \mB \right \}^{-1} \right \rVert^3 \\
     & \le \vu^\top \mR  \vu + p^4 K_M \tr \left ( \mK     \mK^\top \right ) \left \lVert \left \{ \phi_i^{(2)} (\vmu^{t})  + \vbeta - \mB \right \}^{-1} \right \rVert^4 \lVert \vu \rVert^2.
  \end{align*}
  The second inequality follows from swapping expectation and trace operators,
  and the trace inequality for two positive semidefinite matrices
  $\tr (\mL \mJ) \le \tr(\mL) \tr(\mJ)$. The third inequality follows from an
  element-wise second Taylor's expansion of
  $\vech ( \phi_i^{(2)} (\vtheta) )$ followed by taking expectation. The
  fourth equality follows from trace inequality. The fifth inequality follows
  from the trace inequality, $\tr (\mL \otimes \mJ) = \tr(\mL) \tr (\mJ)$, and
  $\widetilde{\mR} = \left \{ \phi_i^{(2)} (\vmu^{t}) + \vbeta - \mB \right \}^{-1}$,
  $\tr(\mR) \le \tr(\widetilde{\mR})$. The sixth inequality follows from the
  fact that for any square matrix $\mJ$ of order $p$, we have
  $\tr(\mJ) \le p \lVert \mJ \rVert$. The seventh inequality follows from
  (\ref{app:CoarseBound}). Consequently, from \eqref{app:eq:tilted-covariance-bound},
  we have
  \begin{equation*}
    \mQ_{\backslash i}^{-1}
    \le \mR +  p^{4} K_M \tr \left ( \mK     \mK^\top \right ) \left \lVert \left \{ \phi_i^{(2)} (\vmu^{t})  + \vbeta - \mB \right \}^{-4} \right \rVert  \mI.
  \end{equation*}
  By taking inverse on both sides, we have
  \begin{align}
    \label{app:eq:LowerBoundCritical}
    \mQ_{\backslash i}
     & \ge \left [ \mR + p^{4} K_M \tr \left ( \mK \mK^\top \right ) \left \lVert \left \{ \phi_i^{(2)} (\vmu^{t})  + \vbeta - \mB \right \}^{-1} \right \rVert^4  \mI \right ]^{-1}                              \nonumber \\
     & \ge \bE \phi_i^{(2)} (\vtheta) + \vbeta - p^{4} K_M \tr \left ( \mK \mK^\top \right ) \left \lVert \left \{ \phi_i^{(2)} (\vmu^{t})
    + \vbeta - \mB \right \}^{-1} \right \rVert^4   \{ \bE \phi_i^{(2)} (\vtheta) + \vbeta \}^2 \nonumber                                                                                                                   \\
     & \ge \bE \phi_i^{(2)} (\vtheta) + \vbeta - p^{4} K_M \tr \left ( \mK \mK^\top \right ) \left \lVert \left \{ \phi_i^{(2)} (\vmu^{t})
    + \vbeta - \mB \right \}^{-1} \right \rVert^4 \left \{ \max \eigen ( \bE \phi_i^{(2)} (\vtheta) + \vbeta ) \right \}^2 \mI \nonumber                                                                                    \\
     & \ge \bE \phi_i^{(2)} (\vtheta) + \vbeta - p^{4} K_M \tr \left ( \mK \mK^\top \right ) \left \lVert \left \{ \phi_i^{(2)} (\vmu^{t})
    + \vbeta - \mB \right \}^{-1} \right \rVert^4  \max\eigen \left \{ \left ( \phi_i^{(2)} (\vmu^{t}) + \vbeta + \mB \right ) \right \}^2 \mI \nonumber                                                                    \\
     & \ge \bE \phi_i^{(2)} (\vtheta) + \vbeta - p^{4} K_M \tr \left ( \mK \mK^\top \right ) \left \lVert \left \{ \phi_i^{(2)} (\vmu^{t})
    + \vbeta - \mB \right \}^{-1} \right \rVert^4  \lVert ( \phi_i^{(2)} (\vmu^{t}) + \vbeta + \mB ) \rVert^2 \mI \nonumber                                                                                                 \\
     & \ge  \bE \phi_i^{(2)} (\vtheta) + \vbeta - p^{7} K_M \tr \left ( \mK \mK^\top \right ) \times \frac{ \{ \max \eigen (\vbeta)
      + \max \eigen (\phi_i (\vmu^{t}) + \mB ) \}^2 }{\{ \min \eigen (\vbeta) + \min \eigen (\phi_i (\vmu^{t}) - \mB ) \}^4 } \mI
  \end{align}
  The second inequality follows from noting the diagonal entries of
  $\mK \mK^\top$ are either 0 or 1 and hence its trace is nonnegative
  \citep{magnus88matrix}. Then, by applying the result that if $\mC$ and $\mD$
  are positive definite matrices such that $\mC^{-1} \mD$ is also positive
  definite, then $\mC + \mD$ is also positive definite and in fact
  $(\mC + \mD)^{-1} \ge \mC^{-1} - \mC^{-2} \mD$. The third inequality follows
  from the result
  $\mA^2 \le \max \eigen ( \mA^2 ) \mI = \max \eigen ( \mA )^2 \mI $ for any psd
  matrix $\mA$. The fourth inequality follows from
  $ \bE \phi_i^{(2)} (\vtheta) \le \phi_i^{(2)} (\vmu^{t}) + \mB$. The fifth
  inequality follows from $\max \eigen ( \mA )^2 \le \lVert \mA \rVert^2$ for
  any psd $\mA$. The sixth inequality follows from
  $\lVert \mA^{-1} \rVert \le \sqrt{p}/ \min \eigen(\mA)$ and
  $\lVert \mA \rVert \le \sqrt{p} \max \eigen(\mA)$ for any psd $\mA$, and
  followed by Weyl's inequality for eigenvalues. Note these argument hold under
  a sufficiently large $\min \eigen (\vbeta)$. We analyse the term
  $\bE \phi_i^{(2)} (\vtheta)$. Note that for any $\vu \in \bR^p$, we have
  \begin{align*}
    \vu^\top \phi_i^{(2)} (\vtheta) \vu \ge (\le) \sum_{j} \sum_\ell u_j u_\ell \left \{ \left [ \phi_i^{(2)} (\vmu_{\backslash i}) \right ]_{j \ell} + \frac{\partial \left [ \phi_i^{(2)} (\vmu_{\backslash i}) \right ]_{j \ell} }{ \partial \vtheta}^\top (\vtheta - \vmu_{\backslash i}) - (+) (K_M/2)  \lVert \vtheta - \vmu_{\backslash i} \rVert^2  \right \}
  \end{align*}
  By taking expectation on both sides, we have
  \begin{align}
    \label{app:eq:TakeExpectation}
    \bE \phi_i^{(2)} (\vtheta) & \ge (\le) \phi_i^{(2)} (\vmu_{\backslash i}) - (+) (K_M/2) \bE\lVert \vtheta - \vmu_{\backslash i} \rVert^2 \mI                                                     \nonumber \\
                               & \ge (\le) \phi_i^{(2)} (\vmu_{\backslash i}) -(+) (p K_M/2) \left \lVert \left \{ \phi_i^{(2)} (\vmu^{t})  + \vbeta - \mB \right \}^{-1} \right \rVert \mI ,
  \end{align}
  where the second inequality follows from (\ref{app:CoarseBound}). Moreover, we have
  \begin{align*}
    \mQ_{\backslash i}
     & \le \bE \{  \phi_i^{(2)}(\vtheta) + \vbeta \}                                                                                                               \\
     & \le \vbeta +  \phi_i^{(2)} (\vmu_{\backslash i}) + (p K_M/2) \left \lVert \left \{ \phi_i^{(2)} (\vmu^{t})  +\vbeta - \mB \right \}^{-1} \right \rVert \mI,
  \end{align*}
  where the first inequality follows from considering the strict log concavity
  of the tilted distribution and an inequality in \citet{saumard14logconcavity}
  \footnote{one line after (10.25)}, and the second inequality follows from
  (\ref{app:eq:TakeExpectation}). Since
  $\mQ_{i}^{t+1} = \mQ_{\backslash i} - \vbeta$
  , we have
  \begin{equation}
    \label{app:eq:UpperBoundCritical}
    \mQ_{i}^{t+1} - \phi_i^{(2)} (\vmu_{\backslash i}) \le (p K_M/2) \left \lVert \left \{ \phi_i^{(2)} (\vmu^{t})  + \vbeta - \mB \right \}^{-1} \right \rVert \mI
  \end{equation}
  By combining the bounds in (\ref{app:eq:LowerBoundCritical}), (\ref{app:eq:TakeExpectation}), \ref{app:eq:UpperBoundCritical}, and then noting that
  \begin{equation*}
    \frac{ \{ \max \eigen (\vbeta) + \max \eigen (\phi_i (\vmu^{t}) + \mB ) \}^2 }{\{ \min \eigen (\vbeta) + \min \eigen (\phi_i (\vmu^{t}) - \mB ) \}^4 } = O( \{\min \eigen (\vbeta) \}^{-2} ),
  \end{equation*}
  and
  \begin{equation*}
    \left \lVert \left \{ \phi_i^{(2)} (\vmu^{t}) - \mB + \vbeta \right \}^{-1} \right \rVert = O( \{\min \eigen (\vbeta) \}^{-1} ),
  \end{equation*}
  we obtain
  \begin{equation*}
    \left \lvert \eigen_j \left \{ \mQ_{i}^{t+1} - \phi_i^{(2)}(\vmu_{\backslash i}) \right \} \right \rvert = O (pK_M \{\min \eigen (\vbeta) \}^{-1}),
  \end{equation*}
  and hence
  \begin{equation*}
    \lVert \mQ_{i}^{t+1} - \phi_i^{(2)}(\vmu_{\backslash i})  \rVert = O\left ( p^2 K_M \left \lVert \left \{ \phi_i^{(2)} (\vmu^{t})  + \vbeta - \mB \right \}^{-1} \right \rVert \right ) = O( \{\min \eigen (\vbeta) \}^{-1} ).
  \end{equation*}
  To show the convergence of $\vmu_{\backslash i}$ to $\vmu^{t}$, we note that the score equation of the tilted distribution (expected first-order derivative of unnormalised tilted density equals zero):
  \begin{equation*}
    \vbeta (\vmu_{\backslash i} - \vmu^{t}) = - \vdelta - \bE\{ \phi_i^{(1)} (\vtheta) \}
  \end{equation*}
  Taylor's expansion of $\phi_i^{(1)} (\vtheta)$ about $\vmu^{t}$ gives us:
  \begin{equation*}
    \bE\{ \phi_i^{(1)} (\vtheta) \} = \phi_i^{(1)}(\vmu^{t}) + \phi_i^{(2)} ( \widetilde{\vmu}) ( \vmu_{\backslash i} - \vmu^{t} ),
  \end{equation*}
  where $\widetilde{\vmu}$ is a convex combination of $\vtheta$ and
  $\vmu^{t}$. By substituting the expansion back
  into the score equation, we have
  \begin{equation*}
    ( \phi_i^{(2)} ( \widetilde{\vmu}) + \vbeta ) (\vmu_{\backslash i} - \vmu^{t}) = - \vdelta - \phi_i^{(1)}(\vmu^{t}).
  \end{equation*}
  By pre-multiplying $(\vmu_{\backslash i} - \vmu^{t})^\top$ on both sides, we have
  \begin{align*}
    (\vmu_{\backslash i} - \vmu^{t})^\top  ( \phi_i^{(2)} ( \widetilde{\vmu}) + \vbeta ) (\vmu_{\backslash i} - \vmu^{t}) & = - (\vmu_{\backslash i} - \vmu^{t})^\top ( \vdelta + \phi_i^{(1)}(\vmu^{t}) )                  \\
                                                                                                                          & \le \lVert \vmu_{\backslash i} - \vmu^{t} \rVert \lVert \vdelta + \phi_i^{(1)}(\vmu^{t}) \rVert
  \end{align*}
  Also, we have the lower bound
  \begin{equation*}
    (\vmu_{\backslash i} - \vmu^{t})^\top  ( \phi_i^{(2)} ( \widetilde{\vmu}) + \vbeta ) (\vmu_{\backslash i} - \vmu^{t}) \ge \lVert \vmu_{\backslash i} - \vmu^{t}  \rVert^2 \min \eigen (\vbeta)
  \end{equation*}
  Consequently, we have
  \begin{equation*}
    \lVert \vmu_{\backslash i} - \vmu^{t} \rVert \le \frac{ \Delta_{r;i}  }{  \min \eigen (\vbeta) }.
  \end{equation*}
  Hence, we have
  \begin{align*}
    \lVert \mQ_{i}^{t+1} - \phi_i^{(2)}(\vmu^{t})  \rVert
     & \le \lVert \mQ_{i}^{t+1} - \phi_i^{(2)}(\vmu_{\backslash i})  \rVert +  \lVert \phi_i^{(2)}(\vmu_{\backslash i}) - \phi_i^{(2)}(\vmu^{t})   \rVert           \\
     & \le \lVert \mQ_{i}^{t+1} - \phi_i^{(2)}(\vmu_{\backslash i})  \rVert + K_M \frac{ \lVert \vdelta + \phi_i^{(1)}(\vmu^{t}) \rVert  }{  \min \eigen (\vbeta) } \\
     & = O \left \{ ( 1 + \Delta_{\vr;i} ) / \min \eigen (\vbeta) \right \}
  \end{align*}
  where the first inequality is an application of triangle inequality and the second inequality follows from a first order Taylor's expansion for an element of $\phi_i^{(2)}(\vmu^{t})$ about $\vmu_{\backslash i}$ and followed by bound on $\lVert \vmu_{\backslash i} - \vmu^{t} \rVert$. Next, we analyse the convergence of the site approximation's linear shift $\vr_{i}^{t+1}$. Recall that we have the score equation:
  \begin{equation*}
    \vbeta (\vmu_{\backslash i} - \vmu^{t}) = - \vdelta - \bE\{ \phi_i^{(1)} (\vtheta) \}
  \end{equation*}
  Note that the tilted distribution has precision matrix $\vbeta + \mQ_{i}^{t+1}$ and $\vbeta \vmu^{t} - \vdelta + \vr_{i}^{t+1}$. A multivariate normal distribution $g$ with the same mean and covariance matrix as the tilted distribution has density of the form:
  \begin{equation*}
    g(\vtheta) \propto \exp \left [ - \tfrac{1}{2} \{ \vtheta - (\vbeta + \mQ_{i}^{t+1})^{-1} ( \vbeta \vmu^{t} - \vdelta + \vr_{i}^{t+1} )  \}^\top (\vbeta + \mQ_{i}^{t+1}) \{ \vtheta - (\vbeta + \mQ_{i}^{t+1})^{-1} ( \vbeta \vmu^{t} - \vdelta + \vr_{i}^{t+1} )  \}  \right ]
  \end{equation*}
  The score equation corresponding to $g$ is:
  \begin{equation*}
    \mQ_{i}^{t+1} \vmu_{\backslash i} - \vr_{i}^{t+1} + \vbeta \vmu_{\backslash i} - \vbeta \vmu^{t} + \vdelta = \vzero.
  \end{equation*}
  By noting that $ \bE\{ \phi_i^{(1)} (\vtheta) \} = - \vdelta - \vbeta (\vmu_{\backslash i} - \vmu^{t})$, we have
  \begin{equation}
    \label{app:eq:rI}
    \vr_{i}^{t+1} = \mQ_{i}^{t+1} \vmu_{\backslash i} - \bE\{ \phi_i^{(1)} (\vtheta) \}.
  \end{equation}
  By considering each element of $\phi_i^{(1)} (\vtheta)$, we have the Taylor's expansion
  \begin{equation*}
    \left [ \phi_i^{(1)} (\vtheta)  \right ]_j =  \left [ \phi_i^{(1)} (\vmu^{t})  \right ]_j + \frac{\partial \left [ \phi_i^{(1)} (\vmu^{t})  \right ]_j}{\partial \vtheta}^\top (\vtheta - \vmu^{t}) + \frac{1}{2} (\vtheta - \vmu^{t})^\top \frac{\partial^2 \left [ \phi_i^{(1)} (\widetilde{\vmu}_{0j} )  \right ]_j}{\partial \vtheta \partial \vtheta^\top} (\vtheta - \vmu^{t}).
  \end{equation*}
  Rearrange the terms and we get
  \begin{align*}
     & \left \lvert \bE  	\left [ \phi_i^{(1)} (\vtheta)  \right ]_j -  \left [ \phi_i^{(1)} (\vmu^{t})  \right ]_j - \frac{\partial \left [ \phi_i^{(1)} (\vmu^{t})  \right ]_j}{\partial \vtheta}^\top (\vmu_{\backslash i} - \vmu^{t}) \right \rvert \\
     & = O\left( K_M \bE \lVert  \vtheta - \vmu^{t} \rVert^2 \right)                                                                                                                                                                                    \\
     & =  O\left( K_M  \tr ( \mQ_{\backslash i}^{-1}                                                                                                                                                                                                    
    ) + K_M \lVert \vmu_{\backslash i} - \vmu^{t} \rVert^2 \right)                                                                                                                                                                                      \\
     & \le O\left( K_M  \tr ( \widetilde{\mR} + p K_M^2 \tr( \mK \mK^\top) \lVert \widetilde{\mR} \rVert^{-4} \mI ) + K_M \lVert \vmu_{\backslash i} - \vmu^{t} \rVert^2 \right)                                                                        \\
     & \le O( p K_M  \{ \min \eigen (\vbeta) \}^{-1} ),
  \end{align*}
  where the first equality follows from the upper bound the second-order
  derivative of
  $\frac{\partial^2 \left [ \phi_i^{(1)} (\widetilde{\vmu}_{0j} ) \right ]_j}{\partial \vtheta \partial \vtheta^\top}$.
  The second equality follows from the result that
  $ \bE \lVert \vtheta - \vmu^{t} \rVert^2 = \tr( \mQ_{\backslash i}^{-1}
    ) + \lVert \vmu_{\backslash i} - \vmu^{t} \rVert^2$.
  The third inequality is an application of an earlier derived result for
  $\mQ_{\backslash i}^{-1}
  $. The fourth inequality follows from an
  earlier derived result for $\tr(\widetilde{\mR})$. Collecting individual
  elements back into a vector, we have
  \begin{equation}
    \label{app:eq:phiPrime}
    \lVert \bE\{ \phi_i^{(1)} (\vtheta) \} - \phi_i^{(1)} (\vmu^{t}) - \phi_i^{(2)} (\vmu^{t}) ( \vmu_{\backslash i} - \vmu^{t} )  \rVert \le O( p^2 K_M  \{ \min \eigen (\vbeta) \}^{-1} )
  \end{equation}
  Hence, from previously-derived result, we have
  \begin{align*}
    \lVert \vr_{i}^{t+1} - \mQ_{i}^{t+1} \vmu^{t} + \phi_i^{(1)} (\vmu^{t}) \rVert
     & \le \lVert \{ \mQ_{i}^{t+1} - \phi_i^{(2)} (\vmu^{t}) \} \{ \vmu_{\backslash i} - \vmu^{t} \} \rVert + O( p^2 K_M  \{ \min \eigen (\vbeta) \}^{-1} )     \\
     & \le  \lVert  \mQ_{i}^{t+1} - \phi_i^{(2)} (\vmu^{t}) \rVert \lVert \vmu_{\backslash i} - \vmu^{t} \rVert + O( p^2 K_M  \{ \min \eigen (\vbeta) \}^{-1} ) \\
     & \le O \left ( \Delta_{r;i}^2  \{ \min \eigen (\vbeta) \}^{-2} +  p^2 K_M  \{ \min \eigen (\vbeta) \}^{-1} \right )
  \end{align*}
  where the first inequality follows from substituting the result
  (\ref{app:eq:phiPrime}) into (\ref{app:eq:rI}). The second inequality follows is
  an application of Cauchy-Schwarz. The third inequality follows by noting
  that
  $\lVert \mQ_{i}^{t+1} - \phi_i^{(2)} (\vmu^{t}) \rVert = O \left ( \{ \min \eigen (\vbeta) \}^{-1} \right )$
  and
  $\lVert \vmu_{\backslash i} - \vmu^{t} \rVert \le \frac{ \lVert \vdelta + \phi_i^{(1)}(\vmu^{t}) \rVert }{ \min \eigen (\vbeta) }$.
\end{proof}

We now state and prove the theorem that in the large data limit
($n \rightarrow \infty$), one cycle of an EP update is equivalent to one cycle
of a Newton Raphson update. Note that $\vr_{0}^{t}$ and $\mQ_{0}^{t}$ denote the
linear shift and precision matrix of the prior site.

\begin{theorem}
  \label{app:thm:asymptotically-newton-raphson}
  Assume~\ref{app:item:theta-bound} to~\ref{app:item:prior-smoothness} hold. Consider
  the EPEL Gaussian approximation at iteration $t$, $\{\vr^{t}_i \}_{i=0}^n$ and
  $\{\mQ_i^{t} \}_{i=0}^n$ for the linear-shift and precision of the site
  approximations. The global approximation mean
  $\vmu^{t} = \{ \sum_i \mQ_i^{t} \}^{-1} \sum_i \vr^{t}_i$ is a fixed vector
  and the global precision is $\mQ^{t} = \sum_{i=0}^n \mQ_i^{t}$. Moreover,
  assume that (i)
  $\min_i \min \eigen ( \sum_{j \neq i} \mQ_{j}^{t}) = pn + O_{p}(1)$ with some
  positive constant $p$; (ii) the means of the tilted distributions are not too
  far apart, i.e.,
  $\sum_i \lVert \vr^{t}_i - \mQ_i^{t} \vmu^{t} + \phi_i^{(1)} (\vmu^{t}) \rVert = O_p(n)$.
  Then, the new global linear shift and new global precision after one EP cycle
  has the asymptotic behaviour:
  \begin{equation*}
    \left \lVert \vr^{t+1} +  \psi^{(1)} (\vmu^{t}) - \mQ^{t+1}  \vmu^{t} \right \rVert = O_p(1)
  \end{equation*}
  and
  \begin{equation*}
    \left \lVert \mQ^{t+1} - \psi^{(2)} (\vmu^{t}) \right \rVert = O_p(1)
  \end{equation*}
  where $\psi (\cdot) = \sum_{i=0}^n \phi_i (\cdot)$.
\end{theorem}
\begin{proof}
  Consider the update for likelihood site $i$. The corresponding cavity
  distribution has precision $\mQ^{t} - \mQ_i^{t}$, linear shift
  $\vr^{t} - \vr^{t}_i$, and mean
  $(\mQ^{t} - \mQ_i^{t})^{-1} (\vr^{t} - \vr^{t}_i)$, where
  $\vr^{t} = \sum_{i=0}^n \vr^{t}_i$. By Theorem \ref{app:thm:cavity-asymptote} with
  $\vdelta = \vdelta_i = \vr^{t}_i - \mQ_i^{t} \vmu^{t}$ and
  $\Delta_{r;i} = \lVert \phi_i(\vmu^t) + \vdelta_i \rVert$, we have
  \begin{equation}
    \label{app:eq:IndividualSiteQ}
    \lVert \mQ_i^{t+1} - \phi_i^{(2)} (\vmu^{t}) \rVert
    =  O_p\left ( \frac{1 + \Delta_{\vr;i} }{  \min \eigen ( \sum_{j \neq i} \mQ_{j}) } \right )
    \le O_p\left ( \frac{1 + \Delta_{\vr;i} }{  \min_i \min \eigen ( \sum_{j \neq i} \mQ_{j}) } \right ).
  \end{equation}
  Hence, the new global precision matrix deviates from $\psi^{(2)} (\vmu^{t})$ by order:
  \begin{align*}
    \lVert \mQ^{t+1} - \psi^{(2)} (\vmu^{t}) \rVert
     & \le \sum_{i=0}^n \lVert \mQ_i^{t+1} - \phi_i^{(2)} (\vmu^{t}) \rVert                                                 \\
     & \le O_p\left ( \frac{ (n + 1)  \sum_i \Delta_{\vr;i} }{ \min_i \min \eigen ( \sum_{j \neq i} \mQ_{j}^{t}) } \right ) \\
     & = O_p(1),
  \end{align*}
  where the first inequality is an application of triangle inequality, the
  second inequality follows from \eqref{app:eq:IndividualSiteQ}, and the third
  equality follows from (i) and (ii). For the new linear shift, we have
  \begin{align*}
    \lVert \vr^{t+1} + \psi^{(1)} (\vmu^{t}) - \mQ^{t+1} \vmu^{t} \rVert
     & \le \sum_{i=0}^n \lVert \vr_i^{t+1} + \phi_i^{(1)} (\vmu^{t}) - \mQ_i^{t+1} \vmu^{t} \rVert                          \\
     & \le O_p \left \{ \frac{(n +1) p^2 K_m }{ \min_i \min \eigen \left ( \sum_{j \neq i} \mQ_{j}^{t} \right ) } \right \} \\
     & = O_p (1)
  \end{align*}
  where the first inequality is an application of triangle inequality, the
  second inequality follows from Theorem \ref{app:thm:cavity-asymptote}, and the
  third equality follows from (i) and (ii).
\end{proof}
\paragraph{Remarks.} This results shows that an EP cycle is asymptotically
equivalent to performing a Newton-Raphson update as specified in
\eqref{app:eq:one-step-newton-update} as $n$ diverges.


\subsection{Algorithmic convergence}
\label{app:sec:epel-convergence}
The EP algorithm involves an arbitrary number of iterations (depending on the
user-specified termination condition). Here, a relevant questions pertains to
the criteria required to guarantee algorithmic convergence for the EP algorithm
\footnote{Convergence to a fixed point}. Since EP is related to the
Newton-Raphson's method as shown in the previous section, we study the
conditions for which the Newton-Raphson's attains algorithmic convergence and
use the results to identify the conditions that guarantee algorithmic
convergence for EP.

Let $\widehat{\vtheta} = \argmax p_{\EL} (\vtheta \mid \sD_n)$ denote the
posterior mode and $\vtheta^0$ denote the initial approximation of the posterior
mode. The updated approximate mode based on the Newton-Raphson algorithm is:
\begin{equation*}
  \vtheta^1 = \vtheta^0 - \{ \psi^{(2)} ( \vtheta^0 ) \}^{-1} \psi^{(1)} (  \vtheta^0 ).
\end{equation*}

\begin{theorem}
  \label{app:thm:stable-region-newton-raphson}
  Assume~\ref{app:item:theta-bound} to~\ref{app:item:prior-smoothness} hold and $ n^{-1} \psi^{(2)} (\widehat{\vtheta}) $
  converges in probability to a positive definite matrix $\mV_{\vtheta^\star}$
  with finite eigenvalues. Consider the updated approximate
  $\widehat{\vtheta}_{1}$ based on the Newton-Raphson algorithm. Then, there
  exists a sequence of stable regions $\Delta_{\operatorname{NR};n}$ such that
  if
  $\lVert \vtheta^0 - \widehat{\vtheta} \rVert \le \Delta_{\operatorname{NR};n}$,
  then
  \begin{equation*}
    \lVert \vtheta^1 - \widehat{\vtheta} \rVert \le \lVert \vtheta^0 - \widehat{\vtheta} \rVert.
  \end{equation*}
  Moreover, $\Delta_{\operatorname{NR};n} = O_p(1)$ and is bounded away from
  $0$.
\end{theorem}
\begin{proof}
  We perform Taylor's expansion of $\psi^{(1)}_j (\widehat{\vtheta})$ about $ \vtheta^0$:
  \begin{align*}
    0
     & = \psi_j^{(1)} (\widehat{\vtheta})                                                                                                                                                                                                                                                                                             \\
     & = \psi_j^{(1)} ( \vtheta^0 ) + \tfrac{\partial}{\partial \vtheta} \psi_j^{(1)} ( \vtheta^0 )^\top ( \widehat{\vtheta} - \vtheta^0 ) +  \tfrac{1}{2} ( \widehat{\vtheta} - \vtheta^0 )^\top \frac{\partial^2 \psi_j^{(1)} ( \widetilde{\vtheta}_j )}{\partial \vtheta \partial \vtheta^\top} ( \widehat{\vtheta} - \vtheta^0 ),
  \end{align*}
  where $\widetilde{\vtheta}_j  = u\vtheta^0 + (1-u) \widehat{\vtheta}$ and $0 < u < 1$. For each $j$, we have
  \begin{equation*}
    \frac{\partial^2 \psi_j^{(1)} ( \widetilde{\vtheta}_j )}{\partial \vtheta \partial \vtheta^\top} \le (n+1) K_3 \mI.
  \end{equation*}
  Consequently, we have
  \begin{equation*}
    \left \{  \psi^{(2)} ( \vtheta^0 )  \right \}^{-1} \psi^{(1)} (\widehat{\vtheta}) =  \left \{  \psi^{(2)} ( \vtheta^0 )  \right \}^{-1} \psi^{(1)} ( \vtheta^0 ) +  ( \widehat{\vtheta} - \vtheta^0 ) + \left \{  \psi^{(2)} ( \vtheta^0 )  \right \}^{-1} \sR_3 (\vtheta^0, \widehat{\vtheta}),
  \end{equation*}
  where
  \begin{align*}
    \lVert \sR_3 (\vtheta^0, \widehat{\vtheta}) \rVert \le \tfrac{p(n+1)K_3 }{2}  \lVert \widehat{\vtheta} - \vtheta^0 \rVert^2.
  \end{align*}
  Since $ \left \lVert \left \{  \psi^{(2)} ( \vtheta^0 )  \right \}^{-1}  \sR_3 (\vtheta^0, \widehat{\vtheta}) \right \rVert= \left \lVert \vtheta^0 -  \left \{  \psi^{(2)} ( \vtheta^0 )  \right \}^{-1} \psi^{(1)} ( \vtheta^0 ) - \widehat{\vtheta} \right \rVert = \left \lVert \vtheta^1  - \widehat{\vtheta} \right \rVert $, we have
  \begin{equation}
    \label{app:eq:IntermediateConditionNRforStability}
    \left \lVert \vtheta^1  - \widehat{\vtheta} \right \rVert \le \frac{p(n + 1)K_3 }{2}  \left \lVert \left \{  \psi^{(2)} ( \vtheta^0 )  \right \}^{-1} \right \rVert  \lVert \widehat{\vtheta} - \vtheta^0 \rVert^2.
  \end{equation}
  The above result facilitate us to verify the following claim: a sufficient condition on the initial value $\vtheta^0$ such that $\{ \lVert \widehat{\vtheta} - \vtheta^t \rVert \}_{t \ge 1}$ is a decreasing function in $t$ is
  \begin{equation*}
    \lVert \widehat{\vtheta} - \vtheta^0 \rVert \le \frac{ \lVert \psi^{(2)} ( \widehat{\vtheta}) \rVert  }{ (Ep/2+1) (n+1) p^2 K_3 }.
  \end{equation*}
  where $E = \sup_{\vtheta \in \mTheta} \lVert \psi^{(2)} ( \vtheta) \rVert \left \lVert \left \{  \psi^{(2)} ( \vtheta )  \right \}^{-1} \right \rVert$. By Lemma~\ref{app:thm:lemma-truth-supported}, $\mTheta \subset \mTheta_{B;n}$ finitely often and hence $E$ is bounded away from 0 and $\infty$ for a sufficiently large $n$. To verify the claim, we use Corollary \ref{app:thm:useful-bound} to obtain the bound
  \begin{equation*}
    \lVert \psi^{(2)} ( \widehat{\vtheta}) - \psi^{(2)} ( \vtheta^0) \rVert \le (n + 1) p^2 K_3 \lVert \widehat{\vtheta} - \vtheta^0 \rVert.
  \end{equation*}
  Hence by triangle inequality, we have
  \begin{equation*}
    \lVert \psi^{(2)} ( \widehat{\vtheta})  \rVert \le \lVert \psi^{(2)} ( \widehat{\vtheta}) - \psi^{(2)} ( \vtheta^0) \rVert + \lVert \psi^{(2)} ( \vtheta^0) \rVert \le (n + 1) p^2 K_3 \lVert \widehat{\vtheta} - \vtheta^0 \rVert + \lVert \psi^{(2)} ( \vtheta^0) \rVert
  \end{equation*}
  and consequently our claimed sufficient condition implies
  \begin{equation*}
    \lVert \widehat{\vtheta} - \vtheta^0 \rVert \le \frac{  (n + 1) p^2 K_3 \lVert \widehat{\vtheta} - \vtheta^0 \rVert + \lVert \psi^{(2)} ( \vtheta^0) \rVert  }{ (Ep/2+1) (n + 1) p^2 K_3 }.
  \end{equation*}
  or equivalently
  \begin{equation*}
    \lVert \widehat{\vtheta} - \vtheta^0 \rVert \le \frac{  2 \lVert \psi^{(2)} ( \vtheta^0) \rVert  }{ Ep^3 (n + 1) K_3 }.
  \end{equation*}
  By combining with equation (\ref{app:eq:IntermediateConditionNRforStability}), we have
  \begin{equation*}
    \frac{p(n + 1)K_3 }{2}  \left \lVert \left \{  \psi^{(2)} ( \vtheta^0 )  \right \}^{-1} \right \rVert  \lVert \widehat{\vtheta} - \vtheta^0 \rVert^2 \le  \frac{  \lVert \psi^{(2)} ( \vtheta^0) \rVert \left \lVert \left \{  \psi^{(2)} ( \vtheta^0 )  \right \}^{-1} \right \rVert    }{Ep^2}  \lVert \vtheta^0 - \widehat{\vtheta} \rVert \le  \lVert \vtheta^0 - \widehat{\vtheta} \rVert
  \end{equation*}
  where the last inequality follows by noting that $\lVert \psi^{(2)} ( \vtheta^0) \rVert \left \lVert \left \{  \psi^{(2)} ( \vtheta^0 )  \right \}^{-1} \right \rVert  \le E$ and $p \ge 1$. Lastly, we examine the asymptotic behaviour of the upper bound
  \begin{equation*}
    \frac{ \lVert \psi^{(2)} ( \widehat{\vtheta}) \rVert  }{ (Ep/2+1) (n + 1) p^2 K_3 }
  \end{equation*}
  Based on the theorem's assumed condition, we have $\lVert \psi^{(2)} ( \widehat{\vtheta}) \rVert = pn + O_p(1)$ and hence the upper bound is $O_p(1)$ and bounded away from $0$.
\end{proof}

\paragraph{Remarks.} The previous theorem shows that if the initial point is
from a set of points that are sufficiently close to $\widehat{\vtheta}$, then
the updated value $\vtheta^1$ belong to the same neighbourhood as the initial.
Furthermore, the sequence of solution $\{ \vtheta^t \}_{t \ge 1}$ will never get
further away from $\widehat{\vtheta}$.

Next, we investigate the convergence of the sequence of EP solutions
$\{ (\vr^t, \mQ^t) \}_{t \ge 1}$ towards a fixed point.
\begin{theorem}
  \label{app:thm:stable-region}
  Assume~\ref{app:item:theta-bound} to~\ref{app:item:prior-smoothness} hold. Consider
  the EPEL Gaussian initializations $\{\vr_i^0 \}_{i=0}^{n}$ and
  $\{\mQ_i^0 \}_{i=0}^{n}$ that satisfy
  \begin{equation*}
    n \max_i \lVert \mQ_i^0 \widehat{\vtheta} - \phi_i^{(1)} (\widehat{\vtheta})  - \vr_i^0    \rVert = \Delta_{\vr}^0
  \end{equation*}
  and
  \begin{equation*}
    n \max_i \lVert  \phi_i^{(2)} (\widehat{\vtheta})  - \mQ_i^0  \rVert = \Delta_{\vbeta}^0,
  \end{equation*}
  where $ \Delta_{\vr}^0 = \sqrt{n}$ and $\Delta_{\vbeta}^0 = \tfrac{1}{2} \lVert \psi^{(2)} (\widehat{\vtheta}) \rVert$. Then, for every $t=1,2,\ldots$, there exists a $\Delta_{\vr}^t$ and $\Delta_{\vbeta}^t$ such that
  \begin{equation*}
    \lVert \mQ_i^t \widehat{\vtheta} - \phi_i^{(1)} (\widehat{\vtheta}) - \vr_i^t \rVert \le n^{-1} \Delta_{\vr}^t
  \end{equation*}
  and
  \begin{equation*}
    \lVert \phi_i^{(2)} (\widehat{\vtheta})  - \mQ_i^t  \rVert \le n^{-1} \Delta_{\vbeta}^t,
  \end{equation*}
  where $\{\vr_i^t \}_{i=0}^{n}$ and $\{\mQ_i^t \}_{i=0}^{n}$ denote the EP
  approximation parameters after the $t$-th cycle. Moreover, if
  $ n^{-1} \psi^{(2)} (\widehat{\vtheta}) $ converges in probability to a
  constant positive definite matrix $\mV_{\vtheta^\star}$ with finite
  eigenvalues, $\sum_{i=0}^n \lVert \phi_i^{(1)} (\widehat{\vtheta}) \rVert$,
  and $\sum_{i=0}^n \lVert \phi_i^{(2)} (\widehat{\vtheta}) \rVert = O_p(n)$,
  then
  \begin{equation*}
    \Delta_{\vr} ^t= O_p(1) \quad \emph{and} \quad \Delta_{\vbeta}^t = O_p(1) \quad \emph{for all} \quad t=2,3 \ldots.
  \end{equation*}
\end{theorem}
\begin{proof}
  Let $\mQ^t = \sum_{i=0}^n \mQ_i^t$ and $\vr^t = \sum_{i=0}^n \vr_i^t$. Denote $\Delta_{\vbeta}^t = n \max_i \lVert \phi_i^{(2)}(\widehat{\vtheta}) -  \mQ_i^t \rVert$ and $\Delta_{\vr}^t = n \max_i \lVert \vr_i^t - \mQ_i^t \widehat{\vtheta} + \phi_i^{(1)} (\widehat{\vtheta}) \rVert$. We begin by observing that
  \begin{equation*}
    \phi_i^{(2)} (\widehat{\vtheta}) - \mQ_i^t \le (\Delta_{\vbeta}^t / n) \mI
  \end{equation*}
  and hence
  \begin{equation*}
    (\mQ^t)^{-1} \le \left \{ \psi^{(2)}(\widehat{\vtheta}) -  \Delta_{\vbeta}^t \mI \right \}^{-1}
  \end{equation*}
  and consequently
  \begin{equation*}
    \frac{1}{\min \eigen ( \mQ^t )} \le \left \lVert \left \{ \psi^{(2)}(\widehat{\vtheta}) -  \Delta_{\vbeta}^t \mI \right \}^{-1} \right \rVert
  \end{equation*}
  Let $\vmu^{t} = (\mQ^t)^{-1} \vr^t$. Following Theorem
  \ref{app:thm:cavity-asymptote} with offset $\vdelta = \vzero$, we can write
  \begin{align*}
    \left \lVert \mQ_{i}^{t+1} - \phi_i^{(2)} (\vmu^{t})  \right \rVert & = O \left ( ( 1 + \lVert \phi_i^{(1)}(\vmu^{t})   \rVert ) \left \lVert \left \{ \psi^{(2)}(\widehat{\vtheta}) -  \Delta_{\vbeta}^t \mI \right \}^{-1} \right \rVert \right )
  \end{align*}
  and
  \begin{align*}
    \left \lVert \vr_i^{t+1} + \phi_i^{(1)} (\vmu^{t}) - \mQ_{i}^t\vmu^{t} \right \rVert & = O( p^2 K_M  \{ \min \eigen (\mQ^{t}) \}^{-1} )                                                                                        \\
                                                                                         & = O \left ( p^2 K_M \left \lVert \left \{ \psi^{(2)}(\widehat{\vtheta}) -  \Delta_{\vbeta}^t \mI \right \}^{-1} \right \rVert \right ).
  \end{align*}
  The strategy in the remaining section of our proof is to modify the above
  two equalities such that $\vmu^{t}$ is replaced with $\widehat{\vtheta}$.
  Now, we have
  \begin{align*}
    \lVert \vr^{t} - \mQ^t \widehat{\vtheta} \rVert & = \lVert \vr^{t} - \mQ^t \widehat{\vtheta} + \psi^{(1)} (\widehat{\vtheta}) \rVert                      \\
                                                    & \le \sum_{i=0}^n \lVert \vr_i^{t} - \mQ_i^t \widehat{\vtheta} + \phi_i^{(1)} (\widehat{\vtheta}) \rVert \\
                                                    & \le \Delta_{\vr}^t
  \end{align*}
  On the other hand
  \begin{align*}
    \lVert \vr^{t} - \mQ^t \widehat{\vtheta} \rVert & \ge \min \eigen ( \mQ^t ) \lVert \widehat{\vtheta} - \vmu^{t} \rVert
  \end{align*}
  and hence $\lVert \widehat{\vtheta} - \vmu^{t} \rVert \le \Delta_{\vr}^t \{\min \eigen ( \mQ^t )\}^{-1} \le \Delta_{\vr}^t \left \lVert \left \{ \psi^{(2)}(\widehat{\vtheta}) -  \Delta_{\vbeta}^t \mI \right \}^{-1} \right \rVert$. By considering a first-order Taylor's expansion of $\phi_i^{(2)} (\vmu^{t})$ about $\widehat{\vtheta}$, we have
  \begin{align*}
    \left \lVert \mQ_{i}^{t+1} - \phi_i^{(2)} ( \widehat{\vtheta} )  \right \rVert
     & \le K_M \Delta_{\vr}^t \left \lVert \left \{ \psi^{(2)}(\widehat{\vtheta})
    - \Delta_{\vbeta}^t \mI \right \}^{-1} \right \rVert                                                                                                                                                 \\
     & + O \left \{ \left (  1 + \lVert \phi_i^{(1)}  ( \vmu^{t} ) \rVert  \right ) \left \lVert \left \{ \psi^{(2)}(\widehat{\vtheta}) -  \Delta_{\vbeta}^t \mI \right \}^{-1} \right \rVert \right \}.
  \end{align*}
  By Taylor's expansion about $\widehat{\vtheta}$, we have
  \begin{align*}
    \sum_{i=0}^n \lVert \phi_i^{(1)} (\vmu^{t}) \rVert
     & \le \sum_{i=0}^n \lVert \phi_i^{(1)} ( \widehat{\vtheta} ) \rVert + \lVert \widehat{\vtheta} - \vmu^{t} \rVert \sum_{i=0}^n \lVert \phi_i^{(2)} ( \widehat{\vtheta} ) \rVert + \frac{nK_3 \lVert \widehat{\vtheta} - \vmu^{t} \rVert^2}{2}                                                                                                                                                                     \\
     & \le \sum_{i=0}^n \lVert \phi_i^{(1)} ( \widehat{\vtheta} ) \rVert +   \Delta_{\vr}^t \left \lVert \left \{ \psi^{(2)}(\widehat{\vtheta}) -  \Delta_{\vbeta}^t \mI \right \}^{-1} \right \rVert \sum_{i=0}^n \lVert \phi_i^{(2)} ( \widehat{\vtheta} ) \rVert + \frac{nK_3 (\Delta_{\vr}^t)^2  }{2} \left \lVert \left \{ \psi^{(2)}(\widehat{\vtheta}) -  \Delta_{\vbeta}^t \mI \right \}^{-1} \right \rVert^2
  \end{align*}
  and hence
    {\scriptsize
      \begin{align}
        \label{app:eq:BoundforQt}
         & \left \lVert \mQ^{t+1} - \psi^{(2)} ( \widehat{\vtheta} )  \right \rVert \nonumber                                                                                                                                                                                                                                                                                                                                                                                                                                                                                     \\
         & \le (n + 1) K_M \Delta_{\vr}^t \left \lVert \left \{ \psi^{(2)}(\widehat{\vtheta}) -  \Delta_{\vbeta}^t \mI \right \}^{-1} \right \rVert \nonumber                                                                                                                                                                                                                                                                                                                                                                                                                     \\
         & + O \left \{ K_M \left (  n + \sum_{i=0}^n \lVert \phi_i^{(1)} ( \widehat{\vtheta} ) \rVert +   \Delta_{\vr}^t \left \lVert \left \{ \psi^{(2)}(\widehat{\vtheta}) -  \Delta_{\vbeta}^t \mI \right \}^{-1} \right \rVert \sum_{i=0}^n \lVert \phi_i^{(2)} ( \widehat{\vtheta} ) \rVert + \frac{nK_3 (\Delta_{\vr}^t)^2  }{2} \left \lVert \left \{ \psi^{(2)}(\widehat{\vtheta}) -  \Delta_{\vbeta}^t \mI \right \}^{-1} \right \rVert^2 \right ) \left \lVert \left \{ \psi^{(2)}(\widehat{\vtheta}) -  \Delta_{\vbeta}^t \mI \right \}^{-1} \right \rVert \right \}.
      \end{align} }
  Moreover, by adding and subtracting the appropriate terms, we get
  \begin{equation*}
    \vr_i^{t+1} - \mQ_i^{t+1} \widehat{\vtheta} + \phi_i^{(1)} (\widehat{\vtheta}) = \phi_i^{(1)} (\widehat{\vtheta}) - \phi_i^{(1)} (\vmu^{t}) + \mQ_i^{t+1} \vmu^{t} - \mQ_i^{t+1}  \widehat{\vtheta}  + \vr_i^{t+1} - \mQ_i^{t+1} \vmu^{t} + \phi_i^{(1)}(\vmu^{t})
  \end{equation*}
  which is used in the following bound
  \begin{align}
    \label{app:eq:BoundforRt}
    \lVert \vr^{t+1}  - \mQ^{t+1} \widehat{\vtheta} + \psi^{(1)} (\widehat{\vtheta}) \rVert
     & \le  \sum_{i=0}^n \lVert \vr_i^{t+1}  - \mQ_i^{t+1} \widehat{\vtheta} + \phi_i^{(1)} (\widehat{\vtheta}) \rVert\nonumber                                                                                                                                                                                                   \\
     & \le  \sum_{i=0}^n \left[ \lVert R_3(\vmu^{t}, \widehat{\vtheta}) \rVert  + \lVert \mQ_i^{t+1} - \phi_i^{(2)} (\widehat{\vtheta}) \rVert \lVert \vmu^{t} - \widehat{\vtheta} \rVert + \lVert \vr_i^{t+1}  - \mQ_i^{t+1} \vmu^{t} + \phi_i^{(1)} (\vmu^{t}) \rVert  \right ] \nonumber                                       \\
     & \le\frac{(n+1)\sqrt{p} K_3}{2}  (\Delta_{\vr}^t)^2 \left \lVert \left \{ \psi^{(2)}(\widehat{\vtheta}) -  \Delta_{\vbeta}^t \mI \right \}^{-1} \right \rVert^2 + \Delta_{\vr}^t \Delta_{\vbeta}^{t+1} \left \lVert \left \{ \psi^{(2)}(\widehat{\vtheta}) -  \Delta_{\vbeta}^t \mI \right \}^{-1} \right \rVert  \nonumber \\
     & + O \left \{ n p^2 K_M \left \lVert \left \{ \psi^{(2)}(\widehat{\vtheta}) -  \Delta_{\vbeta}^t \mI \right \}^{-1} \right \rVert  \right \},
  \end{align}
  where the first inequality is an application of triangle inequality. The
  second inequality follows from the preceding equation, the triangle
  inequality, and a Taylor's expansion:
  $\phi_i^{(1)}(\vmu^{t}) = \phi_i^{(1)}( \widehat{\vtheta}) + \phi_i^{(2)} (\vmu^{t} - \widehat{\vtheta}) - R_3(\vmu^{t}, \widehat{\vtheta})$
  where $R_3$ is a $p$-dimensional column vector with $j$-th entry equals to
  $-\tfrac{1}{2} (\vmu^{t} - \widehat{\vtheta})^\top \tfrac{\partial}{\partial \vtheta \partial \vtheta^\top} [ \phi_i^{(1)} (\widetilde{\vmu}^{t}) ]_j (\vmu^{t} - \widehat{\vtheta})$
  and $\widetilde{\vmu}^{t}$ is a point along the line joining $\vmu^{t}$ and
  $\widehat{\vtheta}$. The third inequality follows from
  $\lVert R_3(\vmu^{t}, \widehat{\vtheta}) \rVert \le \tfrac{1}{2} \sqrt{p} K_3 \lVert \vmu^{t} - \widehat{\vtheta} \rVert^2$.
  By considering the initial values such that
  $\Delta_{\vbeta}^0 = \tfrac{1}{2} \lVert \psi^{(2)} (\widehat{\vtheta}) \rVert $
  and $\Delta_{\vr}^0 = \sqrt{n}$, we have
  \begin{align*}
    \left \lVert \left \{ n^{-1} \psi^{(2)} (\widehat{\vtheta}) -  n^{-1} \Delta_{\vbeta}^0 \mI \right \}^{-1} \right \rVert \xrightarrow{\bP^\star} \left \lVert \left \{ \mV_{\vtheta^\star} -  \tfrac{1}{2} \lVert \mV_{\vtheta^\star} \rVert  \mI \right \}^{-1} \right \rVert \in (0, \infty)
  \end{align*}
  and hence $	\left \lVert \left \{  \psi^{(2)} (\widehat{\vtheta}) -  \Delta_{\vbeta}^0 \mI \right \}^{-1} \right \rVert = O_p(n^{-1})$. By substituting the convergence results into (\ref{app:eq:BoundforQt}) and (\ref{app:eq:BoundforRt}) and specifying $t = 0$, we have
  \begin{align*}
    \Delta_{\vbeta}^1 & = O_p(\sqrt{n}) + O_p(1) = O_p(\sqrt{n})
  \end{align*}
  and
  \begin{align*}
    \Delta_{\vr}^1 & = O_p(1) + O_p(1) + O_p(1) = O_p(1)
  \end{align*}
  By considering  $\Delta_{\vbeta}^1 =  O_p(\sqrt{n}) $ and $\Delta_{\vr}^1 = O_p(1)$, we have
  \begin{align*}
    \left \lVert \left \{ n^{-1} \psi^{(2)} (\widehat{\vtheta}) -  n^{-1} \Delta_{\vbeta}^1 \mI \right \}^{-1} \right \rVert \xrightarrow{\bP^\star} \left \lVert \left \{ \mV_{\vtheta^\star}  \right \}^{-1} \right \rVert \in (0, \infty)
  \end{align*}
  and hence
  $ \left \lVert \left \{ \psi^{(2)} (\widehat{\vtheta}) - \Delta_{\vbeta}^1 \mI \right \}^{-1} \right \rVert = O_p(n^{-1})$.
  By substituting the convergence results into (\ref{app:eq:BoundforQt}) and
  (\ref{app:eq:BoundforRt}) and specifying $t = 1$, we have
  \begin{align*}
    \Delta_{\vbeta}^2 & = O_p(1) + O_p(1) = O_p(1)
  \end{align*}
  and
  \begin{align*}
    \Delta_{\vr}^2 & = O_p(n^{-1}) + O_p(n^{-1}) + O_p(1) = O_p(1).
  \end{align*}
  By continued iteration, we obtain $\Delta_{\vbeta}^t = O_p(1)$ and $\Delta_{\vr}^t = O_p(1)$ for $t\ge 2$.
\end{proof}
\paragraph{Remarks.} Note that as long as we run the EP algorithm for at least
two cycles, Theorem \ref{app:thm:stable-region} guarantees that the EP parameters satisfy
the stable region criteria with an $O_{p}(1)$ stability width. The first and
second moments of the EP solution will also converge to that of Laplace's
approximation for a sufficiently large $n$.

In the rest of this article, we make reference to the following conditions
\begin{equation*}
  \lVert \vr_i^t - \mQ_i^t \widehat{\vtheta} + \phi_i^{(1)} (\widehat{\vtheta}) \rVert \le n^{-1} \Delta_{\vr}^t
\end{equation*}
and
\begin{equation*}
  \lVert \mQ_i^t - \phi_i^{(2)} (\widehat{\vtheta}) \rVert \le n^{-1} \Delta_{\vbeta}^t,
\end{equation*}
as the \textit{stable region criteria} and $(\Delta_{\vr}^t, \Delta_{\vbeta}^t)$ are referred as the \textit{stability width}.

\subsection{Expectation-Propagation Bernstein-von-Mises Theorem}
\label{app:sec:epel-bvm}
We are now ready to prove an expectation-propagation analogue of the
Bernstein-von-Mises theorem for Bayesian empirical likelihood models, where the
EP approximate posterior is constructed by using initial values as laid out in
the previous theorem and the number of iteration cycles is at least two. In
particular, our results says that the target Bayesian empirical likelihood
posterior has an asymptotic normal form that is equivalent to our proposed EP
approximate posterior. To the best of our knowledge, this is the first
Bernstein-von-Mises type result for Bayesian semiparametric EP posteriors. The
proof requires the following additional assumptions:
\begin{enumerate}[start=5, label=(\Roman*)]
  \item We assume that
        $\mS^\star = \bE\left [ h (\vz, \vtheta^\star) h (\vz, \vtheta^\star)^\top \right ]$
        is positive definite and all its entries are
        bounded. \label{app:item:S-star-pd-bounded}
  \item For any $a > 0$, there exists $\nu_{\EL} > 0$ such that as
        $n \rightarrow \infty$, we have \label{app:item:el-bounded}
        \begin{equation*}
          \lim_{n \rightarrow \infty} \bP^\star \left ( \sup_{\lVert \vtheta - \vtheta^\star \rVert \ge a} n^{-1} \left \{ \log \EL_n (\vtheta) - \log \EL_n (\vtheta^\star) \right \} \le - \nu_{\EL} \right ) = 1
        \end{equation*}
  \item The quantities
        $\lVert \partial h(\vz, \vtheta^\star) / \partial \vtheta \rVert$,
        $\lVert \partial^2 h(\vz, \vtheta) / \partial \vtheta \partial \vtheta^\top \rVert$,
        $\lVert h(\vz, \vtheta) \rVert^3$ are bounded by some integrable
        function $\widetilde{G} (\vz)$ in a neighbourhood of $\vtheta^\star$. \label{app:item:h-bounded}
  \item We assume $\mD = \bE\{ \partial h(\vz, \vtheta) / \partial \vtheta \}$
        is full rank, where the expectation is taken with respect to the true
        data generating distribution of $\vz$. \label{app:item:D-full-rank}
\end{enumerate}
We also denote $\mD_i (\vtheta) = \partial h(\vz_i, \vtheta) / \partial \vtheta$
and $\mV_{\vtheta^\star} = \mD^\top (\mS^\star)^{-1} \mD$.

\begin{theorem}
  \label{app:thm:ep-bvm}
  Assume conditions~\ref{app:item:theta-bound} to~\ref{app:item:D-full-rank} hold.
  Consider the EPEL Gaussian posterior parameterized by the linear-shift
  $\vr = \sum_{i=1}^{n} \vr_i$, precision $\mQ = \sum_{i=1}^{n} \mQ_i$ that are
  obtained after at least two EP cycles with initializations
  $\{\vr_i^0 \}_{i=0}^{n}$ and $\{\mQ_i^0 \}_{i=0}^{n}$ that satisfy
  \begin{equation*}
    n \max_i \lVert \mQ_i^0 \widehat{\vtheta} - \phi_i^{(1)} (\widehat{\vtheta}) - \vr_i^0 \rVert = \Delta_{\vr}^0
  \end{equation*}
  and
  \begin{equation*}
    n \max_i \lVert \phi_i^{(2)} (\widehat{\vtheta}) - \mQ_i^0 \rVert = \Delta_{\vbeta}^0,
  \end{equation*}
  where $ \Delta_{\vr}^0 = \sqrt{n}$ and $\Delta_{\vbeta}^0 = \tfrac{1}{2} \lVert \psi^{(2)} (\widehat{\vtheta}) \rVert$. Then, we have
  \begin{equation*}
    \dtv ( \mathcal{N}( \vr, \mQ ), p_{\EL}(\vtheta \mid \sD_{n}) ) = o_p(1).
  \end{equation*}
\end{theorem}
\begin{proof}
  The proof is presented in the following steps
  \begin{enumerate}
    \item Show that
          $\lVert \sum_{i=0}^n \phi_i^{(1)} ( \vtheta^\star ) \rVert = O_p (n^{1/2})$,
          $\lVert \sum_{i=0}^n \phi_i^{(2)} ( \vtheta^\star ) \rVert = O_p (n)$,
          $ \sum_{i=0}^n \lVert \phi_i^{(2)} ( \vtheta^\star ) \rVert = O_p (n)$
          and
          $\sum_{i=0}^n \lVert \phi_i^{(1)} ( \vtheta^\star ) \rVert = O_p(n)$. \label{app:item:bvm-1}
    \item Show that
          $\lVert \vtheta^\star - \widehat{\vtheta} \rVert = O_p(n^{-1/2})$. \label{app:item:bvm-2}
    \item Show that
          $\frac{1}{n} \psi^{(2)} (\widehat{\vtheta})$ converges in probability to a constant positive definite matrix
          and
          $\sum_{i=0}^n \lVert \phi_i^{(1)} ( \widehat{\vtheta} ) \rVert = O_p (n)$. \label{app:item:bvm-3}
    \item Use Theorem \ref{app:thm:stable-region} to show that there exists a stability
          region pivoted around $\widehat{\vtheta}$ and that the width of the
          stability regions scale as $O_p(1)$ whenever we specify the EP
          intial values as laid out in Theorem \ref{app:thm:stable-region} and iterate
          for at least two cycles. \label{app:item:bvm-4}
    \item Use previous step to show that the upper bound for KL-divergence
          between $\mathcal{N}( \vr, \mQ )$ and the Newton-Raphson posterior
          \begin{equation*}
            \mathcal{N} \left( \sum_{i=0}^n \phi_i^{(2)} (\widehat{\vtheta}) \widehat{\vtheta} , \sum_{i=0}^n \phi_i^{(2)} (\widehat{\vtheta})  \right)
          \end{equation*}
          is a function of $\Delta_{\vr}$ and $\Delta_{\vbeta}$ and
          consequently the upper bound is of order $O_p(n^{-1})$, for any
          $\vr$ and $\mQ$ within the stability region. Since
          $\KL \ge 2 (d_{TV})^2$, then upper bound of their total
          variation distance is $O_p(n^{-1/2})$. \label{app:item:bvm-5}
    \item Modify existing local asymptotic normality and Berstein-von-Mises
          results to show that the $\dtv$ between the target EL posterior and
          $\mathcal{N}( \sum_{i=0}^n \phi_i^{(2)} (\widehat{\vtheta}) \widehat{\vtheta} , \sum_{i=0}^n \phi_i^{(2)} (\widehat{\vtheta}) )$
          is of order $o_p(1)$. Note that, to the best of our ability, we are
          only to show that the remainder in the local asymptotic normality
          expansion is $o_p(1)$. \label{app:item:bvm-6}
    \item Previous two steps imply that $\dtv$ between $\mathcal{N}( \vr, \mQ )$
          and target EL posterior is of order $o_p(1)$, for any $\vr$ and $\mQ$
          within the stability region. \label{app:item:bvm-7}
  \end{enumerate}
  \textbf{Step~\ref{app:item:bvm-1}:}\\
  Since
  $\lVert \sum_{i=0}^n \phi_i^{(1)} ( \vtheta^\star ) \rVert \le \lVert \sum_{i=1}^n \phi_i^{(1)} ( \vtheta^\star ) \rVert + \lVert \log p (\vtheta^\star) \rVert$
  and $\lVert \log p (\vtheta^\star) \rVert = O_p(1)$, we need only to analyse
  $\lVert \sum_{i=1}^n \phi_i^{(1)} ( \vtheta^\star ) \rVert$. Note that
  \begin{align*}
    \sum_{i=1}^n \phi_i^{(1)} (\vtheta^\star) & = n \vlambda (\vtheta^\star)^\top  \sum_{i=1}^n w_i(\vtheta^\star; \sD_{n}) \mD_i (\vtheta^\star).
  \end{align*}
  From \citet{owen90empirical} eqn (2.17), we have
  $ \lVert \vlambda (\vtheta^\star) \rVert = \lVert \mS^{\star -1} \vu( \vtheta^\star) \rVert + o_p(n^{-1/2})$,
  where $\vu(\vtheta) = n^{-1} \sum_{i=1}^n \vh(\vz_i, \vtheta)$ and hence
  \begin{equation*}
    \lVert \sum_{i=1}^n \phi_i^{(1)} (\vtheta^\star)  \rVert
    \le n \lVert \vu( \vtheta^\star) \rVert \lVert {\mS^\star}^{-1} \rVert \lVert \sum_{i=1}^n w_i(\vtheta^\star; \sD_{n}) \mD_i (\vtheta^\star) \rVert
    + O_p(n^{1/2}).
  \end{equation*}
  Following Lemma~\ref{app:ConvergenceELmoments} in the auxiliary result section,
  we have
  $\sum_{i=1}^n w_i(\vtheta^\star; \sD_{n}) \mD_i (\vtheta^\star) \xrightarrow{\bP^\star} \mD (\vtheta^\star)$.
  Moreover, $\lVert \vu(\vtheta^\star) \rVert = O_p(n^{-1/2})$. Consequently,
  we have
  $ \lVert \sum_{i=1}^n \phi^{(1)} (\vtheta^\star) \rVert = O_p(n^{1/2})$.

  For the second-order derivative we have
  \begin{align}
    \label{app:LambdaSecondDerivative}
    \sum_{i=1}^n \phi_i^{(2)} (\vtheta^\star)
     & =   \sum_{i=1}^n n w_i(\vtheta^\star) \mC_i (\vtheta^\star) + \left ( \frac{\partial \vlambda (\vtheta^\star) }{ \partial \vtheta} \right ) ^\top \left \{ \sum_{i=1}^n n w_i (\vtheta) h(\vz_i, \vtheta^\star) h(\vz_i, \vtheta^\star)^\top \right \} \left ( \frac{\partial \vlambda (\vtheta^\star) }{ \partial \vtheta} \right ) \nonumber \\
     & - \sum_{i=1}^n n^2 w_i(\vtheta^\star)^2 \mD_i (\vtheta^\star)^\top \vlambda (\vtheta^\star) \vlambda (\vtheta^\star)^\top \mD_i (\vtheta^\star) ,
  \end{align}
  where
  \begin{gather*}
    \frac{\partial \vlambda (\vtheta^\star)}{ \partial \vtheta} = \left \{ \sum_{i=1}^n n w_i (\vtheta) h(\vz_i, \vtheta^\star) h(\vz_i, \vtheta^\star)^\top \right \}^{-1} \left \{ \sum_{i=1}^n n w_i (\vtheta^\star) \mD_i (\vtheta^\star) - \mB (\vtheta^\star) \right \}, \\
    \mB (\vtheta^\star) = \sum_{i=1}^n n^2 w_i (\vtheta^\star)^2 h(\vz_i, \vtheta^\star) \vlambda (\vtheta^\star)^\top \mD_i (\vtheta^\star)
  \end{gather*}
  and
  \begin{equation*}
    \mC_i (\vtheta^\star) = \left [  \frac{\partial^2 h(\vz_i, \vtheta^\star)^\top  }{\partial \theta_r \partial \theta_s} \vlambda(\vtheta^\star)  \right ]_{1 \le r \le p; 1 \le s \le p}.
  \end{equation*}
  It remains for us to show that $\sum_{i=1}^n \phi_i^{(2)} (\vtheta^\star)$ is dominated by the second term in the RHS of (\ref{app:LambdaSecondDerivative})
  \begin{equation*}
    \left ( \frac{\partial \vlambda (\vtheta^\star) }{ \partial \vtheta} \right ) ^\top \left \{ \sum_{i=1}^n n w_i (\vtheta) h(\vz_i, \vtheta^\star) h(\vz_i, \vtheta^\star)^\top \right \} \left ( \frac{\partial \vlambda (\vtheta^\star) }{ \partial \vtheta} \right ).
  \end{equation*}
  For the first term, we have
  \begin{align*}
    \sum_{i=1}^n n w_i(\vtheta^\star) \mC_i (\vtheta^\star)
     & = n \left [  \sum_{i=1}^n w_i(\vtheta^\star) \frac{\partial h(\vz_i, \vtheta^\star)}{ \partial \theta_r \partial \theta_s}^\top \vlambda(\vtheta^\star)  \right ]_{r,s}                                         \\
     & = n \left [  \sum_{i=1}^n w_i(\vtheta^\star) \frac{\partial h(\vz_i, \vtheta^\star)}{ \partial \theta_r \partial \theta_s}^\top \left \{ \mS^\star \vu(\vtheta^\star) + o_p(n^{-1/2}) \right \}  \right ]_{r,s}
  \end{align*}
  Since $\lVert \vu(\vtheta^\star) \rVert = O_p(n^{-1/2})$, we have $\lVert \sum_{i=1}^n n w_i(\vtheta^\star) \mC_i (\vtheta^\star) \rVert = O_p(n^{1/2})$.

  For the second term, we have
  \begin{align*}
     & \left \lVert \left ( \frac{\partial \vlambda (\vtheta^\star) }{ \partial \vtheta} \right ) ^\top \left \{ \sum_{i=1}^n n w_i (\vtheta) h(\vz_i, \vtheta^\star) h(\vz_i, \vtheta^\star)^\top \right \} \left ( \frac{\partial \vlambda (\vtheta^\star) }{ \partial \vtheta} \right ) \right \rVert                                                 \\
     & = n \left \lVert  \left \{ \sum_{i=1}^n w_i (\vtheta^\star) \mD_i (\vtheta^\star) - \mB (\vtheta^\star)/n \right \}^\top  \left \{ \sum_{i=1}^n w_i (\vtheta) h(\vz_i, \vtheta^\star) h(\vz_i, \vtheta^\star)^\top \right \}^{-1} \left \{ \sum_{i=1}^n w_i (\vtheta^\star) \mD_i (\vtheta^\star) - \mB (\vtheta^\star)/n \right \} \right \rVert \\
     & \le n \left \lVert \left \{ \sum_{i=1}^n w_i (\vtheta^\star) \mD_i (\vtheta^\star) \right \}^\top \left \{ \sum_{i=1}^n w_i (\vtheta) h(\vz_i, \vtheta^\star) h(\vz_i, \vtheta^\star)^\top \right \}^{-1} \left \{ \sum_{i=1}^n w_i (\vtheta^\star) \mD_i (\vtheta^\star) \right \} \right \rVert                                                 \\
     & + 2 n \left \lVert \left \{ \sum_{i=1}^n w_i (\vtheta^\star) \mD_i (\vtheta^\star) \right \}^\top \left \{ \sum_{i=1}^n w_i (\vtheta) h(\vz_i, \vtheta^\star) h(\vz_i, \vtheta^\star)^\top \right \}^{-1} \{ \mB (\vtheta^\star)/n \} \right \rVert                                                                                               \\
     & + 2n \left \lVert \{ \mB (\vtheta^\star)/n \}^\top \left \{ \sum_{i=1}^n w_i (\vtheta) h(\vz_i, \vtheta^\star) h(\vz_i, \vtheta^\star)^\top \right \}^{-1} \{ \mB (\vtheta^\star)/n \} \right \rVert
  \end{align*}
  Note that
  \begin{align*}
    \frac{\mB (\vtheta^\star)_{r,s}}{n}
     & = \sum_{i=1}^n w_i (\vtheta^\star) \frac{[ h(\vz_i, \vtheta^\star) \vlambda (\vtheta^\star)^\top \mD_i (\vtheta^\star) ]_{r,s}}{ 1 + \vlambda(\vtheta^\star)^\top h(\vz_i, \vtheta^\star) }    \\
     & \le \frac{1}{1 + o_p(1)}  \sum_{i=1}^n w_i (\vtheta^\star) \left \{ [h(\vz_i, \vtheta^\star) \vu(\vtheta^\star)^\top {\mS^\star}^{-1} \mD_i (\vtheta^\star) ]_{r,s} + o_p(n^{-1/2}) \right \},
  \end{align*}
  where the second inequality follows from arguments similar to \citet{owen90empirical} eqn (2.15). Since $\lVert \vu(\vtheta^\star) \rVert = O_p(n^{-1/2})$, we have the magnitude of each entry of $\mB (\vtheta^\star)/n $ converging to $0$ at rate $n^{-1/2}$. Moreover, by Lemma~\ref{app:ConvergenceELmoments}, $\lVert \{  \sum_{i=1}^n w_i (\vtheta^\star) h(\vz_i, \vtheta^\star) h(\vz_i, \vtheta^\star)^\top \}^{-1} \rVert = O_p(1)$ and $ \lVert \sum_{i=1}^n  w_i (\vtheta^\star) \mD_i (\vtheta^\star) \rVert = O_p(1)$. Since $\mD$ is a full rank matrix, $\lVert \sum_{i=1}^n  w_i (\vtheta^\star) \mD_i (\vtheta^\star) \rVert$ converges to a strictly positive constant. Hence,
  \begin{equation*}
    n^{-1} \left \lVert \left ( \frac{\partial \vlambda (\vtheta^\star) }{ \partial \vtheta} \right ) ^\top \left \{ \sum_{i=1}^n n w_i (\vtheta) h(\vz_i, \vtheta^\star) h(\vz_i, \vtheta^\star)^\top \right \} \left ( \frac{\partial \vlambda (\vtheta^\star) }{ \partial \vtheta} \right ) \right \rVert.
  \end{equation*}
  converges to a strictly positive constant. Consequently, $\psi^{(2)} (\vtheta^\star) = \frac{1}{n} \sum_{i=0} \phi_i^{(2)} (\vtheta^\star)$ converges in probability to a positive definite matrix.
  For the third term of (\ref{app:LambdaSecondDerivative}),
  \begin{align*}
     & \left \lVert \sum_{i=1}^n n^2 w_i(\vtheta^\star)^2 \mD_i (\vtheta^\star)^\top \vlambda (\vtheta^\star) \vlambda (\vtheta^\star)^\top \mD_i (\vtheta^\star) \right \rVert \\
     & \le \sum_{i=1}^n n^2 w_i(\vtheta^\star)^2 \lVert  \mD_i (\vtheta^\star) \rVert^2 \lVert \vlambda (\vtheta^\star) \rVert^2                                                \\
     & \le  \frac{n}{1 + o_p(1)} \sum_{i=1}^n w_i(\vtheta^\star) \lVert  \mD_i (\vtheta^\star) \rVert^2 \lVert \vlambda (\vtheta^\star) \rVert^2 ,
  \end{align*}
  where the third inequality follows from arguments similar to \citet{owen90empirical}
  eqn (2.15). Since $\lVert \vlambda (\vtheta^\star) \rVert^2 = O_p(n^{-1})$,
  we have
  $\lVert \sum_{i=1}^n n^2 w_i(\vtheta^\star)^2 \mD_i (\vtheta^\star)^\top \vlambda (\vtheta^\star) \vlambda (\vtheta^\star)^\top \mD_i (\vtheta^\star) \rVert = O_p(1)$.

  Consequently,
  $n^{-1} \left \lVert \sum_{i=1}^n \phi_i^{(2)} (\vtheta^\star) \right \rVert$
  converges to a strictly positive constant. Using very similar steps to our
  proof for
  $n^{-1} \left \lVert \sum_{i=1}^n \phi_i^{(2)} (\vtheta^\star) \right \rVert$,
  we may also show that
  $n^{-1} \sum_{i=1}^n \left \lVert \phi_i^{(2)} (\vtheta^\star) \right \rVert = O_p(1)$.
  Next, we analyse the divergence properties of
  $\sum_{i=1}^n \lVert \phi_i^{(1)} (\vtheta^\star) \rVert$. Note that
  \begin{equation*}
    \phi_i^{(1)} (\vtheta^\star) = n w_i(\vtheta^\star) \vlambda(\vtheta^\star)^\top \mD_i(\vtheta^\star) + n w_i (\vtheta^\star) h(\vz_i, \vtheta^\star)^\top  \frac{\partial \vlambda (\vtheta^\star) }{ \partial \vtheta}
  \end{equation*}
  and hence
  \begin{align*}
     & \frac{1}{n} \sum_{i=1}^n \lVert \phi_i^{(1)} (\vtheta^\star) \rVert                                                                                                                                                                                                                                                                                                       \\
     & \le \lVert \vlambda(\vtheta^\star) \rVert \sum_{i=1}^n w_i(\vtheta^\star) \lVert  \mD_i(\vtheta^\star) \rVert + \left \lVert  \sum_{i=1}^n w_i(\vtheta^\star)  h(\vz_i, \vtheta^\star) \right \rVert \left \lVert \frac{\partial \vlambda(\vtheta^\star) }{\partial \vtheta} \right \rVert                                                                                \\
     & \le \lVert \vlambda(\vtheta^\star) \rVert \sum_{i=1}^n w_i(\vtheta^\star) \lVert  \mD_i(\vtheta^\star) \rVert + \left \lVert  \sum_{i=1}^n w_i(\vtheta^\star)  h(\vz_i, \vtheta^\star) \right \rVert \left \lVert \left \{ \sum_{i^\prime=1}^n w_{i^\prime} (\vtheta) h(\vz_{i^\prime}, \vtheta^\star) h(\vz_{i^\prime}, \vtheta^\star)^\top \right \}^{-1} \right \rVert \\
     & \times \left \{ \left \lVert  \sum_{i^\prime=1}^n w_{i^\prime} (\vtheta^\star) \mD_{i^\prime} (\vtheta^\star) \right \rVert + \left \lVert  \frac{\mB (\vtheta^\star)}{n} \right \rVert   \right \}
  \end{align*}
  Clearly the first term of the upper bound is order $O_p(n^{-1/2})$ because
  $\lVert \vlambda(\vtheta^\star) \rVert = O_p(n^{-1/2})$ and
  $ \sum_{i=1}^n w_i(\vtheta^\star) \lVert \mD_i(\vtheta^\star) \rVert = O_p(1)$.
  Since $\lVert \mB (\vtheta^\star)\rVert/n$ is of order $O_p(n^{-1/2})$,
  $\left \lVert \left \{ \sum_{i^\prime=1}^n w_{i^\prime} (\vtheta) h(\vz_{i^\prime}, \vtheta^\star) h(\vz_{i^\prime}, \vtheta^\star)^\top \right \}^{-1} \right \rVert$
  converges to a strictly positive constant, and $\mD$ is full rank, we have
  that $\frac{1}{n} \sum_{i=1}^n \lVert \phi_i^{(1)} (\vtheta^\star) \rVert$
  converges to a strictly positive constant.

  \textbf{Step~\ref{app:item:bvm-2}:}\\
  The proof is a slight modification to the proof for Lemma~1 and parts of
  Theorem 1 in \citet{qin94empirical} (alternatively Theorems 3.1 and 3.2 of
  \citet{newey04higher}) and therefore we provide only a sketch here. First, we
  show that the MAP estimator occurs within the ball
  $\{ \vtheta \, : \, \lVert \vtheta - \vtheta^\star \rVert \le n^{-1/3} \}$.
  Now, consider a point on the surface of the ball where may be expressed as
  $\vtheta = \vtheta^\star + \vu n^{-1/3}$, for some $\lVert \vu \rVert = 1$.
  Following \citet{owen90empirical} and~\ref{app:item:h-bounded}, we have
  \begin{equation*}
    \vlambda_s (\vtheta) = \left [ \left \{  \frac{1}{n} \sum_{i=1}^n  h(\vz_i, \vtheta) h(\vz_i, \vtheta)^\top \right \}^{-1} \left ( \frac{1}{n} \sum_{i=1}^n  h(\vz_i, \vtheta) \right ) \right ]_s + o_p(n^{-1/3}).
  \end{equation*}
  Substituting the above and after some algebraic manipulation, we have
  \begin{align*}
     & \sum_{i=0}^n \phi_i (\vtheta) - n \log(n)                                                                                                                                                                                                                               \\
     & = \sum_{i=1}^n \phi_i (\vtheta) - \log p(\vtheta) - n \log(n)                                                                                                                                                                                                           \\
     & = \frac{n}{2} \left ( \frac{1}{n} \sum_{i=1}^n  h(\vz_i, \vtheta) \right )^\top \left [  \frac{1}{n} \sum_{i=1}^n  h(\vz_i, \vtheta) h(\vz_i, \vtheta)^\top \right ]^{-1} \left ( \frac{1}{n} \sum_{i=1}^n  h(\vz_i, \vtheta) \right ) - \log p(\vtheta) + o_p(n^{1/3})
  \end{align*}
  By~\ref{app:item:h-bounded} and a first-order Taylor's expansion of
  $h(\vz_i, \vtheta)$ about $\vtheta^\star$, we have
  $\lVert \tfrac{1}{n} \sum_{i=1}^n h(\vz_i, \vtheta) \rVert \le \lVert \tfrac{1}{n} \sum_{i=1}^n h(\vz_i, \vtheta^\star) \rVert + O_p(n^{-1/3})$
  and hence $\sum_{i=0}^n \phi_i (\vtheta) - n \log(n)$ diverges at rate
  $n^{1/3}$. On the other hand,
  \begin{equation*}
    \sum_{i=0}^n \phi_i (\vtheta^\star) - n \log(n)= \frac{n}{2} \left ( \frac{1}{n} \sum_{i=1}^n  h(\vz_i, \vtheta^\star) \right )^\top \left [  \frac{1}{n} \sum_{i=1}^n  h(\vz_i, \vtheta^\star) h(\vz_i, \vtheta^\star)^\top \right ]^{-1} \left ( \frac{1}{n} \sum_{i=1}^n  h(\vz_i, \vtheta^\star) \right ) -  \log p(\vtheta^\star) + o_p(1)
  \end{equation*}
  and hence $\sum_{i=0}^n \phi_i (\vtheta^\star)$ is bounded from infinity in probability. Hence, we have shown that $\sum_{i=0}^n \phi (\vtheta)$ has a minimiser $\widehat{\vtheta}$ in the interior of the ball  $\{ \vtheta \, : \, \lVert \vtheta - \vtheta^\star \rVert \le n^{-1/3} \}$ with probability approaching 1. Since the ball collapses into a point set $\{ \vtheta^\star \}$ as $n \rightarrow \infty$, this minimiser is consistent. We proceed to evaluate its consistency rate. Note that the maxima $(\widehat{\vtheta}, \widehat{\vlambda})$ of the function
  \begin{equation*}
    \sum_{i=1}^n \log \{ n + n \vlambda^\top h(\vz_i, \vtheta) \} - \log p(\theta)
  \end{equation*}
  satisfies $Q_{1n} (\vtheta, \vlambda) = \vzero$ and $Q_{2n} (\vtheta, \vlambda) - \tfrac{ p^{(1)} (\vtheta) }{n p(\vtheta) } = \vzero$, where
  \begin{equation*}
    Q_{1n} (\vtheta, \vlambda)  = n^{-1} \sum_{i=1}^n \frac{h(\vz_i, \vtheta)}{1 + \vlambda^\top h(\vz_i, \vtheta)}
  \end{equation*}
  and
  \begin{equation*}
    Q_{2n} (\vtheta, \vlambda) = n^{-1} \sum_{i=1}^n \frac{ \nabla_{\vtheta} h(\vz_i, \vtheta) ^\top \vlambda   }{1 + \vlambda^\top h(\vz_i, \vtheta)}.
  \end{equation*}
  By Taylor's expansion of each $Q_{1n}$ and $Q_{2n}$ about $(\vtheta^\star, \vzero)$, we have
  \begin{equation*}
    \vzero = Q_{1n} (\widehat{\vtheta}, \widehat{\vlambda}) = Q_{1n} (\vtheta^\star, \vzero) + \tfrac{\partial}{\partial \vtheta} Q_{1n} (\vtheta^\star, \vzero) ( \widehat{\vtheta} - \vtheta^\star ) + \tfrac{\partial}{\partial \vlambda} Q_{1n} (\vtheta^\star, \vzero) \widehat{\vlambda} +  o_p( n^{-1/2} ),
  \end{equation*}
  and
  \begin{equation*}
    \vzero = Q_{2n} (\widehat{\vtheta}, \widehat{\vlambda}) -  \frac{ p^{(1)}( \widehat{\vtheta}) }{n p(\widehat{\vtheta}) } = Q_{2n} (\vtheta^\star, \vzero) + \tfrac{\partial}{\partial \vtheta} Q_{2n} (\vtheta^\star, \vzero) ( \widehat{\vtheta} - \vtheta^\star )  + \frac{\partial Q_{2n} (\vtheta^\star, \vzero) }{  \partial \vlambda } \widehat{\vlambda} -  \frac{ p^{(1)}( \widehat{\vtheta}) }{n p(\widehat{\vtheta}) } +  o_p( n^{-1/2} ),
  \end{equation*}
  Note that $\tfrac{ \partial Q_{2n} (\vtheta^\star, \vzero) }{\partial \vtheta } = \vzero$. Hence, we have
  \begin{equation*}
    \mJ_n
    \begin{pmatrix}
      \widehat{\vlambda} \\
      \widehat{\vtheta} - \vtheta^\star
    \end{pmatrix}  =  \begin{pmatrix}
      - Q_{1n}(\vtheta^\star, \vzero) \\
      \vzero
    \end{pmatrix}  +  \begin{pmatrix}
      \vzero \\
      - \frac{ p^{(1)}( \widehat{\vtheta}) }{n p(\widehat{\vtheta}) }
    \end{pmatrix} + \vxi
  \end{equation*}
  where
  \begin{equation*}
    \mJ_n = \begin{pmatrix}
      \tfrac{\partial Q_{1n} (\vtheta^\star, \vzero) }{  \partial \vlambda } & \tfrac{\partial Q_{1n} (\vtheta^\star, \vzero) }{  \partial \vtheta } \\
      \tfrac{\partial Q_{2n} (\vtheta^\star, \vzero) }{  \partial \vlambda } & \vzero
    \end{pmatrix}
  \end{equation*}
  and each entry of $\vxi$ is $o_p(a_n)$, where
  $a_n = \lVert \widehat{\vtheta} - \vtheta^\star \rVert + \lVert \widehat{\vlambda} \rVert$.
  By taking $\mJ^{-1}_n$ and $\lVert \cdot \rVert$ on both sides and then
  apply appropriate bounds, we have $a_n = O_p(n^{-1/2})$. Consequently, we
  have
  \begin{align*}
     & \sqrt{n} (\widehat{\vtheta} - \vtheta^\star)                                                                                                                                                                                               \\
     & = \sqrt{n} \mV_{\vtheta^\star}^{-1} \mD ^\top \mS^{\star \, -1} Q_{1n} (\vtheta^\star, \vzero) + n^{-1/2} \mV_{\vtheta^\star}^{-1} \mD^\top \mS^{\star \, -1} \frac{  p^{(1)}(\widehat{\vtheta}) }{p(\widehat{\vtheta}) } + o_p(n^{-1/2}).
  \end{align*}
  Following~\ref{app:item:theta-bound} and~\ref{app:item:prior-smoothness}, we have $0 < \lVert p^{(1)}(\vtheta) \rVert < \infty$ and $0 < p(\vtheta) < \infty$ for all $\vtheta \in \mTheta$. Therefore, the second term involving $\frac{  p^{(1)}(\widehat{\vtheta}) }{p(\widehat{\vtheta}) }$ has norm of order $O_p(n^{-1/2})$. Moreover, $\lVert Q_{1n} (\vtheta^\star, \vzero) \rVert = O_p(n^{-1/2})$. Hence, it follows that $\lVert \widehat{\vtheta} - \vtheta^\star \rVert = O_p(n^{-1/2})$.\\
  \textbf{Step~\ref{app:item:bvm-3}:}\\
  We have
  \begin{align*}
    \left \lVert \sum_{i=0}^n \phi_i^{(2)} (\widehat{\vtheta}) - \sum_{i=0}^n \phi_i^{(2)} (\vtheta^\star) \right \rVert
     & \le  \sum_{i=1}^n \left \lVert \phi_i^{(2)} (\widehat{\vtheta}) -  \phi_i^{(2)} (\vtheta^\star) \right \rVert + \lVert   \phi_0^{(2)}( \widehat{\vtheta} ) -  \phi_0^{(2)}( \vtheta^\star ) \rVert \\
     & \le n p^2 K_3 \lVert \widehat{\vtheta} - \vtheta^\star \rVert + \lVert \phi_0^{(2)}( \widehat{\vtheta} ) -  \phi_0^{(2)}( \vtheta^\star ) \rVert
  \end{align*}
  and hence
  \begin{equation*}
    \left \lVert \sum_{i=1}^n \phi_i^{(2)} (\widehat{\vtheta})  \right \rVert
    \le \left \lVert \sum_{i=1}^n \phi_i^{(2)} (\vtheta^\star)  \right \rVert +  n p^2 K_3 \lVert \widehat{\vtheta} - \vtheta^\star \rVert + \lVert   \phi_0^{(2)}( \widehat{\vtheta} ) -  \phi_0^{(2)}( \vtheta^\star ) \rVert
  \end{equation*}
  From~\ref{app:item:prior-smoothness} about the higher-order differentiability of
  the prior, we can deduce that
  $\lVert \phi_0^{(2)}( \vtheta^\prime ) - \phi_0^{(2)}( \vtheta^{\prime \prime} ) \rVert < \infty$
  for all $\vtheta \in \mTheta$ and $\vtheta^{\prime \prime} \in \mTheta$. Also,
  using results from Step~\ref{app:item:bvm-1} that
  $\left \lVert \sum_{i=1}^n \phi_i^{(2)} (\vtheta^\star) \right \rVert = O_p(n)$
  and from Step~\ref{app:item:bvm-2} that
  $\lVert \widehat{\vtheta} - \vtheta^\star \rVert = O_p (n^{-1/2})$, we have
  $\left  \lVert \frac{1}{n} \sum_{i=0}^n \phi_i^{(2)} (\widehat{\vtheta}) - \frac{1}{n} \sum_{i=0}^n \phi_i^{(2)} (\vtheta^\star) \right \rVert$ converges in probability to $0$. Hence, $\frac{1}{n} \psi (\widehat{\vtheta}) = \frac{1}{n} \psi (\vtheta^\star) + \frac{1}{n} \psi (\widehat{\vtheta}) - \frac{1}{n} \psi (\vtheta^\star)$ converges in probability to a positive definite matrix (same limiting matrix as $\frac{1}{n} \psi (\vtheta^\star)$).
  Moreover, using similar steps, we have
  \begin{equation*}
    \sum_{i=0}^n \left \lVert \phi_i^{(2)} (\widehat{\vtheta})  \right \rVert \le \sum_{i=1}^n \left \lVert \phi_i^{(2)} (\vtheta^\star)  \right \rVert + n K_3 \lVert \widehat{\vtheta} - \vtheta^\star \rVert + \lVert   \phi_0^{(2)}( \widehat{\vtheta} ) -  \phi_0^{(2)}( \vtheta^\star ) \rVert
  \end{equation*}
  and
  \begin{equation*}
    \sum_{i=0}^n \left \lVert \phi_i^{(1)} (\widehat{\vtheta})  \right \rVert \le \sum_{i=1}^n \left \lVert \phi_i^{(1)} ( \vtheta^\star)  \right \rVert + nK_2 \lVert \widehat{\vtheta} - \vtheta^\star \rVert + \lVert   \phi_0^{(1)}( \widehat{\vtheta} ) -  \phi_0^{(1)}( \vtheta^\star ) \rVert
  \end{equation*}
  Hence, we have $ \sum_{i=0}^n \left \lVert \phi_i^{(2)} (\widehat{\vtheta})  \right \rVert = O_p(n)$ and $\sum_{i=0}^n \left \lVert \phi_i^{(1)} (\widehat{\vtheta})  \right \rVert = O_p(n)$.\\
  \textbf{Step~\ref{app:item:bvm-4}:}\\
  Since $n^{-1} \left \lVert \sum_{i=0}^n \phi_i^{(2)} (\widehat{\vtheta})  \right \rVert$ converges to a strictly positive constant, $ \sum_{i=0}^n \left \lVert \phi_i^{(2)} (\widehat{\vtheta})  \right \rVert = O_p(n)$, and $\sum_{i=0}^n \left \lVert \phi_i^{(1)} (\widehat{\vtheta})  \right \rVert = O_p(n)$, our required follows directly from Theorem \ref{app:thm:stable-region}. In particular, for every $t \ge 2$ and $i = 1, \ldots, n$, we have
  \begin{equation*}
    n \lVert \vr_i^t - \mQ_i^t \widehat{\vtheta} + \phi_i^{(1)} (\widehat{\vtheta}) \rVert \le \Delta_{\vr}^t
  \end{equation*}
  and
  \begin{equation*}
    n \lVert \mQ_i^t - \phi_i^{(2)} (\widehat{\vtheta}) \rVert \le \Delta_{\vbeta}^t
  \end{equation*}
  where $\Delta_{\vr}^t = O_p(1)$ and $\Delta_{\vbeta}^t = O_p(1)$.\\
  \textbf{Step~\ref{app:item:bvm-5}:}\\
  Note that the Laplace approximate posterior is
  $\mathcal{N}( \psi^{(2)} (\widehat{\vtheta}) \widehat{\vtheta} , \psi^{(2)} (\widehat{\vtheta}) )$
  and the EP posterior (truncating at iteration $t$) is
  $\mathcal{N}( \vr^t, \mQ^t )$, where we use the linear shift-precision
  parameterisation for the multivariate normal distribution. Recall that
  \begin{equation*}
    \lVert \vr_i^t - \mQ_i^t \widehat{\vtheta} + \phi_i^{(1)} (\widehat{\vtheta}) \rVert \le \Delta_{\vr}^t / n
  \end{equation*}
  and
  \begin{equation*}
    \lVert \mQ_i^t - \phi_i^{(2)} (\widehat{\vtheta}) \rVert \le \Delta_{\vbeta}^t / n
  \end{equation*}
  By letting $\mOmega_1 = \psi^{(2)} (\widehat{\vtheta})$ and $\mOmega_2 = \mQ^t$, the KL-divergence between the two distributions is
  \begin{equation*}
    2 \KL^{12} = \left ( \widehat{\vtheta} - \mOmega_2^{-1} \vr^t  \right )^\top \mOmega_2 \left ( \widehat{\vtheta} - \mOmega_2^{-1} \vr^t  \right ) + \tr \{  (\mOmega_2 - \mOmega_1) \mOmega_1^{-1} \} - \log ( \det (\mOmega_2 ) / \det (\mOmega_1) )
  \end{equation*}
  We first prove that $\lVert \mOmega_1^{-1} \rVert = O_p(n^{-1})$. Note that $\mOmega_1$ is dominated by $\sum_{i=1}^n \phi_i^{(2)} (\widehat{\vtheta})$. Consider the inequality
  \begin{align*}
     & \left \lVert \frac{1}{n} \sum_{i=1}^n \phi_i^{(2)} (\widehat{\vtheta}) - \mD(\vtheta^\star)^\top \mS^{\star -1} \mD(\vtheta^\star)  \right \rVert                                                                                                                                                    \\
     & \le \left \lVert \frac{1}{n}  \sum_{i=1}^n \phi_i^{(2)} (\widehat{\vtheta}) - \frac{1}{n}  \sum_{i=1}^n \phi_i^{(2)} (\vtheta^\star)  \right \rVert + \left \lVert \frac{1}{n}  \sum_{i=1}^n \phi_i^{(2)} (\vtheta^\star) - \mD(\vtheta^\star)^\top \mS^{\star -1} \mD(\vtheta^\star)  \right \rVert \\
  \end{align*}
  Following the proof of Step~\ref{app:item:bvm-1}, the term
  $\sum_{i=1}^n \phi_i^{(2)} (\vtheta^\star)$ is dominated as
  \begin{align*}
     & \frac{1}{n} \sum_{i=1}^n \phi_i^{(2)} (\vtheta^\star)                                                                                                                                                                                                                            \\
     & = \frac{1}{n} \left \{ \sum_{i=1}^n w_i (\vtheta^\star) \mD_i (\vtheta^\star) \right \}^\top  \left \{ \sum_{i=1}^n w_i (\vtheta) h(\vz_i, \vtheta^\star) h(\vz_i, \vtheta^\star)^\top \right \}^{-1} \left \{ \sum_{i=1}^n w_i (\vtheta^\star) \mD_i (\vtheta^\star)  \right \} \\
     & + \text{matrix with elementwise $O_p(n^{-1/2})$}.
  \end{align*}
  Following Lemma~\ref{app:ConvergenceELmoments}, we have the convergence $\sum_{i=1}^n w_i (\vtheta^\star) \mD_i (\vtheta^\star) \xrightarrow{\bP^\star} \mD$ and $\sum_{i=1}^n w_i (\vtheta^\star) h(\vz_i, \vtheta^\star) h(\vz_i, \vtheta^\star)^\top \xrightarrow{\bP^\star} \mS^\star$. Hence,
  \begin{equation*}
    \frac{1}{n} \sum_{i=1}^n \phi_i^{(2)} (\vtheta^\star) \xrightarrow{\bP^\star} \mD ^\top (\mS^\star)^{-1}  \mD = \mV_{\vtheta^\star}
  \end{equation*}
  Since the eigenvalues are continuous functions of matrix entries, by continuous mapping theorem we have for each $j = 1, \ldots, p$,
  \begin{equation*}
    \eigen_j \{ \mOmega_1 / n \} \xrightarrow{\bP^\star} \eigen_j \{ \mV_{\vtheta^\star} \}.
  \end{equation*}
  Since $\mS^\star$ is positive definite and $\mD$ is full rank, the matrix $\mV_{\vtheta^\star}$ is positive definite, we may apply continuous mapping theorem again to obtain
  \begin{equation*}
    1/ \eigen_j \{ \mOmega_1 / n \} \xrightarrow{\bP^\star} 1/\eigen_j \{ \mV_{\vtheta^\star} \}.
  \end{equation*}
  For every $n$, $ \mOmega_1$ is positive definite. Hence, $1/ \eigen_j \{ \mOmega_1 / n \} = \eigen_j \{ n \mOmega_1^{-1} \}$. Consequently, we have
  \begin{equation*}
    \eigen_j \{ n \mOmega_1^{-1} \}  \xrightarrow{\bP^\star} 1/\eigen_j \{ \mV_{\vtheta^\star} \}.
  \end{equation*}
  Note that the squared-Euclidean norm may be expressed as the sum of squares of the eigenvalues:
  \begin{align*}
    \lVert \mOmega_1^{-1} \rVert^2 & = \sum_{j=1}^p \eigen_j \{ \mOmega_1^{-1} \}^2            \\
                                   & = n^{-2}  \sum_{j=1}^p \eigen_j \{ n \mOmega_1^{-1} \}^2.
  \end{align*}
  Since $\sum_{j=1}^p \eigen_j \{ n \mOmega_1^{-1} \}^2 = O_p(1)$, we have $\lVert \mOmega_1^{-1} \rVert = O_p(n^{-1})$. Moreover, since $\lVert \mOmega_2 - \mOmega_1 \rVert \le \Delta_{\vbeta} = O_p(1)$ and $ \mOmega_1/n \xrightarrow{\bP^\star}   \mV_{\vtheta^\star} $ , we have  $ \mOmega_2/n  \xrightarrow{\bP^\star}   \mV_{\vtheta^\star} $ and hence $\lVert \mOmega_2^{-1} \rVert = O_p(n^{-1})$. It remains for us to analyse the asymptotic behaviour of $\left \lVert \widehat{\vtheta} - \mOmega_2^{-1} \vr^t  \right \rVert$ and $\log ( \det (\mOmega_2 ) / \det (\mOmega_1) )$:
  \begin{align*}
    \left \lVert \widehat{\vtheta} - \mOmega_2^{-1} \vr^t  \right \rVert & = \lVert \mOmega_2^{-1}  ( \mOmega_2 \widehat{\vtheta} - \vr^t  ) \rVert                                                \\
                                                                         & \le \lVert \mOmega_2^{-1} \rVert \lVert \mOmega_2 \widehat{\vtheta} - \vr^t  \rVert                                     \\
                                                                         & =  \lVert \mOmega_2^{-1} \rVert   \lVert \vr^t -  \mOmega_2 \widehat{\vtheta}  + \psi^{(1)} (\widehat{\vtheta})  \rVert \\
                                                                         & \le O_p(n^{-1}),
  \end{align*}
  where the third equality follows from noting that
  $\psi^{(1)} (\widehat{\vtheta}) = \vzero$ and the fourth inequality follows
  from Step~\ref{app:item:bvm-4}. Consequently, the first term of the upper bound
  for $\KL^{12}$ is of order
  $\lVert \mOmega_2 \rVert \times O_p(n^{-2}) = O_p (n^{-1})$. Also we analyse
  the second term:
  \begin{align}
    \label{app:remake::secondEqn}
    \log ( \det (\mOmega_2 ) / \det (\mOmega_1) )
     & = \log ( \det (\mOmega_2 \mOmega_1^{-1})  ) \nonumber                                                                        \\
     & = \log \det \left \{ \mI +  (\mOmega_2 - \mOmega_1) \mOmega_1^{-1}  \right \} \nonumber                                      \\
     & = \sum_{j=1}^p \log  \{ 1 + \eigen_j \{  (\mOmega_2 - \mOmega_1) \mOmega_1^{-1}  \} \} \nonumber                             \\
     & =  \sum_{j=1}^p \log  \{ 1 + \eigen_j \{ \mOmega_1^{-1/2}   (\mOmega_2 - \mOmega_1) \mOmega_1^{-1/2}  \} \} \nonumber        \\
     & =  \sum_{j=1}^p \log  \{ 1 + \eigen_j \{ \mOmega_1^{-1/2}   \mOmega_2 \mOmega_1^{-1/2} - \mI   \} \} \nonumber               \\
     & \ge p - \sum_{j=1}^p \left [ 1 + \eigen_j \{ \mOmega_1^{-1/2}   \mOmega_2 \mOmega_1^{-1/2} - \mI \}  \right ]^{-1} \nonumber \\
     & = p - \sum_{j=1}^p \frac{1}{\eigen_j \{ \mOmega_1^{-1/2}   \mOmega_2 \mOmega_1^{-1/2} \} } \nonumber                         \\
     & = p - \tr\{ \mOmega_1 \mOmega_2^{-1}  \}
  \end{align}
  where third equality follows from the fact that for any matrix $\mA$ with real
  eigenvalues (not necessarily symmetric), we have
  $\eigen_j(\mA + \mI) = \eigen_j(\mA) + 1$ for all $j=1,\ldots,p$, and then
  noting that the log determinant of a matrix equals to sum of the
  log-eigenvalues. The fourth equality follows from the fact that for any two
  compatible square matrices $\eigen_j(\mA \mC) = \eigen_j(\mC \mA)$ for all
  $j = 1, \ldots p$. The fifth equality follows from noting that
  $\mOmega_1^{-1/2} (\mOmega_2 - \mOmega_1) \mOmega_1^{-1/2} = \mOmega_1^{-1/2} \mOmega_2 \mOmega_1^{-1/2} - \mI$.
  The sixth inequality follow from $\log(1 + x) \ge x/(1+x)$. The seventh
  equality follows from the identity
  $\eigen_j \{ \mOmega_1^{-1/2} \mOmega_2 \mOmega_1^{-1/2} - \mI \} = \eigen_j \{ \mOmega_1^{-1/2} \mOmega_2 \mOmega_1^{-1/2} \} - 1$.
  The eigth equality follows from the result that for any pd $\mA$, we have
  $\eigen_j(\mA^{-1}) = 1/\eigen_j(\mA)$ and then applying the result that
  $\eigen_j(\mA \mC) = \eigen_j(\mC \mA)$. The ninth equality follows from the
  result that trace of a symmetric matrix equals to the its sum of eigenvalues.
  Hence we have
  \begin{align*}
    \KL^{12}
     & \le \left ( \widehat{\vtheta} - \mOmega_2^{-1} \vr^t  \right )^\top \mOmega_2 \left ( \widehat{\vtheta} - \mOmega_2^{-1} \vr^t  \right ) + \tr \{  (\mOmega_2 - \mOmega_1) \mOmega_1^{-1} \} - p + \tr\{ \mOmega_1 \mOmega_2^{-1}  \}                                                 \\
     & = \left ( \widehat{\vtheta} - \mOmega_2^{-1} \vr^t  \right )^\top \mOmega_2 \left ( \widehat{\vtheta} - \mOmega_2^{-1} \vr^t  \right ) + \tr\{ \mOmega_1 \mOmega_2^{-1}  - \mI \} + \tr\{ \mOmega_2 \mOmega_1^{-1}  - \mI \}                                                          \\
     & \le \left ( \widehat{\vtheta} - \mOmega_2^{-1} \vr^t  \right )^\top \mOmega_2 \left ( \widehat{\vtheta} - \mOmega_2^{-1} \vr^t  \right ) +  \lVert \mOmega_1  - \mOmega_2   \rVert \lVert \mOmega_2^{-1} \rVert +  \lVert \mOmega_2  - \mOmega_1  \rVert \lVert \mOmega_1^{-1} \rVert \\
     & = O_p (n^{-1}),
  \end{align*}
  where the first inequality follows from (\ref{app:remake::secondEqn}), the
  second equality follows from the identity $p = \tr(\mI)$, the third
  inequality follows from
  $\mI = \mOmega_2 \mOmega_2^{-1} = \mOmega_1 \mOmega_1^{-1} $ and the
  result that for two square matrices of order $p$ we have
  $\tr(\mA \mC) = \sum_{j=1}^p \va_{j \cdot}^\top \vc_{\cdot j} \le \sum_{j=1}^p \lVert \va_{j \cdot} \rVert \lVert \vc_{\cdot j} \rVert \le \sqrt{ \sum_{j=1}^p \lVert \va_{j \cdot} \rVert^2 } \times \sqrt{ \sum_{j=1}^p \lVert \vc_{\cdot j} \rVert^2 } = \lVert \mA \rVert \lVert \mC \rVert$,
  and also earlier derived results:
  $\lVert \mOmega_1^{-1} \rVert = O_p(n^{-1})$,
  $\lVert \mOmega_2^{-1} \rVert = O_p(n^{-1})$,
  $\lVert \mOmega_1 \rVert = O_p(n)$, $\lVert \mOmega_2 \rVert = O_p(n)$,
  and
  $\left \lVert \widehat{\vtheta} - \mOmega_2^{-1} \vr^t \right \rVert = O_p(n^{-1})$.
  The last term follows by noting that all terms in the upper bound are
  $O_p(n^{-1})$. Following Pinkser's inequality, we have
  $\dtv ( \mathcal{N} (\vr^t, \mOmega_2), \mathcal{N} ( \mOmega_1 \widehat{\vtheta}, \mOmega_1 ) ) \le \sqrt{ \KL^{12} /2 } = O_p(n^{-1/2})$.

  \textbf{Step~\ref{app:item:bvm-6}}:\\
  Note that
  \begin{align*}
     & \dtv \left \{ \mathcal{N} (\vr^t, \mQ), p_{\EL} (\vtheta \mid \sD_{n}) \right \}                                                                                                                                                                                                                                                                         \\
     & \le \dtv \left \{ \mathcal{N} (\vr^t, \mQ), \mathcal{N} \left ( \psi^{(2)} ( \widehat{\vtheta} ) \widehat{\vtheta} , \psi^{(2)} ( \widehat{\vtheta} )  \right ) \right \} +  \dtv \left \{ \mathcal{N} \left ( \psi^{(2)} ( \widehat{\vtheta} ) \widehat{\vtheta} , \psi^{(2)} ( \widehat{\vtheta} )  \right ), p_{\EL} (\vtheta \mid \sD_{n}) \right \}
  \end{align*}
  As shown in Step~\ref{app:item:bvm-5}, the first term of the upper bound is of
  order $O_p(n^{-1/2})$. The second term follows by a substantial modification
  of the proofs of Lemmas 1 to 3 in \citet{yu24variational} with $\vtheta^\star$
  replaced by $\widehat{\vtheta}$. Details have been provided in the next
  section.
\end{proof}

We state and prove the following results that mirror Lemmas 1 to 3 in
\citet{yu24variational}, with $\vtheta^\star$ replaced by $\widehat{\vtheta}$.
\begin{lemma}
  \label{app:LemmaLANMELE}
  Assume conditions~\ref{app:item:theta-bound} to~\ref{app:item:D-full-rank} hold. Then,
  the log empirical likehood evaluated at the posterior mode satisfies the local
  asymptotic normality (LAN) condition, i.e., for any $t \in \bR^m$, we have
  \begin{align*}
     & \log \EL ( \widehat{\vtheta} + n^{-1/2} t ) - \log \EL ( \widehat{\vtheta})                                                                                                                                                                                                                                                      \\
     & = - \frac{1}{2} t^\top \left \{ \tfrac{1}{n} \sum_{i=1}^n \phi_i^{(2)} (\widehat{\vtheta}) \right \} t + \sqrt{n} t^\top \mD ^\top (\mS^\star)  ^{-1} \vu(\vtheta^\star) - \sqrt{n} (\widehat{\vtheta} - \vtheta^\star)^\top  \left \{ \tfrac{1}{n} \sum_{i=1}^n \phi_i^{(2)} (\widehat{\vtheta}) \right \} t + \widehat{R}_n(t)
  \end{align*}
  where $\widehat{R}_n(t)$ satisfies the stochastic equicontinuity condition
  \citep{andrews94empirical,chernozhukov03mcmc}.
\end{lemma}
\begin{proof}
  We begin by using Lemma~1 in \citet{yu24variational} to evaluate a LAN
  expansion for $\log \EL ( \widehat{\vtheta} + n^{-1/2} t )$ and
  $\log \EL ( \widehat{\vtheta})$ to obtain the difference
  \begin{align*}
     & \log \EL ( \widehat{\vtheta} + n^{-1/2} t ) - \log \EL ( \widehat{\vtheta})                                                                                                                                    \\
     & = - \frac{1}{2} t^\top \mV_{\vtheta^\star} t +  \sqrt{n} t^\top \mD^\top (\mS^\star) ^{-1} \vu(\vtheta^\star) - \sqrt{n} (\widehat{\vtheta} - \vtheta^\star)^\top  \mV_{\vtheta^\star} t  + \overline{R}_n(t),
  \end{align*}
  where
  $\overline{R}_n(t) = R_n ( \sqrt{n} (\widehat{\vtheta} - \vtheta^\star) + t ) - R_n ( \sqrt{n} (\widehat{\vtheta} - \vtheta^\star) )$
  and $ R_n (t) = O_p ( (\lVert t \rVert + \lVert t \rVert^2)/\sqrt{n})$. From
  Step~\ref{app:item:bvm-1} of Theorem~\ref{app:thm:ep-bvm}, we have
  \begin{equation*}
    \left \lVert \frac{1}{n} \sum_{i=1}^n \phi_i^{(2)} (\widehat{\vtheta}) - \mV_{\vtheta^\star} \right \rVert = O_p(n^{-1/2})
  \end{equation*}
  Hence, $\widehat{R}$ may be expressed as
  \begin{equation*}
    \widehat{R}_n (t) = \overline{R}_n(t) + o_p(1)
  \end{equation*}
  To show that $\widehat{R}_n$ indeed satisfies the stochastic equicontinuity conditions, we need only to show that for any sequence $\delta_n \rightarrow 0$, we have
  \begin{equation*}
    \sup_{\lVert t \rVert \le \delta_n \sqrt{n}} \frac{ \lvert \widehat{R}_n(t) \rvert }{1 + \lVert t \rVert^2} = o_p(1)
  \end{equation*}
  We bound
  \begin{align*}
    \lvert \widehat{R}_n(t) \rvert & \le \lvert R_n ( \sqrt{n} (\widehat{\vtheta} - \vtheta^\star) + t ) \rvert + \lvert R_n ( \sqrt{n} (\widehat{\vtheta} - \vtheta^\star) ) \rvert
  \end{align*}
  It suffices to show that
  \begin{equation*}
    \sup_{\lVert t \rVert \le \delta_n \sqrt{n}}  \frac{ \lvert R_n ( \sqrt{n} (\widehat{\vtheta} - \vtheta^\star) + t ) \rvert }{1 + \lVert t \rVert^2} = o_p(1)
  \end{equation*}
  and
  \begin{equation*}
    \sup_{\lVert t \rVert \le \delta_n \sqrt{n}} \frac{ \lvert R_n ( \sqrt{n} (\widehat{\vtheta} - \vtheta^\star) ) \rvert }{1 + \lVert t \rVert^2} = o_p(1)
  \end{equation*}
  To show the first result for $\lvert R_n ( \sqrt{n} (\widehat{\vtheta} - \vtheta^\star) + t ) \rvert$, we note that due to the boundedness of $R_n$ in $\lVert t \rVert \le \sqrt{n} \vdelta$, there exists a sequence of random variables $C_n = O_p(1)$ such that
  \begin{equation*}
    R_n (\sqrt{n} (\widehat{\vtheta} - \vtheta^\star) + t) \le C_n n^{-1/2} \{  \sqrt{n} \lVert \widehat{\vtheta} - \vtheta^\star \rVert + n \lVert \widehat{\vtheta} - \vtheta^\star \rVert^2 \} = o_p(1)
  \end{equation*}
  where the last equality follows by the previously-shown result: $\lVert \widehat{\vtheta} - \vtheta^\star \rVert= O_p(n^{-1/2})$. To show the second result, we consider a sequence $\delta_n = \lVert \widehat{\vtheta} - \vtheta^\star \rVert$. Then, we have
  \begin{equation*}
    \sup_{\lVert t \rVert \le  \sqrt{n}  \lVert \widehat{\vtheta} - \vtheta^\star \rVert }  \frac{ \lvert  R_n ( t ) \rvert }{1 + \lVert t \rVert^2} = o_p(1).
  \end{equation*}
  and hence $\lvert R_n ( \sqrt{n} (\widehat{\vtheta} - \vtheta^\star) ) \rvert = o_p(1)$.
\end{proof}

\begin{lemma}
  \label{app:BvMNormalisingConstant}
  Assume conditions~\ref{app:item:theta-bound} to~\ref{app:item:D-full-rank} hold. Let $\sL (t) = \log \EL (\widehat{\sJ}_n + n^{-1/2} t ) - \log \EL (\widehat{\vtheta}) - \frac{n}{2} \vc_n^\top  \left \{ \tfrac{1}{n} \sum_{i=1}^n \phi_i^{(2)} (\widehat{\vtheta}) \right \} \vc_n$, where
  $$\vc_n =  \left \{ \tfrac{1}{n} \sum_{i=1}^n \phi_i^{(2)} (\widehat{\vtheta}) \right \}^{-1} \mD ^\top (\mS^\star)  ^{-1} \vu(\vtheta^\star) - ( \widehat{ \vtheta} - \vtheta^\star)$$
  and $\widehat{\sJ}_n = \widehat{\vtheta} + \vc_n$.
  Then
  \begin{equation*}
    \int \left \lvert p( \widehat{\sJ}_n + n^{-1/2} t) \exp \{  \sL(t) \} - p ( \widehat{\vtheta} ) \exp \left \{ - \tfrac{1}{2n} t^\top \sum_{i=1}^n \phi_i^{(2)} (\widehat{\vtheta}) t \right \} \right \rvert \; \mathrm{d}t = o_p(1).
  \end{equation*}
\end{lemma}
\begin{proof}
  Since $\widehat{R}_n (t)$ satisfies the stochastic equicontinuity conditions, for every $\epsilon > 0$, there exists a sufficiently small $\delta > 0$ and large $M > 0$ such that
  \begin{itemize}
    \item $\limsup_{n \rightarrow \infty} \bP^\star \left \{ \sup_{ M \le \lVert t \rVert \le \delta n^{1/2} } \frac{ \lvert \widehat{R}_n(t) \rvert }{ \lVert t \rVert^2 } > \epsilon  \right \} < \epsilon $,
    \item $\limsup_{n \rightarrow \infty} \bP^\star \left \{ \sup_{ \lVert t \rVert \le M } \lvert \widehat{R}_n(t) \rvert > \epsilon  \right \} = 0 $.
  \end{itemize}
  The proof proceeds by considering three integral regions : $\sA_{1n} = \{ t \, : \, \lVert t \rVert \le M \}$, $\sA_{2n} = \{ t \, : \, M < \lVert t \rVert \le \delta \sqrt{n} \}$, and $\sA_{3n} = \{ t \, : \,  \lVert t \rVert > \delta \sqrt{n} \}$. Throughout the proof, we use the result that $\lVert \sqrt{n} \vc_n \rVert = O_p(n^{-1/2})$. More rigorously,
  \begin{align*}
    \sqrt{n} \lVert \vc_n \rVert & \le  \lVert \sqrt{n} \vu(\vtheta^\star) \rVert \lVert n \psi^{(2)} (\widehat{\vtheta})^{-1} - \mV_{\vtheta^\star}^{-1} \rVert \lVert \mD \rVert \lVert (\mS^{\star} )^{-1} \rVert + n^{-1/2} \lVert \mV_{\vtheta^\star}^{-1} \rVert \lVert \mD \rVert \lVert (\mS^{\star} )^{-1} \rVert \sup_{\vtheta} \lVert p^{(1)} (\vtheta) \rVert / p(\widehat{\vtheta}) \\
                                 & = O_p(n^{-1/2}),
  \end{align*}
  where the convergence order follows from noting that $\lVert n \{ \sum_{i=1}^n \phi_i^{(2)} (\widehat{\vtheta}) \}^{-1} - \mV_{\vtheta^\star}^{-1} \rVert = O_p(n^{-1/2})$.  We begin by consider the integral over $\sA_{3n}$:
  \begin{align*}
     & \int_{\sA_{3n}} \left \lvert p( \widehat{\sJ}_n + n^{-1/2} t) \exp \{  \sL(t) \} - p( \widehat{\vtheta} ) \exp \left \{ - \tfrac{1}{2n} t^\top \sum_{i=1}^n \phi_i^{(2)} (\widehat{\vtheta}) t \right \} \right \rvert \; \mathrm{d}t          \\
     & \le \int_{\sA_{3n}} p( \widehat{\sJ}_n + n^{-1/2} t) \exp \{  \sL(t) \} \; \mathrm{d}t + \int_{\sA_{3n}} p ( \widehat{\vtheta} ) \exp \left \{ - \tfrac{1}{2n} t^\top \sum_{i=1}^n \phi_i^{(2)} (\widehat{\vtheta}) t \right \} \; \mathrm{d}t
  \end{align*}
  The second integral in the upper bound is
  \begin{align*}
     & \int_{\sA_{3n}} p ( \widehat{\vtheta} ) \exp \left \{ - \tfrac{1}{2n} t^\top \sum_{i=1}^n \phi_i^{(2)} (\widehat{\vtheta}) t \right \} \; \mathrm{d}t                                                                                                                                                                                  \\
     & \le \left \{ \sup_{\vtheta}  p ( \vtheta ) \right \} \int_{\sA_{3n}}  \exp \left \{ - \tfrac{1}{2n} t^\top \sum_{i=1}^n \phi_i^{(2)} (\widehat{\vtheta}) t \right \} \; \mathrm{d}t                                                                                                                                                    \\
     & \le \left \{ \sup_{\vtheta}  \pi ( \vtheta ) \right \} (2 \pi)^{p/2} \det \left (  \left \{ n^{-1} \sum_{i=1}^n \phi_i^{(2)} (\widehat{\vtheta}) \right \}^{-1/2} \right ) \bP_{t \sim \mathcal{N} (\vzero, \left \{  \tfrac{1}{n} \sum_{i=1}^n \phi_i^{(2)} (\widehat{\vtheta}) \right \}^{-1} )}(  \lVert t \rVert^2 > n \delta^2  ) \\
     & \le \left \{ \sup_{\vtheta}  \pi ( \vtheta ) \right \} (2 \pi)^{p/2} \det \left (  \left \{ n^{-1} \sum_{i=1}^n \phi_i^{(2)} (\widehat{\vtheta}) \right \}^{-1/2} \right ) \frac{\tr( \{ \tfrac{1}{n} \sum_{i=1}^n \phi_i^{(2)} (\widehat{\vtheta}) \}^{-1} )}{ n \delta^2 } = O_p(n^{-1}),
  \end{align*}
  where the last inequality corresponds to Markov's inequality and noting that $\det \left (  \left \{ n^{-1} \sum_{i=1}^n \phi_i^{(2)} (\widehat{\vtheta}) \right \}^{-1/2} \right ) \xrightarrow{p} \det \left (  \left \{ \mV_{\vtheta^\star} \right \}^{-1/2} \right )$ and  $\tr( \{ \tfrac{1}{n} \sum_{i=1}^n \phi_i^{(2)} (\widehat{\vtheta}) \}^{-1} ) \xrightarrow{p} \tr( \mV_{\vtheta^\star}^{-1} )$. To analyse the first integral, we first note that
  \begin{align*}
    \log \EL (\vtheta) - \log \EL (\widehat{\vtheta}) - \log p ( \widehat{\vtheta} )
     & \le  \log \EL (\vtheta) - \log \EL (\vtheta^\star) - \log p( \vtheta^\star )                                                                      \\
     & \le \sup_{ \lVert \vtheta - \vtheta^\star \rVert > \delta } \left \{ \EL (\vtheta) - \log \EL (\vtheta^\star) \right \} - \log p( \vtheta^\star ) \\
     & \le -n v_{\EL} - \log p( \vtheta^\star ) + o_p(1),
  \end{align*}
  where the first inequality follows from the definition of
  $\widehat{\vtheta}$ as the posterior mode. The third inequality follows
  from~\ref{app:item:el-bounded}. Hence, the first integral in the upper bound is
  \begin{align*}
    \int_{\sA_{3n}} p( \widehat{\sJ}_n + n^{-1/2} t) \exp \{  \sL(t) \} \; \mathrm{d}t & \le \left \{ \sup_{\vtheta}  p ( \vtheta ) \right \}  \int_{\sA_{3n}} \exp \{  \sL(t) \} \; \mathrm{d}t                                                           \\
                                                                                       & \le \left \{ \sup_{\vtheta}  p( \vtheta ) \right \} \int_{\sA_{3n}} \exp \{  -n v_{\EL} - \log p( \vtheta^\star ) + \log p( \widehat{\vtheta} ) \} \; \mathrm{d}t \\
                                                                                       & \le O(e^{-n v_{\EL}} ),
  \end{align*}
  where the second inequality follows from~\ref{app:item:theta-bound}
  and~\ref{app:item:prior-smoothness}. To examine the convergence of the integral
  in $A_{2n}$, we have
  \begin{align*}
     & \int_{\sA_{2n}} \left \lvert p( \widehat{\sJ}_n + n^{-1/2} t) \exp \{  \sL(t) \} - p ( \widehat{\vtheta} ) \exp \left \{ - \tfrac{1}{2n} t^\top \sum_{i=1}^n \phi_i^{(2)} (\widehat{\vtheta}) t \right \} \right \rvert \; \mathrm{d}t         \\
     & \le \int_{\sA_{2n}} p( \widehat{\sJ}_n + n^{-1/2} t) \exp \{  \sL(t) \} \; \mathrm{d}t + \int_{\sA_{2n}} p ( \widehat{\vtheta} ) \exp \left \{ - \tfrac{1}{2n} t^\top \sum_{i=1}^n \phi_i^{(2)} (\widehat{\vtheta}) t \right \} \; \mathrm{d}t
  \end{align*}
  For the second term in the upper bound:
  \begin{align*}
    \int_{\sA_{2n}} p ( \widehat{\vtheta} ) \exp \left \{ - \tfrac{1}{2n} t^\top \sum_{i=1}^n \phi_i^{(2)} (\widehat{\vtheta}) t \right \} \; \mathrm{d}t
     & \le p(\widehat{\vtheta}) \int_{\sA_{2n}}  \exp \left \{ - \frac{ \rho_{\min,n} }{2} \lVert t \rVert^2 \right \}                                    \\
     & \le p(\widehat{\vtheta}) \exp \left \{ - \tfrac{ 1 }{2} \rho_{\min,n} M^2 \right \} \int_{\sA_{2n}} \; \mathrm{d}t                                 \\
     & = \frac{\pi^{p/2}}{ \Gamma(p/2 + 1) } p(\widehat{\vtheta}) \exp \left \{ - \tfrac{ 1 }{2} \rho_{\min,n} M^2 \right \} \{ \delta^p n^{p/2} - M^p \} \\
     & \le C^\prime p(\widehat{\vtheta}) \exp \left \{ - \tfrac{ 1 }{2} \rho_{\min,n} M^2  + \tfrac{p}{2} \log(n) \right  \}
  \end{align*}
  By choosing
  $M = \sqrt{(p \{ \min\eigen(\mV_{\vtheta^\star} ) \}^{-1} + 1) \log(n)}$ and
  noting that
  $\rho_{\min,n} \xrightarrow{p} \min\eigen(\mV_{\vtheta^\star} ) > 0$, we
  have
  $\int_{\sA_{2n}} p ( \widehat{\vtheta} ) \exp \left \{ - \tfrac{1}{2n} t^\top \sum_{i=1}^n \phi_i^{(2)} (\widehat{\vtheta}) t \right \} \; \mathrm{d}t = O_p(n^{-1})$.
  For the first term in the upper bound:
  \begin{align*}
    \sL(t) \le - \frac{1}{2n} t^\top  \sum_{i=1}^n \phi_i^{(2)} (\widehat{\vtheta}) t + \lvert \widehat{R}_n (t+ n^{1/2} \vc_n) \rvert
  \end{align*}
  Following \citet{davidson94stochastic}, we have for a sufficiently small
  $\delta > 0$ and large $M$, there exists an $\epsilon > 0$ such that:
  $$\liminf_{n \rightarrow \infty} \bP^\star \left\{ \sup_{M \le \lVert t \rVert \le \delta \sqrt{n}} \frac{\widehat{R}_n(t + n^{1/2} \vc_n)}{\lVert t + n^{1/2} \vc_n \rVert^2} \le \tfrac{1}{4} \rho_{\min,n} \right\} \ge 1 - \epsilon,$$
  and hence
  \begin{align*}
    \lvert \widehat{R}_n (t+ n^{1/2} \vc_n) \rvert & \le \tfrac{\rho_{\min,n}}{4} \lVert t + n^{1/2} \vc_n \rVert^2 + o_p(1)
  \end{align*}
  Consequently, by writing $\mA_n = \tfrac{1}{n} \sum_{i=1}^n \phi_i^{(2)} (\widehat{\vtheta}) - \frac{\rho_{\min,n}}{2}
    \mI$ and $\ell_n = \{ M^2 - \lVert  \tfrac{1}{2} \sqrt{n} \rho_{\min,n} \mA_n^{-1} \vc_n \rVert^2 \}/ \lVert \mA^{-1/2} \rVert^2$, we have
  \begin{align*}
     & \int_{\sA_{2n}} p( \widehat{\sJ}_n + n^{-1/2} t) \exp \{  \sL(t) \} \; \mathrm{d}t                                                                                                                                                                               \\
     & \le \sup_{\vtheta \in \mTheta} p (\vtheta) \int_{\sA_{2n}}  \exp \left \{  - \tfrac{1}{2n} t^\top \sum_{i=1}^n \phi_i^{(2)} (\widehat{\vtheta}) t +  \lvert \widehat{R}_n (t+ n^{1/2} \vc_n) \rvert  \right \} \; \mathrm{d}t                                    \\
     & \le \sup_{\vtheta \in \mTheta} p (\vtheta) \int_{\sA_{2n}}  \exp \left \{  - \tfrac{1}{2n} t^\top \sum_{i=1}^n \phi_i^{(2)} (\widehat{\vtheta}) t +  \tfrac{\rho_{\min,n}}{4} \lVert t + n^{1/2} \vc_n \rVert^2 \right \} \; \mathrm{d}t                         \\
     & \le \sup_{\vtheta \in \mTheta} p (\vtheta) \int_{ \lVert t \rVert \ge M}  \exp \left \{  - \tfrac{1}{2n} t^\top \sum_{i=1}^n \phi_i^{(2)} (\widehat{\vtheta}) t +  \tfrac{\rho_{\min,n}}{4} \lVert t + n^{1/2} \vc_n \rVert^2 \right \} \; \mathrm{d}t           \\
     & \le \sup_{\vtheta \in \mTheta} p (\vtheta) (2 \pi)^{p/2} \det ( \mA_n^{-1/2} ) \exp \left \{ \frac{n \rho_{\min,n} \lVert \vc_n \rVert^2}{4} + \tfrac{1}{8} \rho_{\min,n}^2 n \vc_n^\top \mA_n \vc_n  \right \} \times \bP_{\chi_p^2} \{  \chi_p^2 \ge \ell_n \} \\
     & \le G_n \exp \left \{ - \frac{\ell_n}{2} + \frac{p}{2} (1 + \log(\ell_n/p))  \right \} = O_p(n^{-1}),
  \end{align*}
  where the first and second inequality follows from our bound on
  $\lvert \widehat{R}_n \rvert$, the third inequality follow by dropping upper
  limit of the integral range, the fourth inequality follows from triangle
  inequality on
  $\lVert \mA_n^{-1/2} \widetilde{t} + \tfrac{1}{2} \mA_n^{-1} \rho_{\min,n} n^{1/2} \vc_n \rVert$
  where $\widetilde{t} \sim \mathcal{N}(\vzero_p, \mI_p)$, the fifth inequality
  follows from upper tail bounds of $\chi_p^2$ as presented in Theorem 1 of
  \citet{ghosh21exponential} and noting that $G_n = O_p(1)$, and the last
  equality follows by noting that $-\ell_n/2$ is the dominant term in the
  exponent of the previous line and that $\ell_n = O_p(\log(n))$.

  Next, we examine the convergence of the integral over $\sA_{1n}$. We have
  \begin{align*}
     & \int_{\sA_{1n}} \left \lvert  p( \widehat{\sJ}_n + n^{-1/2} t) \exp \{  \sL(t) \} - p ( \widehat{\vtheta} ) \exp \left \{ - \tfrac{1}{2n} t^\top \sum_{i=1}^n \phi_i^{(2)} (\widehat{\vtheta}) t  \right \} \right \rvert \; \mathrm{d}t \le \widetilde{M}_n \int_{\sA_{1n}} \; \mathrm{d}t,
  \end{align*}
  where
  \begin{equation*}
    \widetilde{M}_n = \sup_{\lVert t \rVert \le M} \left \lvert  p( \widehat{\sJ}_n + n^{-1/2} t) \exp \{  \sL(t) \} - p ( \widehat{\vtheta} ) \exp \left \{ - \tfrac{1}{2n} t^\top \sum_{i=1}^n \phi_i^{(2)} (\widehat{\vtheta}) t  \right \} \right \rvert
  \end{equation*}
  Hence, we need only to examine the convergence rate of $\widetilde{M}_n$. We have
  \begin{align*}
    \widetilde{M}_n & \le \sup_{\lVert t \rVert \le M} \left \lvert p(\widehat{J}_n + t n^{-1/2}) \exp \left \{ \sL (t) \right \} - p(\widehat{J}_n + t n^{-1/2}) \exp \left \{ -\tfrac{1}{2} t^\top \left \{ \tfrac{1}{n}  \sum_{i=1}^n \phi_i^{(2)} (\widehat{\vtheta}) \right \} t \right \} \right \rvert \\
                    & + \sup_{\lVert t \rVert \le M} \left \lvert \left \{ p(\widehat{J}_n + t n^{-1/2}) - p ( \widehat{\vtheta} ) \right \} \exp \left \{ -\tfrac{1}{2} t^\top \left \{ \tfrac{1}{n}  \sum_{i=1}^n \phi_i^{(2)} (\widehat{\vtheta}) \right \} t \right \} \right \rvert                      \\
                    & \le \sup_{\vtheta \in \mTheta} p (\vtheta) \sup_{\lVert t \rVert \le M} \exp \left \{ - \tfrac{1}{2n} t^\top \sum_{i=1}^n \phi_i^{(2)} (\widehat{\vtheta}) t \right \}  \sup_{\lVert t \rVert \le M} \left \lvert \exp \{ \widehat{R}_n (t + n^{1/2} \vc_n) \} - 1 \right \rvert        \\
                    & + \sup_{\lVert t \rVert \le M} \exp \left \{ - \tfrac{1}{2n} t^\top \sum_{i=1}^n \phi_i^{(2)} (\widehat{\vtheta}) t \right \} \sup_{\lVert t \rVert \le M}  \left \lvert p(\widehat{J}_n + t n^{-1/2}) - p ( \widehat{\vtheta} ) \right \rvert                                          \\
                    & \le e^1 \sup_{\vtheta \in \mTheta} \pi (\vtheta) \sup_{\lVert t \rVert \le M} \left \lvert \widehat{R}_n (t + n^{1/2} \vc_n) \right \rvert + \sup_{\lVert t \rVert \le M}  \left \lvert p(\widehat{J}_n + t n^{-1/2}) - p ( \widehat{\vtheta} ) \right \rvert = o_p(1)
  \end{align*}
  where the first inequality is an application of triangle inequality, the second inequality follows by noting that $\sup_{\lVert t \rVert \le M} p(\widehat{J}_n + t n^{-1/2}) \le \sup_{\vtheta} p(\vtheta)$, the third equality follows by noting that $e^{\lvert \widehat{R} \rvert } - 1 \le e^1  \lvert \widehat{R} \rvert$ for all $\lvert \widehat{R} \rvert \le 1$, $\left \lvert \widehat{R}_n (t + n^{1/2} \vc_n) \right \rvert \xrightarrow{p} 0$, and that $ \sum_{i=1}^n \phi_i^{(2)} (\widehat{\vtheta})$ is positive definite, the fourth equality follows from the stochastic equicontinuity properties of $\widehat{R}$ and $\sup_{\lVert t \rVert \le M}  \left \lvert p(\widehat{J}_n + t n^{-1/2}) - p ( \widehat{\vtheta} ) \right \rvert \le \lVert\sup_{\vtheta} p^{(1)} (\vtheta) \rVert (M n^{-1/2} + \lVert \vc_n \rVert ) = O_p \left ( n^{-1/2} \sqrt{ \log(n)} \right )$ and $\left \lvert \widehat{R}_n (t + n^{1/2} \vc_n) \right \rvert = o_p(1)$.
\end{proof}

\begin{lemma}
  \label{app:thm:posterior-mode-bvm}
  Assume conditions~\ref{app:item:theta-bound} to~\ref{app:item:D-full-rank} hold. Then,
  \begin{equation*}
    \dtv \left \{ N \left ( \psi^{(2)} (\widehat{\vtheta}) \widehat{\vtheta} , \psi^{(2)} (\widehat{\vtheta})  \right ), p_{\EL} (\vtheta \mid \sD_{n}) \right \} = o_p(1).
  \end{equation*}
\end{lemma}
\begin{proof}
  Define $t = \sqrt{n}(\vtheta - \widehat{\sJ}_n)$ and $p_t ( t  \mid \sD_{n} ) = p_{\EL} ( \vtheta \mid \sD_{n} ) \lvert \det \left ( \tfrac{\partial \vtheta}{\partial t} \right ) \rvert$. Then, by a change-of-variables argument, we can express the total variation distance as
  \begin{align*}
     & \dtv \left \{ N \left ( \psi^{(2)} (\widehat{\vtheta}) \widehat{\vtheta} , \psi^{(2)} (\widehat{\vtheta})  \right ), p_{\EL} (\vtheta \mid \sD_{n}) \right \}                                                                                                                                                                                                   \\
     & = \frac{1}{2} \int \left \lvert K_n^{-1} p ( \widehat{\sJ}_n + n^{-1/2} t)\exp \{ \sL(t)  \} -   \frac{\det \left ( \frac{1}{n} \psi^{(2)} (\widehat{\vtheta}) \right )^{\tfrac{1}{2}}}{  (2\pi)^{p/2} } \exp \left \{  -\tfrac{1}{2n} (t + \sqrt{n} \vc_n)  ^\top \psi^{(2)} (\widehat{\vtheta}) (t + \sqrt{n} \vc_n)  \right \} \right \rvert \; \mathrm{d} t \\
     & \le \frac{1}{2} \int \left \lvert K_n^{-1} p ( \widehat{\sJ}_n + n^{-1/2} t)\exp \{ \sL(t)  \} -   \frac{\det \left ( \frac{1}{n} \psi^{(2)} (\widehat{\vtheta}) \right )^{\tfrac{1}{2}}}{  (2\pi)^{p/2} } \exp \left \{  -\tfrac{1}{2n} t   ^\top \psi^{(2)} (\widehat{\vtheta}) t  \right \} \right \rvert \; \mathrm{d} t                                    \\
     & +\dtv \left \{ N \left ( \vzero , \frac{1}{n} \psi^{(2)} (\widehat{\vtheta})  \right ), N \left ( - n^{-1/2} \psi^{(2)} (\widehat{\vtheta}) \vc_n  , \frac{1}{n} \psi^{(2)} (\widehat{\vtheta})  \right )  \right \},
  \end{align*}
  where an application of Lemma~\ref{app:BvMNormalisingConstant} allows us to express the normalizing constant $K_n$ as
  \begin{align*}
    K_n & = \int p ( \widehat{\sJ}_n + n^{-1/2} t)\exp \{ \sL(t)  \}  \; \mathrm{d}t                                                                                                                                                                                                                                                                                                         \\
        & \le \int  p (\widehat{\vtheta}) \exp \left \{  -\tfrac{1}{2n} t^\top \sum_{i=1}^n \phi_i^{(2)} ( \widehat{\vtheta} ) t  \right \} \; \mathrm{d} t + \int \left \lvert p ( \widehat{\sJ}_n + n^{-1/2} t)\exp \{ \sL(t)  \} - p (\widehat{\vtheta}) \exp \left \{  -\tfrac{1}{2n} t^\top \sum_{i=1}^n \phi_i^{(2)} ( \widehat{\vtheta} ) t  \right \}  \right \rvert \; \mathrm{d} t \\
        & = p (\widehat{\vtheta}) \left (2 \pi \right )^{p/2} \det \left ( \mV_{\vtheta^\star}  \right )^{-1/2} + o_p(1),
  \end{align*}
  where the last equality follows from an easily deducible result: $\lVert \tfrac{1}{n} \sum_{i=1}^n \phi_i^{(2)} ( \widehat{\vtheta}) - \mV_{\vtheta^\star}  \rVert = O_p(n^{-1/2})$. Hence,
  \begin{align*}
     & \frac{1}{2} \int \left \lvert K_n^{-1} p ( \widehat{\sJ}_n + n^{-1/2} t)\exp \{ \sL(t)  \} - \frac{\det \left ( \frac{1}{n} \psi^{(2)} (\widehat{\vtheta}) \right )^{\tfrac{1}{2}}}{  (2\pi)^{p/2} } \exp \left \{  -\tfrac{1}{2n} t   ^\top \psi^{(2)} (\widehat{\vtheta}) t  \right \} \right \rvert \; \mathrm{d} t \\
     & \le \frac{1}{2K_n} \int \left \lvert p ( \widehat{\sJ}_n + n^{-1/2} t)\exp \{ \sL(t)  \} -   p(\widehat{\vtheta}) \exp \left \{  -\tfrac{1}{2n} t   ^\top \psi^{(2)} (\widehat{\vtheta}) t  \right \} \right \rvert \; \mathrm{d} t + o_p(1)                                                                           \\
     & = o_p(1),
  \end{align*}
  where the second inequality follows by factorising our $K_n^{-1}$ and noting
  that
  $\lVert \tfrac{1}{n} \psi^{(2)} (\widehat{\vtheta}) - \mV_{\vtheta^\star} \rVert = O_p(n^{-1/2})$
  and the third result follows from Lemma~\ref{app:BvMNormalisingConstant}. It
  remains for us to show that
  \begin{equation}
    \label{app:eqnRemainsDTV}
    \dtv \left \{ \mathcal{N} \left ( \vzero , \frac{1}{n} \psi^{(2)} (\widehat{\vtheta})  \right ),
    \mathcal{N} \left ( - n^{-1/2} \psi^{(2)} (\widehat{\vtheta}) \vc_n,
    \frac{1}{n} \psi^{(2)} (\widehat{\vtheta})  \right )  \right \}
    = O_p(n^{-1/2}).
  \end{equation}
  Indeed, we begin by computing the KL-divergence between two Gaussians:
  \begin{align*}
     & \KL \left \{ \mathcal{N} \left ( \vzero , \frac{1}{n} \psi^{(2)} (\widehat{\vtheta})  \right ),
    \mathcal{N} \left ( - n^{-1/2} \psi^{(2)} (\widehat{\vtheta}) \vc_n ,
    \frac{1}{n} \psi^{(2)} (\widehat{\vtheta})  \right )  \right \}                                    \\
     & = n \vc_n ^\top \left \{  \frac{\psi^{(2)} (\widehat{\vtheta}) }{n }  \right \} \vc_n
  \end{align*}
  Hence,
  \begin{equation*}
    \KL \left \{ \mathcal{N} \left ( \vzero , \frac{1}{n} \psi^{(2)} (\widehat{\vtheta})  \right ),
    \mathcal{N} \left ( - n^{-1/2} \psi^{(2)} (\widehat{\vtheta}) \vc_n  ,
    \frac{1}{n} \psi^{(2)} (\widehat{\vtheta})  \right )  \right \} = O_p(n^{-1}).
  \end{equation*}
  Hence, by the Pinkser's inequality, (\ref{app:eqnRemainsDTV}) holds.
\end{proof}

\begin{lemma}
  \label{app:ConvergenceELmoments}
  Assume~\ref{app:item:theta-bound} to~\ref{app:item:D-full-rank} holds. Let
  $g: \sZ \times \mTheta \rightarrow \bR^q$ denote a function such that
  $\bE \{ \lVert g(\vz_1, \vtheta^\star) \rVert^2 \} < \infty$,
  $\lVert \bE \{ g(\vz_1, \vtheta^\star) h( \vz_1, \vtheta^\star )^\top \} \rVert < \infty$,
  and
  $\bE \{ \lVert g(\vz_1, \vtheta^\star) \rVert \lVert h( \vz_1, \vtheta^\star ) \rVert^2 \} < \infty$.
  Then
  \begin{equation*}
    \left \lVert \sum_{i=1}^n w_i(\vtheta^\star) g(\vz_i, \vtheta^\star) \right \rVert = \left \lVert \frac{1}{n} \sum_{i=1}^n g(\vz_i, \vtheta^\star) \right \rVert + O_p(n^{-1/2}).
  \end{equation*}
\end{lemma}
\begin{proof}
  Let $\gamma_i = \vlambda(\vtheta^\star)^\top h(\vz_i, \vtheta^\star)$. Similar to \citet{owen90empirical} eqn(2.16) \footnote{The cited equation has a typo. The denominator of $\gamma_i^2$ should be $1 + \gamma_i$ and not $1 - \gamma_i$}, we may write
  \begin{align*}
    \sum_{i=1}^n w_i(\vtheta^\star) g(\vz_i, \vtheta^\star)
     & = \frac{1}{n} \sum_{i=1}^n g(\vz_i, \vtheta^\star) \{ 1 - \gamma_i + \gamma_i^2 / (1 + \gamma_i) \}                                                                                                                                        \\
     & = \frac{1}{n} \sum_{i=1}^n g(\vz_i, \vtheta^\star) - \frac{1}{n} \sum_{i=1}^n g(\vz_i, \vtheta^\star) h(\vz_i, \vtheta^\star)^\top \vlambda (\vtheta^\star) + \frac{1}{n} \sum_{i=1}^n g(\vz_i, \vtheta^\star) \gamma_i^2 / (1 + \gamma_i)
  \end{align*}
  and consequently
  \begin{align*}
     & \lVert \sum_{i=1}^n w_i(\vtheta^\star) g(\vz_i, \vtheta^\star) \rVert                                                                                                                                                                                                                                                                                                           \\
     & \le \lVert  \frac{1}{n} \sum_{i=1}^n g(\vz_i, \vtheta^\star) \rVert + \left \lVert \frac{1}{n} \sum_{i=1}^n  g(\vz_i, \vtheta^\star) h(\vz_i, \vtheta^\star) \right \rVert \lVert \vlambda (\vtheta^\star) \rVert + \frac{\lVert \vlambda (\vtheta^\star)  \rVert^2}{n} \sum_{i=1}^n \lVert h(\vz_i, \vtheta^\star) \rVert^2  \lVert g(\vz_i, \vtheta^\star) \rVert / \lvert 1 + \gamma_i \rvert
  \end{align*}
  From \citet{owen90empirical} eqn (2.17), $\lVert \vlambda( \vtheta^\star) \rVert = O_p(n^{-1/2})$. Also, since
  $\lVert \bE \{  g(\vz_1, \vtheta^\star)  h( \vz_1, \vtheta^\star )^\top  \}  \rVert  < \infty$,
  therefore $  \lVert  \frac{1}{n} \sum_{i=1}^n g(\vz_i, \vtheta^\star)  h(\vz_i, \vtheta^\star)^\top \rVert \lVert \vlambda (\vtheta^\star) \rVert = O_p(n^{-1/2})$. Moreover,
  \begin{equation*}
    \frac{1}{n} \sum_{i=1}^n \lVert h(\vz_i, \vtheta^\star) \rVert^2  \lVert g(\vz_i, \vtheta^\star) \rVert / \lvert 1 + \gamma_i \rvert \le \frac{1}{n \min_{1 \le i \le n} \lvert 1 + \gamma_i \rvert } \sum_{i=1}^n \lVert h(\vz_i, \vtheta^\star) \rVert^2  \lVert g(\vz_i, \vtheta^\star) \rVert
  \end{equation*}
  Since $\bE \{ \lVert g(\vz_1, \vtheta^\star) \rVert \lVert h( \vz_1, \vtheta^\star ) \rVert^2 \}  < \infty$, we have
  \begin{equation*}
    \frac{1}{n} \sum_{i=1}^n \lVert h(\vz_i, \vtheta^\star) \rVert^2  \lVert g(\vz_i, \vtheta^\star) \rVert  = O_p(1).
  \end{equation*}
  It remains for us to show that  $\min_{1 \le i \le n} \lvert 1 + \gamma_i \rvert$ is bounded away from $0$. Indeed,
  \begin{align*}
    \min_{1 \le i \le n} \lvert 1 + \gamma_i \rvert
     & = \sqrt{ \min_{1 \le i \le n} ( 1 + \gamma_i )^2 }                \\
     & = \sqrt{ \min_{1 \le i \le n} ( 1 + \gamma_i^2 + 2 \gamma_i ) }   \\
     & \ge \sqrt{ 1 +\min \gamma_i^2 - 2 \max \lvert \gamma_i \rvert  }.
  \end{align*}
  By similar arguments to \citet{owen90empirical} eqn (2.15), $\min  \lvert \gamma_i \lvert \le \max  \lvert \gamma_i  = o_p(1)$. Hence,
  \begin{equation*}
    \frac{\lVert \vlambda \rVert^2}{n} \sum_{i=1}^n \lVert h(\vz_i, \vtheta^\star) \rVert^2  \lVert g(\vz_i, \vtheta^\star) \rVert / \lvert 1 + \gamma_i \rvert = O_p(n^{-1})
  \end{equation*}
  and we have our required result.
\end{proof}

\section{Hyperparameters of the algorithms in the experiments}
\label{app:sec:hyperparams}

\paragraph{HMC.} We use the HMC implementation in \texttt{blackjax}
\citep{cabezas24blackjax}. This implementation uses the minimal-norm integrator
\citep{blanes14numerical}, which we found to be more stable than the standard
leapfrog integrator in \citet{kien24elhmc}. In particular, the latter is prone
to divergence in the example with generalized estimating equations. We set the
step size to 0.01 and the number of integration steps to 50. The mass matrix is
set to identity.

\paragraph{Metropolis–Hastings.} We use a Gaussian random-walk proposal. The
covariance of the proposal is set to an estimate from a preliminary run of 10000
samples. We also multiply this covariance estimate by a `shrinkage factor', as
suggested in \citet{yang12bayesian}. This shrinkage factor is 0.7 for all
examples, except the examples of generalized estimating equations (0.3) and
high-dimensional linear regression (0.5).

\paragraph{Expectation-propagation.} We factorize the posterior into
$\nsites = 4$ sites for the Kyphosis dataset and $\nsites = 6$ sites for the
rest. We include the prior as a separate site, and distribute the likelihood
equally among the remaining sites. The choice of $\nsites$ ensures that the
sites have equal numbers of data points. In computing the tilted distribution,
we use Laplace's approximation for the first 50 iterations before switching to
importance sampling with 5000 samples for the rest. We set the damping factor to
$\alpha = 0.1$.

\paragraph{Variational Bayes.} We use $\log(\ndata) / 2$ for the adjustment
factor of the adjusted empirical likelihood, as suggested in
\citet{chen08adjusted}. The gradient in stochastic variational Bayes is computed
with the pathwise gradient estimator (also known as the reparameterization
trick), drawing one sample for the gradient estimator each time. The evidence
lower bound is maximized with the Adam optimizer \citep{kingma15adam} with the
learning rate set to $10^{-3}$ and the remaining optimizer hyperparameters set
to the default values of the \texttt{optax} implementation
\citep{deepmind20deepmind}.

\section{Additional experiments}
\label{app:sec:additional-experiments}

This section reports additional experiments that complement the comparisons in
the main text.

\subsection{Linear regression}
We show the contrast in performance between large- and small-$p$ settings with
linear regression. The data were generated from
$y_{i} = \vx_{i}^{\top} \vtheta_{0} + \epsilon_{i}, i = 1, \ldots, 100$, where
$\vx_{i} = (1, x_{i1}, \ldots, x_{i;(\dparam-1)})^{\top}$ and
$\epsilon_{i} \sim \mathcal{N}(0, 1)$. The covariates were drawn from the
standard normal distribution, $x_{ij} \sim \mathcal{N}(0, 1)$ for all
$j = 1, \ldots, \dparam - 1$. We consider two examples: $\dparam=2$,
$\vtheta_{0} = (0.5, 1)^{\top}$ and $\dparam=10$,
$\vtheta_{0} = (0.5, 1, 0.5, -1, 0.5, 0, \ldots, 0)^{\top}$. The vector
of regression coefficients, $\vtheta$, is our parameter of interest. In Bayesian
empirical likelihood, we avoid specifying the distribution of $\epsilon$ as in
the usual parametric approach. Rather, it is common to assume the data satisfy a
weaker orthogonality constraint,
$\bE_{(y, \vx) \sim F_0}[\vx(y - \vx^{\top} \vtheta)] = \vzero_{\dparam}$. The
constraint function is therefore
$\cons(\vz, \vtheta) = \vx(y - \vx^{\top} \vtheta)$, where $\vz = (y, \vx)$.

The results are presented in Figure~\ref{app:fig:nbp-median-additional-experiments}
and Tables~\ref{app:tab:nbp-threshold-additional}
and~\ref{app:tab:nbp-threshold-additional-skew}. In the $\dparam=2$ example
(Figure~\ref{app:fig:nbp-median-regression}), we observe that EPEL, the Laplace
approximation, and variational Bayes are all similar to the gold standard and
are consistent with the asymptotic Gaussian behaviour in
Theorem~\ref{app:thm:ep-bvm} and \citet{yu24variational}. EPEL and variational Bayes
attain the threshold immediately, while HMC and Metropolis--Hastings attain it
after 33.4 and 71.0 seconds, respectively. The same cannot be said in the larger
$\dparam$ setting, where the ratio $\ndata / \dparam$ is much smaller.
For variational Bayes, the added empirical-likelihood adjustment leads to
a visibly poorer approximation to the gold-standard posterior.
Similarly, the Laplace approximation is not close to the
posterior. However, EPEL can still attain sufficient accuracy in this setup,
reaching the threshold after 83.6 seconds. Sampling-based HMC and
Metropolis--Hastings can generally produce a good approximation to the posterior
when given enough compute budget, attaining the threshold after 153.3 and 376.4
seconds, respectively.

\begin{table}[t]
  \centering
  \begin{tabular}{lrrrr}
    \toprule
    Setup & EPEL & HMC & MH & VB \\
    \midrule
    Linear regression, $\dparam=2$ & 0.0 & 33.4 & 71.0 & 0.0 \\
    Linear regression, $\dparam=10$ & 83.6 & 153.3 & 376.4 & -- \\
    Logistic regression (\texttt{kyphosis}) & 14.6 & 43.8 & 87.6 & 7.3 \\
    \bottomrule
  \end{tabular}
  \caption{Time (in seconds) to reach sufficient approximation quality,
    defined as the first recorded time at which the 50 replicate NBP statistics
    are significantly above 474 by a one-sided Wilcoxon signed-rank test at the
    5\% level. Missing entries indicate that the method did not attain this
    threshold within the recorded time grid.}
  \label{app:tab:nbp-threshold-additional}
\end{table}

\begin{figure*}[t]
  \centering
  \begin{subfigure}{0.49\textwidth}
    \includegraphics[width=\linewidth]{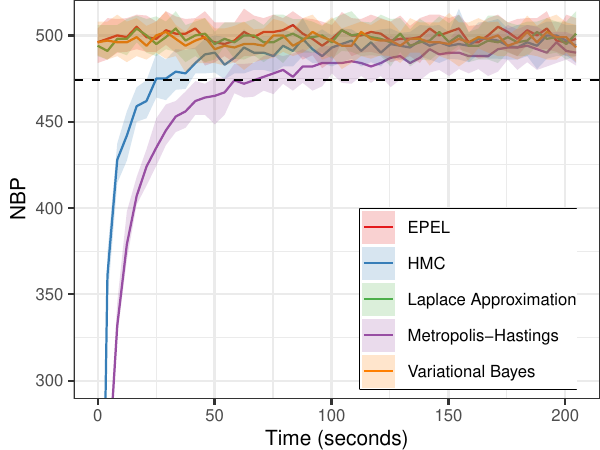}
    \caption{Linear regression, $\dparam = 2$}
    \label{app:fig:nbp-median-regression}
  \end{subfigure}%
  \begin{subfigure}{0.49\textwidth}
    \includegraphics[width=\linewidth]{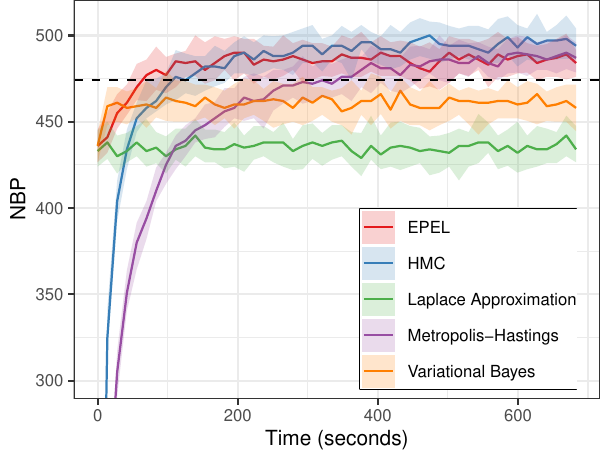}
    \caption{Linear regression, $\dparam = 10$}
    \label{app:fig:nbp-median-regression10}
  \end{subfigure}
  \begin{subfigure}{0.49\textwidth}
    \includegraphics[width=\linewidth]{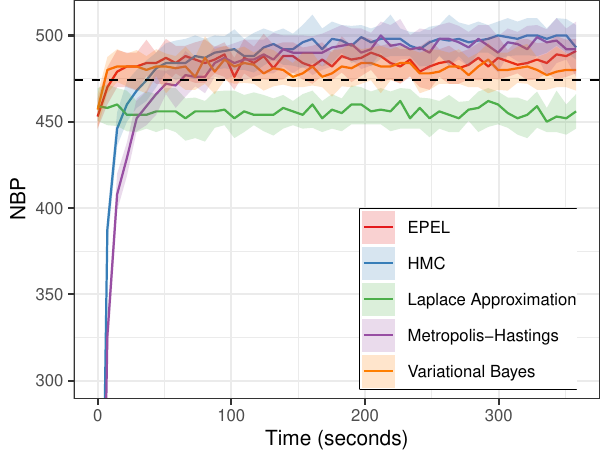}
    \caption{Logistic regression with the \texttt{kyphosis} data}
    \label{app:fig:nbp-median-kyphosis}
  \end{subfigure}%
  \begin{subfigure}{0.49\textwidth}
    \includegraphics[width=\linewidth]{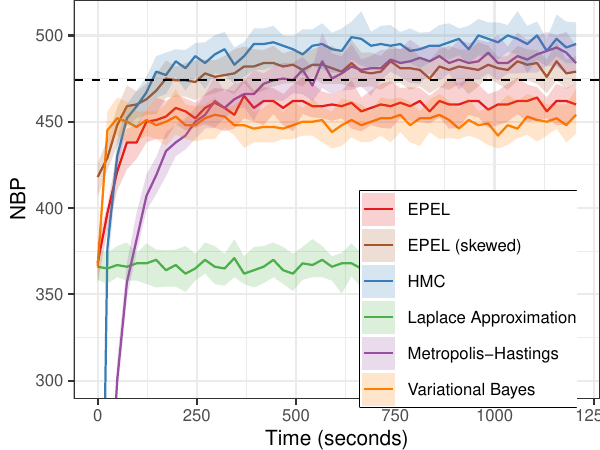}
    \caption{Logistic regression with the \texttt{breastfeed} data}
    \label{app:fig:nbp-median-breastfeed}
  \end{subfigure}%
  \caption{NBP statistics tracked over computation time for the additional
    experimental setups. Coloured curves show the median NBP with respect to the
    gold standard over 50 repetitions. Shaded bands denote the 0.25--0.75
    quantiles. The dotted horizontal line marks the accuracy threshold of 474.}
  \label{app:fig:nbp-median-additional-experiments}
\end{figure*}

\subsection{Logistic regression with the \texttt{kyphosis} data}
We demonstrate the methods on the \texttt{kyphosis} dataset
\citep{chambers92statistical} with a logistic regression model. The dataset
contains the outcomes of 81 children after corrective spinal surgery. The
response variable $y$ is binary (absence or presence of kyphosis after surgery)
and the covariates $\vx$ are the age of the child, the number of vertebrae
involved in the surgery, and the number of the topmost vertebra operated on (in
addition to an intercept). The dataset was obtained from the \texttt{rpart} R
package. The covariates were standardized before model fitting. We have an
orthogonality constraint
$\cons(y, \vx, \vtheta) = \vx(y - \expit(\vx^{\top} \vtheta))$ where $\vtheta$
is the coefficient vector.

The results of this real-world example are similar to those of the quantile
regression example. Variational Bayes and EPEL attain the NBP threshold after
7.3 and 14.6 seconds, respectively, compared with 43.8 seconds for HMC and
87.6 seconds for Metropolis--Hastings
(Table~\ref{app:tab:nbp-threshold-additional}). Both Metropolis--Hastings and HMC
eventually overtake the Gaussian-based approximations, but only after
substantially more computation. The Laplace approximation performed poorly in
this example.

\subsection{Logistic regression with the \texttt{breastfeed} data}
We show an additional example with a skewed posterior. The \texttt{breastfeed}
data \citep{heritier09robust} record breastfeeding decisions in a UK hospital
survey of expectant mothers and are available as the \texttt{breastfeed} dataset
in the \texttt{mpath} R package \citep{wang24mpath}. After removing records with
missing age or education (135 observations remain), we use a binary response
indicating whether the baby was breastfed. The included covariates are age,
education, pregnancy stage, current feeding intention, future feeding intention,
partner status, current smoking, smoking before birth, and ethnicity. For
logistic regression, we have an orthogonality constraint
$\cons(y, \vx, \vtheta) = \vx(y - \expit(\vx^{\top} \vtheta))$. All continuous
covariates are standardized.

In this example, the skew-corrected EPEL approximation attains the NBP
threshold at 319.8 seconds, HMC attains it at 172.2 seconds, and
Metropolis--Hastings attains it at 516.6 seconds
(Table~\ref{app:tab:nbp-threshold-additional-skew}). The uncorrected EPEL approximation
and variational Bayes do not attain the threshold within the recorded time grid.

\begin{table}[t]
  \centering
  \begin{tabular}{lrrrrr}
    \toprule
    Setup & EPEL & \shortstack[c]{EPEL\\(post-process)} & HMC & MH & VB \\
    \midrule
    Logistic regression (\texttt{breastfeed}) & -- & 319.8 & 172.2 & 516.6 & -- \\
    \bottomrule
  \end{tabular}
  \caption{Time (in seconds) to reach sufficient approximation quality, defined
    as the first recorded time at which the 50 replicate NBP statistics are
    significantly above 474 by a one-sided Wilcoxon signed-rank test at the 5\%
    level. The EPEL (post-process) column reports EPEL after the skewness
    correction. Missing entries indicate that the method did not attain this
    threshold within the recorded time grid.}
  \label{app:tab:nbp-threshold-additional-skew}
\end{table}

\end{document}